%% file: undid.tex
\documentclass[12pt]{article}
\usepackage{setspace} 
\usepackage[utf8]{inputenc}
\usepackage[english]{babel}
\usepackage[dvipsnames]{xcolor}
\usepackage[]{appendix}
\usepackage{booktabs}
\usepackage{tabularx}
\usepackage{caption}
\usepackage{array}

\usepackage{pgfplots}
\usepgfplotslibrary{groupplots}
\pgfplotsset{compat=1.18}

\usepackage{tikz}
\usetikzlibrary{arrows.meta, positioning, shapes.geometric}

\tikzstyle{startstop} = [rectangle, rounded corners, minimum width=3.5cm, minimum height=1cm, text centered, draw=black, fill=blue!10]
\tikzstyle{decision} = [diamond, aspect=2, text centered, draw=black, fill=green!20, inner sep=1pt]
\tikzstyle{process} = [rectangle, minimum width=3.5cm, minimum height=1cm, text centered, draw=black, fill=blue!20]
\tikzstyle{arrow} = [thick,->,>=stealth]

\usepackage{subcaption}
\usepackage{graphicx}
\usetikzlibrary{arrows.meta}

\makeatletter
  \newcommand*\l@authors{\@dottedtocline{1}{0pt}{0pt}}
\makeatother

\usepackage[letterpaper,top=2cm,bottom=2cm,left=2cm,right=2cm,marginparwidth=2cm]{geometry}

\usepackage[parfill]{parskip}    		

\usepackage{graphicx}
\usepackage[colorlinks=true, allcolors=blue]{hyperref}
\usepackage{pdflscape}
\usepackage{rotating}
\usepackage{multirow}
\usepackage{caption}
\usepackage{subcaption} 
\usepackage[flushleft]{threeparttable}
\usepackage{multicol}
\usepackage{wrapfig}
\renewcommand{\arraystretch}{1.5}
\usepackage{amsmath}
\usepackage{amssymb}
\usepackage{bbm}

\usepackage{amsthm}

\newtheorem{theorem}{Theorem}
\theoremstyle{definition}

\newtheorem{lemma}{Lemma}
\theoremstyle{definition}

\newtheorem{assumption}{Assumption}
\theoremstyle{definition}

\theoremstyle{definition}

\newtheorem{corollary}{Corollary}
\theoremstyle{definition}

\usepackage{enumerate}

\usepackage{comment}

\usepackage{pifont}
\usepackage{placeins}

\newtheorem{dataassumption}{Data Assumption}

\usepackage[]{natbib}
\usepackage[style=british]{csquotes}
\setcitestyle{authoryear,open={(},close={)}}
\title{Difference-in-Differences with Unpoolable Data
\thanks{We are grateful to the Canadian Institutes of Health Research (CIHR) for funding this project: grant number PJT-175079.  Thanks also to Jamie Daw, Kim McGrail, James MacKinnon, Douglas Miller, and David Rudoler for many helpful discussions, to Nick Brown, Austin Denteh, Michel Grignon, and Rapha{\"e}l Langevin for detailed comments,  to Kwasi Tabiri for research assistance, and to Eric Jamieson for coding the software packages.  We appreciate comments from audience members at the Canadian Society for Epidemiology and Biostatistics 2023 conference, Canadian Stata Users Group 2023 Meeting,  the Canadian Econometrics Study Group 2023 conference, the Southern Economic Association 2023 Conference, Carleton University, Indiana University, the Canadian Association for Health Services and Policy Research 2024 conference, the Canadian Health Economics Study Group 2024 meeting, the Canadian Economics Association 2024 annual meetings, the Society for Epidemiological Research 2024 conference, American Society for Health Economics 2024 annual meeting, and the International Association of Applied Econometrics 2024 conference.}}
\author{Sunny Karim \and Matthew D. Webb \and Nichole Austin \and Erin Strumpf\thanks{Karim: Carleton University, \url{sunny.karim@cmail.carleton.ca}.  Webb: Carleton University, \url{matt.webb@carleton.ca}. Austin: Dalhousie University, \url{nichole.austin@dal.ca}. Strumpf: McGill University, \url{erin.strumpf@mcgill.ca.}}}
\date{}

\begin{document}
\maketitle

\begin{abstract}

    Difference-in-differences (DID) is widely used to estimate treatment effects, but standard pooled implementations may be unavailable when privacy, legal, or data-governance restrictions prevent treated and comparison data from being analyzed together. We develop UN-DID, a regression-based approach for DID with siloed data that accommodates covariates, multiple groups, and staggered adoption. Without covariates, UN-DID reproduces DID estimates and standard errors. With covariates, conventional TWFE can be biased when covariate effects differ across silos, whereas UN-DID allows silo-specific covariate effects. UN-DID allows the estimation of cross-jurisdictional treatment effects with unpoolable data, expanding the set of research questions that can be addressed.
\end{abstract}

\section{Introduction}
\label{c1intro}

Difference-in-differences (DID) is a widely used method to evaluate the impacts of policies and interventions. It is used to estimate both intended treatment effects and unintended consequences. The source of identifying variation in DID analyses is often a policy change that occurs in some political jurisdictions (e.g., states, countries) but not in others. Treated observations come from jurisdictions or units in which the intervention occurs, while control or untreated observations come from jurisdictions that do not experience that change (at least not at the same time). DID can be used to estimate both the intent-to-treat (ITT) effect and the average treatment effect on the treated (ATT). In the absence of randomized treatment assignment, a widely documented set of assumptions supports the unbiased estimation of these effects using conventional DID methods such as the Two-way Fixed Effects (TWFE) estimator. These assumptions include parallel pre-intervention trends, no anticipation and no staggered adoption, and new DID estimators have been developed in settings where these assumptions do not hold \citep{roth2022s, abadie2005semiparametric, de2020twott, callaway2021difference, de2020twoseveral}. Since treatment assignment occurs at the jurisdictional level, using conventional DID estimation relies on an additional assumption that has not yet been made explicit: that data from the treatment group and the control group can be combined, or pooled, for analyses.

This paper develops the DID for unpoolable data (UN-DID) method, enabling the estimation of treatment effects using DID when the poolable data assumption is violated. This occurs when data come from the same jurisdictional level as the policy variation, and may not be combined on the same computer or server, usually due to privacy concerns or legal restrictions on data sharing. Usually researchers accessing data stored in secured servers, often called silos, sign an agreement not to import data to or export data from these servers. The data are therefore siloed at the level of treatment, and standard pooled DID implementations are unavailable under the data-access restrictions we consider. Examples include health care  data from different insurers (e.g., Canadian provinces \citep{HDRN}, private and public insurers, all-payer claims databases across states \citep{aspe2023}), restricted-use micro-data held in secure environments (e.g., vital statistics, income tax, business data held in Statistics Canada's or the US Census Bureau's Research Data Centres \citep{CanRDC, USRDC}), and electronic health records federated across countries (e.g., the European Health Data Space \citep{EHDS}).

Sometimes datasets from different jurisdictions also differ in how variables are measured, coded, or categorized. UN-DID does not require covariates to have identical representations across silos. Rather, the locally available measures must capture sufficiently comparable underlying concepts to support the intended conditional comparison. For example, one jurisdiction may record income continuously while another records it in categories, or jurisdictions may use different measures of educational attainment. Because covariate adjustment is performed separately within each silo, UN-DID can accommodate such differences without first forcing the variables into a common representation. UN-DID therefore facilitates cross-jurisdictional policy evaluation when treatment varies at the level of the data silos and both data access and variable codings differ across jurisdictions.

When the ability to pool data across treatment and control jurisdictions is limited, the choice of appropriate counterfactuals is necessarily constrained. Researchers may resort to estimating first-differences \citep{Fischer2021} or using interrupted time series \citep{Tamblyn2001, bernal2017interrupted, Farin2024}, which only provide unbiased treatment effect estimates if no other changes coincide with the intervention of interest. More convincing treatment effect estimates are obtained leveraging within-jurisdiction controls \citep{Rudoler2023, Strumpf2017, LiHurley2014, GIRMA2007} or individual fixed effects \citep{Dumont2008}, both of which may still suffer from selection on unobservables. The impacts of national policies that take effect everywhere at the same time can sometimes be identified using pre-policy variation \citep{Finkelstein2007, Duflo2001}, but in other cases, including nationwide cannabis legalization in Canada in 2018, rigorous policy evaluations using external controls remain extremely limited or non-existent. 

If policy designs permit the use of Regression Discontinuity (RD) or Instrumental Variables (IVs) to estimate local average treatment effects \citep{almond2010, Hutcheon2020}, the identifying variation may be convincing, but the generalizability of the estimates is narrow by definition. Meta-analysis can be used with unpoolable data, and is an effective way to leverage sample size by combining silo-specific treatment effect estimates into a single parameter \citep{xiong2023federated}. For example, the Canadian Network for Observational Drug Effect Studies (CNODES) routinely employs meta-analytical techniques to more comprehensively examine drug safety using siloed provincial or national data \citep{Suissa2012, Filion2014, Durand2021}. However, meta-analysis cannot be used to estimate treatment effects of an intervention in one jurisdiction using another unexposed jurisdiction as a control. UN-DID will allow researchers to estimate treatment effects using better counterfactuals, leveraging data from non-treated jurisdictions and limiting selection bias while maintaining generalizability.

Data poolability limitations can also constrain which research questions get asked, specifically which outcomes are examined in evaluations of policy interventions. National and international surveys include multiple states/provinces or countries, so DID models are easily estimated with treated and control units in the same dataset \citep{Baker2008, Charters2013, Dube2019}. The range of outcomes available in a survey, however, is necessarily limited with respect to content, detail, and method of reporting (self-report vs. administrative data). In the absence of a method like UN-DID to allow treatment effect estimation across separate datasets, the questions that have been asked and answered have been a function of the (pooled) datasets that are available. This constraint has become even more problematic in recent years, precisely because large, detailed, and high-quality data are increasingly available for research. All-payer health insurance claims, electronic medical records, detailed individual- and firm-level tax filings, genetic data, and immigration records offer comprehensiveness, accuracy and detail that allow researchers to ask and answer new questions. However, these data are generally accessible only in secure research data centers and are not shareable or poolable across jurisdictions. By enabling treatment effect estimation using siloed data, UN-DID will open up opportunities to answer novel research questions and exploit new sources of identifying variation using separate datasets across jurisdictions.

The UN-DID method allows the estimation of treatment effects with unpoolable data. This has important implications for research across the fields of applied economics, and across a range of disciplines including public policy, epidemiology, health research, environmental science, political science, sociology, and others. With UN-DID, new research questions can be answered and better counterfactuals can be used to reduce bias when estimating treatment effects. While researchers with treatment and control observations in separate datasets could already calculate the double-difference ``by hand", our formal results justify this approach and extend beyond the 2$\times$2 case, accommodating staggered adoption, more than two groups, and covariates. Furthermore, UN-DID enables the estimation of appropriate (clustered) standard errors and therefore proper inference. We have developed software packages in R, Stata, and Julia, which are described in Section \ref{sec:software}. In a separate publication, we will guide users of the packages through an empirical example using real-world siloed datasets. These will enable UN-DID to be implemented across secure data silos, even when local data analysts are not skilled econometricians.

The intuition behind UN-DID is based on non-parametric DID estimation where conditional mean outcomes are estimated for each of four treatment$\times$timing groups: pre-intervention treatment, pre-intervention control, post-intervention treatment and post-intervention control. The post-pre first difference can then be calculated for each treatment and control group, and then the DID estimate is simply the difference between those first differences. UN-DID estimates first-differences and standard errors within jurisdictional silos using a regression-based tool and provides a method to combine these estimates into a treatment effect estimate and to test hypotheses. We provide formal results which show that this method accurately estimates the ATT. UN-DID also estimates the parameters necessary to evaluate parallel pre-intervention trends. 

In this paper, we begin by reviewing the DID identifying assumptions and then introduce two data-access assumptions that characterize the estimation environment, including whether treated and comparison data can be analyzed together. We propose a method for estimating the ATT when the data-poolability assumption does not hold, which we call UN-DID (difference- in-differences) for unpoolable data). Because researchers commonly rely on regression-based tools for estimating treatment effects, we introduce the UN-DID method as a regression-based tool for this purpose. Using analytical proofs, we first demonstrate that the UN-DID can identify the ATT under the strong parallel trends and no-anticipation assumptions. Additionally, we demonstrate that without covariates, the conventional and UN-DID estimators yield algebraically equivalent estimates when applied to a poolable dataset. Next, we relax the strong parallel trends assumption to the conditional parallel trends assumption to incorporate covariates. Alongside this, we adopt the recently introduced state-invariant common causal covariates (CCC) assumption from \cite{karim2024good}. This assumption states that each covariate shares a common coefficient across different silos. In this setting, we show that the UN-DID identifies the ATT regardless of whether the state-invariant CCC holds or is violated. This is because the UN-DID conducts separate regressions within each silo, allowing them to capture how covariate effects on the outcome vary across silos. However, the conventional DID regression usually applies a common coefficient to covariates across silos, since the implications of the CCC assumption were not previously understood. As a result, if the state-invariant CCC assumption does not hold, the conventional DID model is mis-specified, potentially leading to a biased ATT estimate. In contrast, when the state-invariant CCC holds, the conventional DID regression remains unbiased. Despite this, the UN-DID and the conventional DID estimators are generally not algebraically equivalent, except in the special case where covariates are time-invariant and panel data are available. We also turn to the many-groups, many-periods (G$\times$T) setting and relax the no staggered adoption assumption to show how the UN-DID method works in this setting, in the spirit of \cite{callaway2021difference}.   Cluster robust inference is a challenge, but possible, as described in Section \ref{sec:cluster} and a companion project \citep*{karim2025slides}. 

With analytical results in hand, we conduct three sets of Monte Carlo experiments.  The first and third use completely simulated data, and the second uses the Current Population Survey.  The first set establishes that coefficient estimates from conventional DID and UN-DID coincide with time-invariant covariates (In Appendix \ref{sec:moremc}), and are similar with time-varying covariates when the state-invariant CCC assumption holds (in Section  \ref{sec:results}).  The second set shows that with staggered adoption, cluster-robust inference is quite reliable, and ATT estimates follow a similar distribution to CSDID.  The third (in Appendix \ref{sec:diff_covar}) shows that UN-DID can work well even when the covariates are not harmonized across silos.

This paper contributes to the dynamic literature on DID methods. In the past five years, an important series of papers have looked under the hood of the conventional DID estimator, exploring what is actually being estimated and what can go wrong when each of the identifying assumptions is violated \citep{roth2022s, abadie2005semiparametric, de2020twott, callaway2021difference, de2020twoseveral, goodman2021difference}. They have developed methods to identify the ATT in each case, changing the way that DID analyses are conducted. We highlight the heretofore implicit assumption of data poolability and provide a solution to estimate the ATT when this assumption does not hold. We assess UN-DID's performance and demonstrate its equivalence to conventional DID. 

We evidently aim to contribute to the large literature that uses DID to estimate treatment effects of policy changes and other interventions, both in economics and across a wide range of disciplines. In this paper, we apply UN-DID to two published policy evaluations that used conventional DID methods \citep{sabia2017effect, conley_2011}. Our primary goal is to compare the estimates from the two methods, with a secondary goal of highlighting how UN-DID could be used in empirical analyses. We have produced UN-DID software packages in several languages, which are explained in detail in Section \ref{sec:software}. A forthcoming user guide will use UN-DID to evaluate a policy change using real-world siloed data, where the standard pooled DID implementation is unavailable and no ATT estimates currently exist for comparison.

Proofs and supplementary analyses appear in the appendices. Section~\ref{sec:theory} lays out the DID framework and the data-access assumptions. Section~\ref{sec: UNDIDintroduction} introduces UN-DID in the two-silo, two-period case, including inference and silo differences in covariate representation. Section~\ref{sec:UNDID} adds covariates, the state-invariant common causal covariates assumption, and the link to DID-INT. Section~\ref{sec:mult} extends the method to staggered adoption and to multiple silos per treatment date, and describes the two-stage estimator and the accompanying software. Monte Carlo designs and results for time-varying covariates are in Section~\ref{section: MC}. Section~\ref{sec:mc_stag} reports a CPS placebo-law exercise comparing UN-DID with CSDID, including cluster jackknife and randomization inference. Section~\ref{sec:examples} presents two applications; Section~\ref{sec:concl} concludes.

\section{Theoretical Framework} \label{sec:theory}
\label{section: Theoretical Framework}

In this section, we introduce the core assumptions of DID, most of which also apply to UN-DID. Suppose an intervention was introduced at different times across different jurisdictions (e.g., states, provinces). Individual-level data are held by jurisdictions and are not poolable across them - in other words, the data cannot be easily combined for analysis. There are several possible variations of this common scenario: for example, implementation timing may be common or staggered, there may be multiple treatment and control groups (e.g., several states may have enforced a policy, while several others did not), and so on. UN-DID offers a novel approach to handle these analytical contexts.

We begin by describing the canonical case (two groups and two periods, with common adoption and covariates) before expanding to staggered adoption scenarios in Section \ref{sec:stag}. The basic idea of DID estimation is to compare the difference in outcomes before and after the treatment between groups that  received the treatment and groups that did not \citep{bertrand2004much}. 
The change in outcomes for the control group is used as the unobserved counterfactual of the treated group. \cite{card1993minimum} popularized the conventional DID to estimate the effect of an increase in minimum wage on unemployment in the American state of New Jersey using changes in outcomes over the same period in Pennsylvania as the control. The conventional DID estimate of the ATT uses the following regression:
\begin{equation}
\label{equation: didsimple}
    Y^{Pooled}_{i,s,t} = \beta_0 + \beta_1 D^{Pooled}_{s} + \beta_2 P^{Pooled}_{t} + \beta_3 P^{Pooled}_{t} * D^{Pooled}_{s} + \beta_4 X^{Pooled}_{i,s,t} + \epsilon^{Pooled}_{i,s,t}.
\end{equation}
Here, $D^{Pooled}_{s}$ is a dummy variable that takes on a value of 1 if group $s$ is in the treatment group, and 0 otherwise. $P^{Pooled}_{t}$ is a dummy variable that takes on a value of 1 if the observation is in the post-intervention period, and 0 otherwise. $X^{Pooled}_{i,s,t}$ are the covariates that researchers include to improve the plausibility of parallel trends. The superscript $Pooled$ implies that we are combining the dataset from both the treated group ($Treat$) and the control group ($Control$) into a common dataset ($Pooled = Treat + Control$). Under very strong assumptions, $\hat{\beta_3}$ (the coefficient of the interaction term between $D^{Pooled}_{s}$ and $P^{Pooled}_{t}$) identifies the ATT \citep{roth2022s}. To simplify notation, we will not use the $Pooled$ superscripts to denote a combined dataset for the rest of the paper. 

\begin{assumption}[Stable Unit Treatment Value Assumption]
\label{assumption: sutva}
For each unit $i$, region $s$, and period $t$, the potential outcomes
$Y_{i,s,t}(1)$ and $Y_{i,s,t}(0)$ are well defined and depend only on the
treatment status of unit $i$'s region. In particular, there is no interference
across units or regions, and the observed outcome satisfies
\begin{equation}
\label{equation: sutva}
Y_{i,s,t}
=
D_{s,t}Y_{i,s,t}(1)
+
\left(1-D_{s,t}\right)Y_{i,s,t}(0),
\end{equation}
where $D_{s,t} = (D_s * P_t) \in\{0,1\}$ denotes the treatment status of region $s$ in
period $t$.
\end{assumption}

\begin{assumption}[Treatment is Binary]
\label{assumption: binary}
    This implies  individual $i$ in state $s$ can be either treated or not treated at time $t$. There are no variations in treatment intensity.
\[ D_{s,t} = \begin{cases} 
      1 & \mbox{if individual i in state s is treated at time t}.\\
      0 & \mbox{if individual i in state s is not treated at time t}. \\
       \end{cases}
\]
\end{assumption}

\begin{assumption}[Conditional Parallel Trends]
\label{assumption: spt}
    The evolution of untreated potential outcomes conditional on covariates is the same between treated and control groups.
    \begin{align}
        \label{equation: cptundid}
        \begin{split}
           & \biggr[E[Y_{i,s,t}(0)|D_s = 1,P_t=1,X_{i,s,t}] - E[Y_{i,s,t}(0)|D_s = 1,P_t=0,X_{i,s,t}]\biggr] \\
              = & \ \biggr[E[Y_{i,s,t}(0)|D_s = 0,P_t=1,X_{i,s,t}] - E[Y_{i,s,t}(0)|D_s = 0, P_t= 0,X_{i,s,t}] \biggr].
        \end{split}
    \end{align}
\end{assumption}

In Equation \eqref{equation: cptundid}, $Y_{i,s,t}(0)$ is the untreated potential outcome for the relevant group and period. Assumption \eqref{assumption: spt} ensures that the selection bias is 0. A plausibility test for parallel trends is to analyze the evolution of outcomes conditional on covariates for both treated and control groups in the pre-intervention period and check if the trends in outcomes are parallel.

\begin{assumption}[No Anticipation] 
\label{assumption: na}
Treated units do not change behavior before treatment occurs.
\begin{equation}
    \label{equation: na}
    \begin{gathered}
        \biggl[E[Y_{i,s,t}(1)|D_{s} = 1,P_t = 0,X_{i,s,t}] = E[Y_{i,s,t}(0)|D_{s} = 1, P_{t} = 0, X_{i,s,t}]\biggr] \mbox{\quad\textit{a.s.} for all} \;\; t < t'.
    \end{gathered}
\end{equation}    
\end{assumption}

This assumption requires that, prior to treatment, the conditional mean of the treated potential outcome is equal to the conditional mean of the untreated potential outcome for a given value of the covariates. Here, $t'$ is the period in which the treated group is first treated. Violation of no anticipation can also lead to deviations in parallel trends before treatment. Here, $Y_{i,s,t}(1)$ is the treated potential outcome of individual $i$ in state $s$ at calendar year $t$. Under Assumptions \eqref{assumption: spt} and \eqref{assumption: na}, the conditional ATT for a given value of $X_{i,s,t}$ is shown in Equation \eqref{equation: attcond}. Refer to \cite{callaway2021difference} for a simple proof.

\begin{corollary}
[ATT under conditional parallel trends and no anticipation]
\label{theorem: att}
    \begin{equation}
    \label{equation: att}
    \begin{split}
        ATT = & \biggl[E \bigr[[E[Y_{i,s,t}|D_{s} = 1,P_t = 1,X_{i,s,t} = x] - E[Y_{i,s,t}|D_{s} = 1, P_{t} = 0,X_{i,s,t} = x]\biggr] \\ - & \biggl[E[Y_{i,s,t}|D_{s} = 0, P_{t} = 1,X_{i,s,t} = x] - E[Y_{i,s,t}|D_{s} = 0, P_t = 0,X_{i,s,t} = x] \biggr| D_s = 1 \biggr].
    \end{split}
    \end{equation}
\end{corollary}
In Equation \eqref{equation: att}, $E[Y_{i,s,t}|D_s = 1, P_t = 1,X_{i,s,t} = x]$ is the expected outcome for the population in the treated group in the post-intervention period with a covariate value $x$ and $E[Y_{i,s,t}|D_s = 1, P_t = 0,X_{i,s,t} = x]$ is the expected outcome for the population in the treated group in the pre-intervention period with a covariate value $x$. Similarly, $E[Y_{i,s,t}|D_s = 0, P_t = 1,X_{i,s,t} = x]$ is the expected outcome for the population in the control group in the post-intervention period with a covariate value $x$ and $E[Y_{i,s,t}|D_s = 0, P_t = 0,X_{i,s,t} = x]$ is the expected outcome for the population in the control group in the pre-intervention period with a covariate value $x$. It is important to distinguish the unconditional ATT from the conditional ATT shown below, which is just the ATT for a single realization of $X_{i,s,t} = x$. 
\vspace{-0.5em}
\begin{corollary}[Conditional ATT given X under conditional parallel trends and no anticipation]
\label{theorem: attcond}
    \begin{equation}
    \label{equation: attcond}
    \begin{split}
        ATT(x) = & \biggl[E[Y_{i,s,t}|D_{s} = 1,P_t = 1,X_{i,s,t} = x] - E[Y_{i,s,t}|D_{s} = 1, P_{t} = 0,X_{i,s,t} = x]\biggr] \\ - & \biggl[E[Y_{i,s,t}|D_{s} = 0, P_{t} = 1,X_{i,s,t} = x] - E[Y_{i,s,t}|D_{s} = 0, P_t = 0,X_{i,s,t} = x]\biggr].
    \end{split}
    \end{equation}
\end{corollary}
\begin{assumption}[No Staggered Adoption]
\label{assumption: sa}
    All units in the treated group are treated at calendar year $t$. There is no variation in treatment timing.
\end{assumption}
\vspace{-0.75em}
Assumption \eqref{assumption: sa} implies treatment occurs only once \citep{callaway2021difference,de2020twott}. For now, we introduce Assumption \eqref{assumption: sa} to develop an intuitive understanding of the ``simple'' two group, two period building blocks of the UN-DID with staggered adoption. This assumption will be relaxed later when we explore the more complex setup where multiple groups adopt treatment at different times.   

The papers cited above have  contributed to a rich recent literature investigating what occurs when each of the aforementioned assumptions is violated. In addition to the identifying assumptions above, we introduce two data-access assumptions that characterize the information environment in which estimation takes place. These assumptions concern whether data can be combined across silos and what information can be transferred between them. Unlike the assumptions above, they are implementation rather than identification conditions.

\begin{dataassumption}[Data are Poolable]
\label{assumption:pooled}
    Data from the treatment and control groups are available on a single server and can be combined together for analyses. 
    \[ Y^{pooled} = \begin{bmatrix} Y^{Treat}\\ Y^{Control} \end{bmatrix} = \begin{bmatrix} y_{Treat,1}\\ y_{Treat,2}\\ \vdots\\ y_{Treat,N_{Treat}} \\ y_{Control,1}\\ y_{Control,2}\\ \vdots\\ y_{Control, N_{Control},} \end{bmatrix}, X^{Pooled} = \begin{bmatrix} X^{Treat}\\ X^{Control} \end{bmatrix} = \begin{bmatrix} x_{Treat,1}\\ x_{Treat,2}\\ \vdots\\ x_{Treat, N_{Treat}} \\ x_{Control,1}\\ x_{Control,2}\\ \vdots\\ x_{Control, N_{Control}} \end{bmatrix} \]
\end{dataassumption}

Here, $Y^{Treat}$ and $Y^{Control}$ are two matrices containing the outcome variable for group $Treat$ and $Control$, respectively, and $N_{Treat}$ and $N_{Control}$ are the total number of observations in the matrices, respectively. $X^{Treat}$ and $X^{Control}$ are the matrices of covariates for the treated and the control groups, respectively, which can also be stacked together into a common matrix $X$.

\[ Y^{Treat} = \begin{bmatrix} y_{Treat,1}\\ y_{Treat,2}\\ \vdots\\ y_{Treat,N_{Treat}} \end{bmatrix} , Y^{Control} = \begin{bmatrix} y_{Control,1}\\ y_{Control,2}\\ \vdots\\ y_{Control, N_{Control}} \end{bmatrix}, \]

\[  X^{Treat} = \begin{bmatrix} x_{Treat,1}\\ x_{Treat,2}\\ \vdots\\ x_{Treat, N_{Treat}} \end{bmatrix} , X^{Control} = \begin{bmatrix} x_{Control,1}\\ x_{Control,2}\\ \vdots\\ x_{Control, N_{Control}} \end{bmatrix} \]

Under Data Assumption \eqref{assumption:pooled}, the two matrices can be stacked together into a common matrix $Y^{pooled}$, which is not feasible when there are legal restrictions preventing them from being stacked together (data are siloed or unpoolable). For settings where Assumption \eqref{assumption:pooled} is violated and the data are not poolable, we introduce the UN-DID method to estimate the ATT. In such cases, when we access the data from the treated silo, we have no knowledge of the data from the control silo. Similarly, once we have access to the data from the untreated silo, we have no information about the data from the treated silo. The conventional DID method becomes impractical if siloed datasets cannot be combined for regression analysis. 

\begin{dataassumption}[No Entry/Extraction Only]
\label{assumption: noentry}
The various silos contain the individual level observations. The data-sharing agreements permit only a prespecified set of approved summary outputs, such as coefficients, standard errors, and cell counts, to be extracted; arbitrary sufficient statistics are not assumed to be releasable, and no information from another silo can be introduced into a silo’s secure environment.

\end{dataassumption}

Data Assumption \eqref{assumption: noentry} implies that when data are unpoolable, researchers are unable to exchange any information between silos, including coefficients and standard errors from regressions conducted in a different silo. The extracted coefficients and standard errors can only be combined in a server outside the respective data silos. We discuss alternative covariate-adjusted DID estimators including inverse probability weighting, regression adjustment and doubly robust DID and their feasibility under data-access restrictions considered here in Appendix \ref{appendix:matching}. 

\section{Unpooled Difference-in-differences (UN-DID) Estimator}
\label{sec: UNDIDintroduction}

In this section we introduce the UN-DID estimator in the simplest setting, where there are only two silos and two time periods.  The estimator is extended to multiple time periods and multiple silos in Section \ref{sec:mult}.

\subsection{UN-DID with 2 Silos and 2 Time periods}

When Data Assumption \eqref{assumption:pooled} does not hold and Data Assumption \eqref{assumption: noentry} does hold, we can visit each silo (physically or virtually) and run regressions shown in equations \eqref{equation: undid1} and \eqref{equation: undid2} for the treated and untreated silos, respectively. In the simplest case, we assume that there are only two silos: the treated silo ($Treat$) and the untreated silo ($Control$). 
\begin{equation}
\label{equation: undid1}
    \mbox{For treated:} \; Y^{Treat}_{i,t} = \lambda_1^{Treat} pre_t^{Treat} + \lambda_2^{Treat} post_t^{Treat} + \lambda_3^{Treat} X^{Treat}_{i,t} + \nu^{Treat}_{i,t}.
\end{equation}
\begin{equation}
\label{equation: undid2}
    \mbox{For Control:} \; Y^{Control}_{i,t} = \lambda_1^{Control} pre_t^{Control} + \lambda_2^{Control} post_t^{Control} + \lambda_3^{Control} X^{Control}_{i,t} + \nu^{Control}_{i,t}.
\end{equation}

Here, $Y^{Treat}_{i,t}$ is the outcome for an individual from the silo that is treated at time $t$ and $Y^{Control}_{i,t}$ is the outcome of an individual from the silo that is untreated at time $t$. $post_t^{Treat}$ is a dummy variable that takes on a value of 1 when the treated observation is in the post-intervention period, and 0 otherwise. Similarly, $post_t^{Control}$ is a dummy variable that takes on a value of 1 when the untreated observation is in the post-intervention period, and 0 otherwise. $pre_t^{Treat}$ is a dummy variable that takes on a value of 1 if the treated observation is in the pre-intervention period, hence $pre_t^{Treat} = 1 - post_t^{Treat}$. Similarly, $pre_t^{Control}$ is a dummy variable that takes on a value of 1 if the untreated observation is in the pre-intervention period, hence $pre_t^{Control} = 1 - post_t^{Control}$. $X^{Treat}_{i,t}$ and $X^{Control}_{i,t}$ are the covariates for the treated and the control groups respectively. Note that the regressions in equations \eqref{equation: undid1} and \eqref{equation: undid2} do not include a constant. Therefore, none of the variables are dropped because of multicollinearity.

These regressions are silo-specific. Because data are siloed, the treated regression does not contain any data from the untreated silo and the untreated regression does not contain any data from the treated silo. Here, $(\widehat{\lambda}_2^{Treat} - \widehat{\lambda}_1^{Treat}) - (\widehat{\lambda}_2^{Control} - \widehat{\lambda}_1^{Control})$ is the estimate of the ATT, as shown in Equation \eqref{equation: undidattmain}. 

\begin{equation}
\label{equation: undidattmain}
    \widehat{\tau}_{UN} = (\widehat{\lambda}_2^{Treat} - \widehat{\lambda}_1^{Treat}) - (\widehat{\lambda}_2^{Control} - \widehat{\lambda}_1^{Control}).
\end{equation}

\subsection{Cluster-Robust Inference for UN-DID} \label{sec:cluster}
The analytical results above assume that the error terms are independent. This permits standard errors that are robust to heteroskedasticity, but not to clustering. Clustering more finely than the silo level is straightforward when each silo contains multiple clusters. For example, researchers could cluster by individual in panel data or by city when states are the silos. In many applications, however, treatment varies at the silo level, making the silo the natural level of clustering \citep{bertrand2004much,MACKINNON2023}. In this case, each silo-specific dataset contains only one cluster. A conventional cluster-robust variance estimator is therefore undefined because its variance matrix is rank deficient.

We consider two alternatives for inference at the silo level. The first is a cluster jackknife, related to the procedures in \citet{MNW-bootknife} and \citet{hansen2025standard}. The second is randomization inference, following the design-based approach in \citet{mackinnon2020randomization}. Both procedures are implemented in the \texttt{UNDID} software packages and developed further in companion work \citep{karim2025slides}.  This paper only considers inference for the $ATT$, but \citep{karim2025slides} considers inference for sub-aggregate estimands, particularly those at the cohort or state level.

The cluster jackknife deletes one jurisdiction at a time and estimates variance from the dispersion of the leave-one-state-out estimates:
\begin{equation}
\label{eq:jk_var}
\widehat{\mathrm{Var}}_{\mathrm{JK}}(\widehat{\theta})
=
\frac{S-1}{S}\sum_{s=1}^{S}\left(\widehat{\theta}^{(s)}-\bar{\theta}^{(\cdot)}\right)^2,
\qquad
\bar{\theta}^{(\cdot)}=\frac{1}{S}\sum_{s=1}^{S}\widehat{\theta}^{(s)}.
\end{equation}
Here, $S$ is the number of silos and $\widehat{\theta}^{(s)}$ is the estimate after deleting silo $s$. The jackknife is appealing when treatment is assigned at the silo level and the target remains well defined after deletion \citep{MNW-bootknife,karim2026improvedinferencecsdidusing}. In a staggered-adoption design, however, deleting a silo can change the weights, the set of treated group--time cells, or the nuisance parameters used for the remaining cells. 

Randomization inference treats the observed assignment as one draw from a design that fixes the adoption dates and number of adopting states in each cohort. It tests the sharp null of no effect for the units and assignments being permuted, using the raw estimate rather than a studentized statistic \citep{mackinnon2020randomization,athey2022design}. 

For each target, the randomization-inference $p$-value is the share of admissible reassignment estimates that are at least as large as the observed estimate in absolute value. This is a design-based test of a sharp null, not the same inferential object as a conventional or jackknife large-sample test. The distinction is important even when the two procedures have similar rejection frequencies under the placebo null.

\subsection{Differing Covariates by Silo }

One interesting feature of siloed data is that the information may be collected and coded differently according to the 
processes in place in each jurisdiction. Each silo is primarily interested in analyzing the data in isolation for their own purposes, not in comparisons with other jurisdictions. It is therefore likely that not all variables will be standardized or harmonized across silos. UN-DID is designed for settings where data are unpoolable in the sense that they cannot be combined on one server or in one regression, but where the information they capture is similar enough that fundamental issues of data harmonization are not a problem. It is convenient, therefore, that UN-DID can accommodate settings where data are sufficiently harmonized so that variables capture the same concepts across silos, even if they are coded somewhat differently. 

Consider, for instance, a variable like family income. Researchers must confirm that this variable is comparable across silos. For example, they should check whether this variable includes the same types of income (wages, investment earnings, government transfers, etc.) and is continuous and not-censored in both silos. However, it is plausible that both silos record family income as categorical: one using bins in increments of \$10,000, while the other uses increments of \$25,000.  This would cause issues for 
conventional DID, as the categories are non-overlapping across silos without a problematic degree of aggregation. With UN-DID, the silo-specific coefficients can handle the different categories easily, each first-difference can be estimated conditional on the covariate, and the conditional ATT and SEs can also be estimated. 
While variables must be sufficiently comparable across silos to capture the same underlying concepts, UN-DID can accommodate differences in how those variables are measured or coded, as well as differences in their relationships with outcomes across silos. We provide Monte Carlo evidence of this in Appendix \ref{sec:diff_covar}.

\section{The UN-DID Estimator with Covariates}\label{sec:UNDID}
\label{section: Proofs}

In this section, we highlight two key advantages of UN-DID over the conventional estimator. The first advantage of UN-DID is that it provides researchers with a simple yet innovative econometric tool to circumvent the problem of ``unpoolable" data and conduct DID analysis.
The second advantage is that it allows for a greater variety of covariates, especially time-varying covariates. We then present theoretical results of the UN-DID estimator relative to the conventional estimator.  We begin with the simplest case with no covariates and with poolable data. Here, we show that the UN-DID can identify the ATT, and can recover algebraically equivalent estimates to the conventional DID regression when applied to a poolable dataset.

\subsection{UN-DID with two groups/two periods}

We begin by analyzing the simplest case with no covariates. Theorem \ref{lemma: equalATT} proves that both the conventional and the UN-DID regressions converge in probability to the same population parameter. Let $\beta_3^*$ be the probability limit of the conventional DID estimator, and $\tau^*_{UN}$ be the probability limit of the UN-DID estimator. Next, we define the population regression models for the conventional DID and the UN-DID regressions, shown in Equations \eqref{equation: didpopulationnocov} and \eqref{equation: undidpopulationnocov} respectively.
\begin{equation}
\label{equation: didpopulationnocov}
Y_{i,s,t} = \beta_0^* + \beta_1^*D_{s} + \beta_2^*P_{t} +
\beta_3^*P_{t}*D_{s} + \epsilon^*_{i,s,t}.
\end{equation}
\begin{equation}
\label{equation: undidpopulationnocov}
Y^j_{i,t} = \lambda_1^{j*}pre_t^j + \lambda_2^{j*}post_t^j + \epsilon^{j*}_{i,t}, \qquad
j\in\{ Treat,Control\}.
\end{equation}
From Equation \eqref{equation: undidpopulationnocov}, $\tau^*_{UN}$ is defined as:
\begin{equation}
    \label{equation: tauUNdefnocov}
    \tau^*_{UN} = (\lambda_2^{Treat*} - \lambda_1^{Treat*}) - (\lambda_2^{Control*} - \lambda_1^{Control*}).
\end{equation}
The proof relies on a stronger version of Assumption \eqref{assumption: spt}, called the \textbf{strong parallel trends} assumption, introduced in Equation \eqref{equation: pt}. To compare the conventional and UN-DID estimators directly, we consider a setting in which Data Assumption \eqref{assumption:pooled} holds, so that both estimators can be applied to the same underlying data. UN-DID itself can be applied whether or not the data are poolable. When Data Assumption \eqref{assumption:pooled} does not hold, the standard pooled regression is unavailable, while UN-DID remains implementable under the data-access restrictions considered here.

\begin{assumption}[Strong Parallel Trends]
\label{assumption: pt}
    The evolution of untreated potential outcomes is the same between treated and control groups.
    \begin{align}
        \label{equation: pt}
        \begin{split}
           & \biggr[E[Y_{i,s,t}(0)|D_s = 1,P_t=1] - E[Y_{i,s,t}(0)|D_s = 1,P_t=0]\biggr] \\
              = & \ \biggr[E[Y_{i,s,t}(0)|D_s = 0,P_t=1] - E[Y_{i,s,t}(0)|D_s = 0, P_t= 0] \biggr].
        \end{split}
    \end{align}
\end{assumption}
\begin{corollary}
[ATT under strong parallel trends and no anticipation]
\label{theorem: attuncond}
    \begin{equation}
    \label{equation: attuncond}
    \begin{split}
        ATT = & \biggl[E[Y_{i,s,t}|D_{s} = 1,P_t = 1] - E[Y_{i,s,t}|D_{s} = 1, P_{t} = 0]\biggr] \\ - & \biggl[E[Y_{i,s,t}|D_{s} = 0, P_{t} = 1] - E[Y_{i,s,t}|D_{s} = 0, P_t = 0]\biggr].
    \end{split}
    \end{equation}
\end{corollary}

\begin{theorem}[Equivalence of the probability limit of the conventional DID and the UN-DID with no covariates] Suppose Assumptions \eqref{assumption: sutva}, \eqref{assumption: binary}, \eqref{assumption: na} and \eqref{assumption: pt} hold. 
\label{lemma: equalATT}
\begin{equation}
\beta_3^* = \tau_{UN}^*.
\end{equation}
\end{theorem}
A proof of Theorem \eqref{lemma: equalATT} is shown in Appendix \ref{appendix: nocovarpopequivalence}. After showing that the conventional and UN-DID estimators both converge in probability to the ATT as per Theorem 
\eqref{lemma: equalATT}, we can now establish a stronger result that these estimators are \textbf{algebraically equivalent} in finite samples. Algebraic equivalence requires that the two estimators must simplify to the same mathematical expression and, consequently, produce exactly the same numerical results in every finite sample when applied to the same data, no matter the sample size. 

\begin{theorem}[Algebraic Equivalence of the conventional and the UN-DID estimators]
\label{theorem: theorem2}
\begin{equation}
    \widehat{\beta}_3 = \widehat{\tau}_{UN}   
\end{equation}
\end{theorem}

A proof of Theorem \eqref{theorem: theorem2} is provided in Appendix \ref{ssection: Lemma1}. We will now demonstrate that both the conventional DID and the UN-DID estimators are consistent under the standard OLS regularity conditions along with the maintained identifying assumption. 

\begin{theorem}[Consistency of the conventional DID and the UN-DID estimators] Suppose Assumptions \eqref{assumption: sutva}, \eqref{assumption: binary}, \eqref{assumption: na} and \eqref{assumption: pt} hold. Under these assumptions:
\label{theorem: consistencywithoutcovariates}
\begin{equation}
    \widehat{\tau_{UN}} \xrightarrow{p} \tau^*_{UN}, \;\;  \widehat{\beta}_3 \xrightarrow{p} \beta_3^*
\end{equation}
\end{theorem}
The proof of Theorem \eqref{theorem: consistencywithoutcovariates} is shown in Appendix \ref{appendix:consistency}. In the more realistic case, researchers typically want to incorporate covariates into conventional DID to ensure the plausibility of parallel trends conditional on covariates, particularly when strong parallel trends do not hold \citep{heckman1997matching, abadie2005semiparametric}. Likewise, we include only covariates necessary for conditional parallel trends to hold. 

The second main advantage of the UN-DID estimator is that it allows researchers to include a greater variety of covariates compared to the conventional DID estimator. The DID literature emphasizes the importance of thoughtfully selecting covariates in DID analyses. In particular, it suggests using covariates that are time-invariant \citep{abadie2005semiparametric, callaway2021difference, caetano2022timevarying, roth2022s, sant2020doubly}. When dealing with time-varying covariates, researchers frequently select a specific value of these covariates during a pre-treatment period, as if it were a time-invariant covariate \citep{caetano2022timevarying}. The UN-DID estimator enables researchers to use time-varying covariates whose relationships with outcomes differ across silos, as discussed in more detail in the following subsection. In particular, we will show that the UN-DID estimator can provide consistent estimates of the ATT when time-varying covariates with varying effects on outcomes are included, whereas the conventional estimator is inconsistent.

\subsection{Common Causal Covariates} \label{sec3:ccc}

Time-varying covariates are known to complicate DID analysis \citep{caetano2022timevarying}. Additionally, \cite{karim2024good} demonstrate that the conventional TWFE estimator can be biased when the common causal covariates (CCC) assumption is violated. The CCC assumption has been implied in previous DID literature, but has not been explicitly addressed. Their paper introduces three types of CCC assumptions, but this study focuses on the state-invariant CCC assumption, as stated below.

\begin{assumption}[State-invariant Common Causal Covariate]
\label{assumption: stateinvariantccc} The relationship between the covariate and the untreated potential outcome variable is equal between silos.  Consider the following silo-specific Data Generating Processes (DGPs) for the untreated potential outcomes.
\begin{align}
\text{For treated:} \quad 
    Y^{Treat}_{i,t}(0) &\equiv \alpha_{\circ}^{Treat} pre_t^{Treat} + \beta_{\circ}^{Treat} post_t^{Treat} + \lambda_{\circ}^{Treat} X^{Treat}_{i,t} + \mu^{Treat}_{i,t}. \label{eq: ccc1} \\
\text{For control:} \quad 
    Y^{Control}_{i,t}(0) &\equiv \alpha_{\circ}^{Control} pre_t^{Control} + \beta_{\circ}^{Control} post_t^{Control} + \lambda_{\circ}^{Control} X^{Control}_{i,t} + \mu^{Control}_{i,t}. \label{eq: ccc2}
\end{align}

\noindent The subscript $\circ$ denotes the true value of a parameter in the DGP.
We assume that $\mathbb{E}\left[\mu^{j}_{i,t} \mid P_t,X^{j}_{i,t} \right] = 0$ for $ j\in\{Treat,Control\}$. The coefficients relating the covariates to the untreated potential outcome are assumed to be equal in the treatment and control groups: $\lambda^{Treat}_{\circ} = \lambda^{Control}_{\circ}$. Throughout the analysis, a violation of Assumption \eqref{assumption: stateinvariantccc} refers specifically to a failure of this equality, such that $\lambda^{Treat}_{\circ} \neq \lambda^{Control}_{\circ}$. The silo-specific DGPs, shown in Equations \eqref{eq: ccc1} and \eqref{eq: ccc2}, together with the conditional mean zero restriction on $\mu^{j}_{i,t}$, remain maintained.

\end{assumption}

One important distinction between the conventional DID regression using poolable data in Equation \eqref{equation: didsimple} and the UN-DID regressions using siloed data in Equations \eqref{equation: undid1} and 
\eqref{equation: undid2} is how the covariates are handled. In the conventional DID regression there is a single coefficient, $\beta_4$, whereas in the UN-DID regressions there are two $\lambda_3^{Treat}$ and $\lambda_3^{Control}$. Whether this distinction matters depends on whether the relationship between the covariates and the untreated potential outcome varies across silos. When Assumption \eqref{assumption: stateinvariantccc} is violated, the conventional DID estimator and other estimators such as \cite{callaway2021difference} can be biased \citep{karim2024good}.  Conversely, under the valid identifying assumptions, the UN-DID estimator is consistent for the ATT in these settings.  The closely related state-varying Integrated Difference-in-Differences (DID-INT) estimator developed for pooled data in \cite{karim2024good} also allows covariate effects to vary across silos.  

\begin{assumption}[Covariate Invariant Treatment Effect] \label{assumption: homATTx} The ATT for the treated population in the post-intervention period does not vary with the covariates. For every $x$ in $X_{i,s,t}$ for the treated population:
\begin{equation}
    ATT = \mathbb{E}[Y_{i,s,t}(1) - Y_{i,s,t}(0) \mid D_s = 1, P_t = 1, X_{i,s,t} = x] = \mathbb{E}[Y_{i,s,t}(1) - Y_{i,s,t}(0) \mid D_s = 1, P_t = 1].
\end{equation}
\end{assumption}

Assumption \eqref{assumption: homATTx} allows individual treatment effects to be heterogeneous, but requires the conditional ATT to remain invariant with respect to covariates. We maintain this restriction throughout the paper. When the ATT is allowed to vary with covariates, the integrated DID-INT developed by \cite{karim2024good} can account for heterogeneity by integrating covariate-specific ATTs over the relevant treated covariate distribution. However, under the data poolability restrictions considered in this paper, implementing such an approach is challenging and represents a scope for future research.  

\subsection{Conventional and UN-DID ATT estimates with  covariates}

This subsection compares the conventional and UN-DID estimators with covariates. We analyze separately the scenarios where Assumption \eqref{assumption: stateinvariantccc} holds and where it is violated, while maintaining Assumptions \eqref{assumption: sutva}, \eqref{assumption: binary}, \eqref{assumption: spt},  \eqref{assumption: na} and \eqref{assumption: homATTx}. 
We also introduce a linearity assumption for the untreated potential outcomes below and relate the corresponding latent population regression coefficients to the silo-specific untreated potential outcome DGPs. Under these conditions, Assumption \eqref{assumption: stateinvariantccc} is no longer required to identify the ATT.  Moreover, we demonstrate that when a common covariate coefficient is imposed across silos, the conventional regression becomes mis-specified if Assumption \eqref{assumption: stateinvariantccc} is violated, causing the estimator to be biased. Furthermore, we find that the UN-DID and conventional estimators are generally not algebraically equivalent in finite samples with covariates, except when time-invariant covariates are used and panel data are available. We begin by defining the population regression models for both the conventional DID and the UN-DID regressions, as presented in Equations 
\eqref{equation: didpopulation} and \eqref{equation: undidpopulation}. For simplicity, let us assume that there is a single covariate. Using a common poolable dataset, we then compare the probability limits of the conventional DID and UN-DID estimators. We show that they coincide when Assumption \eqref{assumption: stateinvariantccc} holds but can differ when it is violated. 
\begin{equation}
\label{equation: didpopulation}
Y_{i,s,t} = \beta_0^* + \beta_1^*D_{s} + \beta_2^*P_{t} +
\beta_3^*P_{t}*D_{s} + \beta_4^{*}X_{i,s,t} + \epsilon^*_{i,s,t}.
\end{equation}
\begin{equation}
\label{equation: undidpopulation}
Y^j_{i,t} = \lambda_1^{j*}pre_t^j
+ \lambda_2^{j*}post_t^j +
\lambda_3^{j*}X^j_{i,t} +
\epsilon^{j*}_{i,t}, \qquad
j\in\{ Treat,Control\}.
\end{equation}

Equations \eqref{equation: didpopulation} and \eqref{equation: undidpopulation} are the population models for the observed outcomes, while Equations \eqref{eq: ccc1} and \eqref{eq: ccc2} represent the silo-specific data generating processes (DGPs) for the untreated potential outcomes. The coefficients might differ between these two representations. To link the untreated potential outcome DGPs to the corresponding population regression coefficients, consider the following latent population regression for the untreated potential outcome:
\begin{equation}
\label{equation: undidlatentpopulation}
Y^j_{i,t}(0) = \lambda_{1}^{j*,\circ}pre_t^j
+ \lambda_{2}^{j*,\circ}post_t^j +
\lambda_{3}^{j*,\circ}X^j_{i,t} +
\epsilon^{j*,\circ}_{i,t}, \qquad
j\in\{ Treat,Control\}.
\end{equation}
Equation \eqref{equation: undidlatentpopulation} is the latent population counterpart to the untreated potential outcome DGPs in Equations \eqref{eq: ccc1} and \eqref{eq: ccc2}. Under correct specification, the population coefficients in this latent regression coincide with the corresponding coefficients of the untreated potential outcome DGP. Therefore, $\lambda_{1}^{j*,\circ} = \alpha_\circ^j$, $\lambda_{2}^{j*,\circ} = \beta_\circ^j$ and $\lambda_{3}^{j*,\circ} = \lambda_\circ^j$. We next impose a linearity condition on the conditional mean of the untreated potential outcome using the coefficients from the latent population regression.

\begin{assumption}[Linearity] 
    \label{as: linearity}
    The conditional mean of the untreated potential outcome 
    $Y^j_{i,t}(0)$ is linear in the covariates $X^j_{i,t}$ within each silo.
\begin{equation}
\label{equation: linearity}
\begin{aligned}
\mathbb{E}\left[
Y^j_{i,t}(0) \mid S_i=s,P_t=p,X^j_{i,t}=x\right]&=\lambda_{p+1}^{j*,\circ} +
\lambda^{j*, \circ}_{3}x.
\end{aligned}
\end{equation} 
\end{assumption}
We use the index $p+1$ to denote period-specific intercepts. When $P_t = 0$ in the pre-intervention period, the intercept is $\lambda^{j*,\circ}_{0+1} = \lambda^{j*,\circ}_{1}$, while when $P_t = 1$ in the post-intervention period, the intercept is $\lambda^{j*,\circ}_{1+1} = \lambda^{j*,\circ}_{2}$. Under Assumption \eqref{assumption: homATTx}, the covariate coefficient in the observed-outcome population regression coincides with the corresponding coefficient in the latent untreated potential outcome regression under the correct specification. Therefore,
\begin{equation*}
   \lambda_3^{j*} = \lambda_{3}^{j*,\circ}, \qquad
j\in\{Treat,Control\}. 
\end{equation*}

Using this coefficient mapping, Assumption \eqref{as: linearity} implies that the conditional and unconditional means of the covariate-adjusted untreated potential outcomes are equal within each region-period cell. 
\begin{equation}
\label{equation: meanindependence}
    \mathbb{E}\left[Y^j_{i,t}(0) - \lambda^{j*}_3 X^j_{i,t} \mid S_i=s,P_t=p,X^j_{i,t}= x\right] = \mathbb{E}\left[Y^j_{i,t}(0) - \lambda^{j*}_3 X^j_{i,t}  \mid S_i=s,P_t=p\right].
\end{equation} 
Here, $S_i$ is the silo containing individual $i$. A derivation of this result is shown in Appendix \ref{conditionalmeanindependenceresult}.

\begin{theorem}[The UN-DID population parameter $\tau^*_{UN}$ identifies the ATT with covariates]
\label{theorem: theorem3} Suppose Assumptions \eqref{assumption: sutva}, \eqref{assumption: binary}, \eqref{assumption: spt}, \eqref{assumption: na}, \eqref{assumption: homATTx} and \eqref{as: linearity} hold and the relevant population models are correctly specified, regardless of whether Assumption \eqref{assumption: stateinvariantccc} holds or is violated. Under these assumptions:
\begin{equation}
    \tau^*_{UN} = ATT
\end{equation}
\end{theorem}
A proof of Theorem \eqref{theorem: theorem3} is shown in Appendix \ref{theorem:UNDIDtau}. Theorem \eqref{theorem: theorem3} establishes that the UN-DID population parameter identifies the ATT even when the state-invariant CCC assumption is violated. The key feature is that UN-DID estimates the covariate relationship separately within each silo, so identification does not require the covariate coefficient to be common across both treated and control silos. 

\begin{theorem}[Non-equivalence of the probability limit of conventional DID and UN-DID under state-invariant CCC violations] Suppose Assumptions \eqref{assumption: sutva}, \eqref{assumption: binary},  \eqref{assumption: spt}, \eqref{assumption: na},  \eqref{assumption: homATTx}, and \eqref{as: linearity} hold, and the relevant population models are correctly identified, while Assumption \eqref{assumption: stateinvariantccc} is violated. Suppose further that the resulting misspecification components do not exactly offset in the conventional DID estimand. Under these conditions,
\label{lemma: equivtimevaryingcov}
\begin{equation}
\beta_3^* \neq \tau_{UN}^*.
\end{equation}
\end{theorem}
A proof of Theorem \eqref{lemma: equivtimevaryingcov} is shown in Appendix \ref{appendix:nonequivconvundid}. Theorem \eqref{lemma: equivtimevaryingcov} shows why the conventional DID and UN-DID population parameters are not equivalent when the state-invariant CCC assumption is violated. The conventional specification imposes a common covariate coefficient across silos, whereas UN-DID allows that coefficient to vary by silo. When the true covariate effects differ across silos, the conventional DID is misspecified. Unless the resulting misspecification happens to offset exactly, the conventional DID estimand differs from the ATT. However, Theorem \eqref{theorem: theorem3}  establishes that the UN-DID population parameter continues to identify the ATT.

This problem can be circumvented by interacting the covariates with an indicator for each silo in the conventional regression, as shown in Equation \eqref{equation: didinteracted}. This approach allows researchers to capture differing covariate effects across silos and mitigate bias from CCC violations, as highlighted in \cite{karim2024good}. A recent alternative is the FLEX estimator proposed in \cite{deb2024flexible}; however, it only considers treatment-cohort group-level variation in coefficients, not silo-level variation. In practice, researchers have typically not done this.
\begin{equation}
\label{equation: didinteracted}
    Y_{i,s,t} = \iota_0 + \iota_1 D_{s} + \iota_2 P_{t} + \iota_3 P_{t} * D_{s} + \iota_4 D_s*X_{i,s,t} + \iota_5 (1-D_s)*X_{i,s,t} + \varepsilon_{i,s,t}.
\end{equation}

The key difference between $\beta_3^*$ in Equation \eqref{equation: didpopulation} and $\tau_{UN}^*$ in Equation \eqref{equation: undidpopulation} lies in how covariate coefficients are identified under each regression. If the true covariate coefficients vary by silo, the common-slope specification of the conventional DID regression is misspecified. The UN-DID regression allows the coefficients on the covariates to differ across silos because it is estimated separately within each silo. 

\begin{theorem}[Consistency of the UN-DID estimator for the  ATT] Suppose Assumptions \eqref{assumption: sutva}, \eqref{assumption: binary}, \eqref{assumption: spt}, \eqref{assumption: na}, \eqref{assumption: homATTx},    and \eqref{as: linearity} hold, and the relevant population models are correctly identified, regardless of whether Assumption \eqref{assumption: stateinvariantccc} holds or is violated.
\label{theorem: consistencyundidwithcov}
\begin{equation}
\widehat{\tau}_{UN} \xrightarrow{p} \tau_{UN}^{*} = ATT.
\end{equation}
\end{theorem}
\begin{theorem}[Inconsistency of the conventional DID estimator for the ATT] Suppose Assumptions \eqref{assumption: sutva}, \eqref{assumption: binary}, \eqref{assumption: spt}, \eqref{assumption: na},  \eqref{assumption: homATTx},  and \eqref{as: linearity} hold, while Assumption \eqref{assumption: stateinvariantccc} is violated. Suppose further that the resulting misspecification components do not exactly offset in the conventional DID estimand. Under these conditions,
\label{theorem: inconsistencydidwithcov}
\begin{equation}
\widehat{\beta}_3 \xrightarrow{p} \beta_3^* \neq ATT.
\end{equation}
\end{theorem}
Under the standard OLS regularity conditions, the silo-specific coefficients converge in probability to their population parameters. Therefore, $\widehat{\tau}_{UN} \xrightarrow{p} \tau_{UN}^{*}$, and Theorem \ref{theorem: theorem3} establishes $\tau_{UN}^{*} = ATT$. Theorem \ref{theorem: inconsistencydidwithcov} similarly follows from the convergence of the conventional OLS estimator to its population parameter together with Theorem \ref{lemma: equivtimevaryingcov}. 
\begin{theorem}[Equivalence of the probability limit of conventional DID and UN-DID when state-invariant CCC assumption holds] Suppose Assumptions \eqref{assumption: sutva}, \eqref{assumption: binary}, \eqref{assumption: spt}, \eqref{assumption: na}, \eqref{assumption: stateinvariantccc}, \eqref{assumption: homATTx} and \eqref{as: linearity} hold. Under these assumptions,
\label{theorem: theorem4}
\begin{equation}
\beta_3^* = \tau_{UN}^*.
\end{equation}
\end{theorem}
When the state-invariant CCC assumption holds, the covariate coefficients are common across treated and control silos. Therefore, the conventional DID specification is correctly specified, and both the conventional DID and UN-DID population parameters identify the ATT. A proof of Theorem \eqref{theorem: theorem4} is shown in Appendix \ref{appendix:Equivalenceundidconvccc}. Although the conventional DID and UN-DID estimators with covariates both converge in probability to the same population parameter under Assumption \eqref{assumption: stateinvariantccc}, they are not algebraically equivalent. In this context, the two estimators can give different estimates in finite samples, even though their difference approaches zero as the sample size grows. 
\begin{theorem}[Non-equivalence of the UN-DID and the conventional DID estimators in finite samples when Assumption \eqref{assumption: stateinvariantccc} holds]
\label{theorem: nonequivestimatorsassumption9}
    \begin{equation}
    \widehat{\beta}_3 \not\equiv \widehat{\tau}_{UN}.
\end{equation}
\end{theorem}
Therefore, the two estimators can produce different estimates in finite samples, although they may coincide for particular sample restrictions. A proof of Theorem \eqref{theorem: nonequivestimatorsassumption9} is shown in Appendix \ref{sec:proof_conventionaltimevaryingatt}.
\begin{theorem}[Algebraic equivalence of the UN-DID and the conventional DID estimators in finite samples with time-invariant covariates and a balanced panel structure] Suppose the same underlying balanced panel is used and the covariates are time invariant.
\label{theorem: theorem11}
\begin{equation}
\widehat{\beta}_3 = \widehat{\tau}_{UN}.
\end{equation}    
\end{theorem}
A proof of Theorem \eqref{theorem: theorem11} is shown in Appendix \ref{sec: theorem11}.

\begin{theorem}[Equivalence of the probability limit of conventional DID and UN-DID under time-invariant covariates and a balanced panel data structure] Suppose Assumptions \eqref{assumption: sutva}, \eqref{assumption: binary},  \eqref{assumption: spt}, \eqref{assumption: na}, \eqref{assumption: homATTx}, and \eqref{as: linearity} hold, regardless of whether Assumption \eqref{assumption: stateinvariantccc} holds or is violated. Under these assumptions, 
\label{theorem: theorem5}
\begin{equation}
\beta_3^* = \tau_{UN}^*.
\end{equation}
\end{theorem}
Since the conventional DID and the UN-DID estimators are algebraically equivalent under time-invariant covariates and a balanced panel, as shown in Appendix \ref{sec: theorem11}, it follows that their probability limits are also equal.

\subsection{Equivalence of the UN-DID estimator to the Integrated Difference-in-Differences under CCC violations}

\cite{karim2024good} show that the Two-Way Fixed Effects (TWFE) DID estimator is inconsistent for the ATT when any of the three types of CCC assumptions stated in their paper are violated. They also demonstrate that the three types of CCC violations are plausible in real datasets. In light of this issue, \cite{karim2024good} propose a new estimator called the Integrated Difference-in-differences (DID-INT), which can provide a consistent estimate of the ATT under the three types of CCC violations. In a simple two-group and two-period setup, the TWFE estimator is equivalent to the conventional regression shown in Equation \eqref{equation: didsimple}. 

In this paper, we show that the UN-DID estimator is equivalent to the state-varying DID-INT without staggered adoption in a two-group, two-period setting. The state-varying DID-INT is defined in Equation \eqref{equation: DIDINT}. The state-varying DID-INT is designed to handle violations of Assumption \eqref{assumption: stateinvariantccc}. It is important to note that the current UN-DID specification allows covariate coefficients to differ across silos but requires them to remain constant over time within a silo. Therefore, it does not accommodate violations of the time-invariant CCC restriction arising from within-silo changes in covariate coefficients over time.  However, a modified version of UN-DID could handle two-way violations. Let $\tau^*_{DT}$ be the probability limit of the state-varying DID-INT. 
\begin{equation}
    \label{equation: DIDINT}
        Y_{i,s,t} = \psi_1 pre_t^{Treat} + \psi_2 post_t^{Treat} + \psi_3 pre_t^{Control} + \psi_4 post_t^{Control} + \psi_5 X^{Treat}_{i,s,t} + \psi_6 X^{Control}_{i,s,t} + e_{i,t}
\end{equation}
In the above regression, $pre_t^{Treat} = pre_t * D_s$; $post_t^{Treat} = post_t * D_s$; $pre_t^{Control} = pre_t * (1 - D_s)$; $post_t^{Control} = post_t * (1 - D_s)$; $X^{Treat}_{i,s,t} = X_{i,s,t} * D_s$ and $X^{Control}_{i,s,t} = X_{i,s,t} *(1 - D_s)$. In other words, $pre_t$, $post_t$ and $X_{i,s,t}$ are interacted with each silo. The estimate of the ATT from the above regression is given by:
\begin{equation}
    \widehat{\tau}_{DT} = \left[ (\hat{\psi_2} - \hat{\psi_1}) - (\hat{\psi_4} - \hat{\psi_3})\right]
\end{equation}

\begin{theorem}[The population parameter of the state-varying DID-INT, $\tau^*_{DT}$, identifies the ATT with covariates] Suppose Assumptions \eqref{assumption: sutva}, \eqref{assumption: binary}, \eqref{assumption: spt}, \eqref{assumption: na}, \eqref{assumption: homATTx} and \eqref{as: linearity} hold and the state-varying DID-INT is correctly specified, regardless of whether Assumption \eqref{assumption: stateinvariantccc} holds or is violated. Under these assumptions,
\label{theorem: DIDINT}
\begin{equation}
    \tau^*_{DT} = ATT
\end{equation}
\end{theorem}

A proof of Theorem \eqref{theorem: DIDINT} is shown in Appendix \ref{appendix:DIDINT}.

\begin{theorem}[Equivalence of the UN-DID and the state-varying DID-INT in finite samples] Assuming that the two estimators are computed using the same underlying sample, 
\label{theorem: theorem13DIDINTequiv}
\begin{equation}
\widehat{\tau}_{DT} = \widehat{\tau}_{UN}
\end{equation}
\end{theorem}
A proof of Theorem \eqref{theorem: theorem13DIDINTequiv} is shown in Appendix \ref{section:proof_RPeqUNDID}.

To summarize, this section demonstrates two key advantages of UN-DID over the conventional estimator. First, we show that UN-DID can estimate the ATT using unpooled data in the data-access environment considered here, whereas the standard pooled DID implementation is unavailable. Second, we show that UN-DID allows time-varying covariates with differing slopes across silos to be used. We also establish the equivalence of the UN-DID and the conventional estimators in settings without covariates, or with time-invariant covariates. With time-varying covariates, the two estimators are generally not algebraically equivalent. When the state-invariant CCC assumption holds, both converge to the ATT. When it is violated, the conventional common-slope specification can be biased, whereas UN-DID remains consistent for the ATT under the maintained assumptions. In addition, we have introduced the state-invariant CCC assumption, and shown that the conventional estimator becomes biased when the assumption is violated and time-varying covariates are used. In contrast, the UN-DID remains unbiased.

\section{Multiple Silos Per Treatment Time} \label{sec:mult}

Up until this point, we have concentrated on settings in which there is only one silo per treatment time.  However, even with common timing it is typically the case that
many treated and many comparison silos will be used.  In that case, it is easier to estimate the ATT by means of a second stage regression.  The details differ slightly depending on whether there is common or staggered adoption. Our software packages always use this approach, in part because it makes it computationally much faster to conduct cluster robust inference.

\subsection{Common Adoption}

Imagine that there are two treated silos, A and B, and one control silo, C. Under Data Assumption \eqref{assumption: noentry}, information from one silo cannot be transferred into another. The treatment effect is therefore constructed in a separate estimation environment, which we call the central server, using only the approved outputs exported by each silo. In the first stage, each silo separately regresses the outcome on a complete set of time indicators and the available covariates:
\begin{equation}
    Y_{ist}
    =
    \sum_{\tau=1}^{T}\lambda_{\tau}^{s}
    \mathbbm{1}\{t=\tau\}
    +
    X_{ist}'\gamma_s
    +
    \varepsilon_{ist}.
    \label{equation: firststagetime}
\end{equation}
Because this regression is estimated separately within each silo, both the time effects and the covariate coefficients may differ across silos. Individual-level observations never leave the silo. Each silo exports only its estimated time effects, $\widehat{\lambda}_{t}^{s}$, together with the number of observations used to estimate each silo-time effect, to the central server.

The first stage therefore creates a server-level dataset with one observation for every silo-time combination. With three silos and $T$ periods, this dataset contains $3T$ observations. Its principal variables are $\widehat{\lambda}_{t}^{s}$, the estimated outcome for silo $s$ in period $t$; $D_{s,t}$, an indicator that the silo is treated in that period; and $N_{s,t}$, the number of individual observations represented by the silo-time estimate. We define the population weight as
\begin{equation}
    w_{s,t}
    =
    \frac{N_{s,t}}
    {\sum_{j=1}^{S}\sum_{\tau=1}^{T}N_{j,\tau}}.
    \label{equation: silotimeweight}
\end{equation}
In the second stage, the server estimates
\begin{equation}
   \widehat{\lambda}_{t}^{s}
    =
    \alpha_s
    +
    \delta_t
    +
    \beta D_{s,t}
    +
    u_{s,t},
    \label{equation: secondstagetwfe}
\end{equation}
using weighted least squares with weights $w_{s,t}$. Here, $\alpha_s$ denotes silo fixed effects, $\delta_t$ denotes time fixed effects, and $\beta$ is the estimated $ATT$. Thus, the second stage has the same two-way fixed-effects structure as the corresponding pooled individual-level regression.

Without covariates, population weighting makes Equation \eqref{equation: secondstagetwfe} algebraically equivalent to pooled two-way fixed-effects estimation, even when silo-time populations differ. When covariates are included, allowing their coefficients to differ across silos and then applying population weights in the second stage provides a very close approximation to pooled estimation with silo-specific covariate slopes. 

This representation is also useful for cluster-level inference. Because the server-level dataset identifies the silo associated with every time-effect estimate, a cluster jackknife can be implemented by deleting all observations from one silo at a time and re-estimating Equation \eqref{equation: secondstagetwfe}. The resulting leave-one-silo-out estimates can be used to construct a jackknife standard error without requiring individual-level data to be pooled across silos.

\subsection{Staggered Adoption With Unique State Timing}\label{sec:stag}

UN-DID is easily extended to staggered adoption contexts. Suppose we have data for multiple silos, labeled as $s = 1, 2, \ldots, S$. The silos which received the intervention or treatment are called the treated groups, and those that did not are called the control group. For expositional clarity, consider that only one silo is included in each timing group, and only one silo is a control. We will shortly relax this requirement. We also assume data for several calendar years, $t = 1,2,\ldots,T$. In this setting, different silos adopt treatment at different points in time, creating a staggered adoption framework. The period before a group $s$ is treated is called the pre-intervention period. We impose an additional restriction that treatment is an absorbing state, which implies that, once a unit is treated, it remains treated for the rest of the period of the study.   Several papers in the existing DID literature, such as \cite{callaway2021difference}, \cite{de2020twott} and \cite{borusyak2017revisiting}, have already proposed methods to estimate the ATT when data are poolable. 

Similar to \cite{callaway2021difference}, UN-DID with staggered adoption decomposes the dataset into several $2\times2$ group--time blocks, which are the building blocks for estimating the ATT under staggered treatment timing. 
Following \cite{callaway2021difference}, we group together each individual into $s \in S$ cohorts based on the treatment start-time dummies. Let $t^s$ be the year group $s$ is first treated. For a particular group $s$, any period where $t<t^s$ is called a pre-treatment period. However, following \cite{callaway2021difference}, we define the pre-treatment period for group $s$ as the year before the group is treated ($t^s-1$). Each of these group--time blocks contains a group that is currently treated and a group that has not yet been treated. For each $2\times2$ group--time block, we will run the following two regressions:

\begin{equation}
\label{equation: SAtreat}
    \mbox{For treated:} \; Y^{Treat}_{i,t} = \zeta_1^{Treat} pre_t^{Treat} + \zeta_2^{Treat} post_t^{Treat} + \zeta_3^{Treat} X^{Treat}_{i,t} + \nu^{Treat}_{i,t}.
\end{equation}
\begin{equation}
\label{equation:SAcontrol}
    \mbox{For untreated:} \; Y^{Control}_{i,t} = \zeta_1^{Control} pre_t^{Control} + \zeta_2^{Control} post_t^{Control} + \zeta_3^{Control} X^{Control}_{i,t} + \nu^{Control}_{i,t}.
\end{equation}

In the above regressions, $Treat$ is an index for the treated group and $Control$ is an index for the control group. $X^j_{i,t}$'s are the covariates researchers may be interested in controlling for, where $j = \{Treat,Control\}$. $\nu^j_{i,t}$ are the potentially clustered error terms. After these two regressions are estimated, the estimate of the ATT(s,t) is:
\begin{equation}
\widehat{ATT(s,t)} = (\widehat\zeta^{Treat}_2 - \widehat\zeta^{Treat}_1)-(\widehat\zeta^{Control}_2 - \widehat\zeta^{Control}_1).
\label{equation: ATTsts}
\end{equation}

\noindent Note, the ATT${(s,t)}$ defined here is identical to the ATT${(g,t)}$ from \cite{callaway2021difference} when there is only one silo per treatment timing group.  However, that is not the case when there are additional silos, as in Section \ref{sec:stag2}.  There are many cases in which ATT${(s,t)}$ are the more relevant objects as discussed in detail in our companion project \cite{karim2025slides}.

The above process is repeated for all relevant $2\times2$ group--time blocks.  These ATT$(s,t)$ cells are then aggregated together to get an overall estimate of the ATT based on Equation \eqref{eq:sa}. $w_{s,t}$ are the weights associated with each of the ATT(s,t) cells. Common weighting schemes include equal weighting for each cell and population-based weighting. Refer to \cite{callaway2021difference} and \cite{xiong2023federated} for details. 

\begin{equation}\label{eq:sa}
\widehat{ATT}
=
\sum_{s:\,t^s\leq T}
\sum_{t=t^s}^{T}
w_{s,t}\,\widehat{ATT}(s,t).
\end{equation}

\subsection{Staggered Adoption With Non-unique State Timing}\label{sec:stag2}

Staggered adoption imposes an additional computational burden when multiple silos share treatment timing.  Each individual ATT$(s,t)$ can 
be calculated as before, and then an overall ATT can be estimated as a weighted average (see Equation \eqref{eq:sa}). To fix ideas, consider the 
simplest staggered adoption setting, where there are three silos (A,B,C) and three time periods (1,2,3).  In this section we use the treatment timing notation from \cite{callaway2021difference}. Silo A
is first treated in period 2 ($g=2$).  Silo B is first treated in period 3 ($g=3$).  Silo C is never treated ($g=\infty$).  
In this 3$\times$3 example, without covariates, and using the never-treated cells as controls, the three $ATT(s,t)$ terms are:
\begin{equation*}
	\begin{array}{rcl}
		ATT_{A2} & = & (\bar Y_{A2} - \bar Y_{A1}) - (\bar Y_{C2}  - \bar Y_{C1}) \\[8pt]
		ATT_{A3} & = & (\bar Y_{A3} - \bar Y_{A1}) - (\bar Y_{C3}  - \bar Y_{C1}) \\[8pt]
		ATT_{B3} & = & (\bar Y_{B3} - \bar Y_{B2}) - (\bar Y_{C3}  - \bar Y_{C2})
	\end{array}
\end{equation*}
In this example, we need to calculate differences in mean outcomes within each silo for potentially many pairs of years per silo. Specifically, we need 
to estimate the difference between year 2 and 1, and 3 and 1 for silo A.  For silo B, we only need the difference between year 3 and 2.  Finally, for silo C, we need all three of these differences. We propose estimating these differences, and their standard errors, within each silo, exporting them and storing them together in a matrix on the central server. In the example above, we need to store six different contrasts. In addition to the contrast, or difference, we need to store: which silo the contrast is from ($s$); the (post treatment) time period for the contrast ($t$); whether a silo is treated ($d=1$) or control ($d=0$) for that comparison; and which sub-aggregate ATT ``block'' we are estimating.  For the block, we introduce the index $h$.  

For control silos the value of $h$ is set to the silo's value of $g$.  For a comparison group, each difference can be needed for a different sub-aggregate ATT block.  In the example above, for control silo C, the 3-1 and 2-1 differences serve as comparisons for group A, so these rows would be coded as $h=2$; however, the 3-2 difference for silo C serves as a comparison for  B, so this row would be coded as $h=3$.

Table \ref{tab:stagmat} illustrates this matrix for a simple 3$\times$3 example. With multiple silos per group, each silo is recorded in its own row for each difference across relevant years. To fill this matrix, each of the within-silo, two-year differences is calculated using a separate regression.

\begin{table}[ht]
    \centering
    \caption{Second Stage Matrix}
    \begin{tabular}{l l c c c c}
        \toprule
        Contrast & $s$ & $d$ & $h$ & T & diff \\
        \midrule
        A22 & A & 1 & 2 & 2 & $\bar{Y}_{A2} - \bar{Y}_{A1}$ \\
        A23 & A & 1 & 2 & 3 & $\bar{Y}_{A3} - \bar{Y}_{A1}$ \\
        B33 & B & 1 & 3 & 3 & $\bar{Y}_{B3} - \bar{Y}_{B2}$ \\
        C22 & C & 0 & 2 & 2 & $\bar{Y}_{C2} - \bar{Y}_{C1}$ \\
        C23 & C & 0 & 2 & 3 & $\bar{Y}_{C3} - \bar{Y}_{C1}$ \\
        C33 & C & 0 & 3 & 3 & $\bar{Y}_{C3} - \bar{Y}_{C2}$ \\
        \bottomrule
    \end{tabular}
    \label{tab:stagmat}
\end{table}

 \noindent This second stage matrix, assembled on the central server, can then be used to estimate each of the ATT$(g,t)$ terms.  Specifically, consider the following regression. 

 \begin{equation}
     \textrm{diff}_{s,h,t} = \alpha + \beta d_{s,h,t} + \epsilon_{s,h,t}  \textrm{ if } h=g \textrm{ and } t=T. 
    \label{equation:second_stage}
 \end{equation}

 In this regression $\beta$ estimates the  $ATT(g,t)$ given by the sample restriction.  For the 3$\times$3 example,
 the three restrictions needed would be $h$=2, $t$=2; $h$=2, $t$=3; and $h$=3, $t$=3.  Our software packages implement this algorithm.  An ATT$(s,t)$ term can similarly be estimated by restricting the sample to only one treated silo per cohort.  For instance, imagine that the sample included a fourth silo, D, which shared treatment timing with B. This would add a 7th row to the matrix in Table \ref{tab:stagmat}.  Estimating  \eqref{equation:second_stage} with the restriction $h=3$, and $t=3$ would yield an estimate of ATT$(3,3)$ for the combined treatment effect in B and D. However, one could further restrict the sample for $s \in \{B,C\}, h=3, t=3$ to estimate ATT$(B,3)$ or  $s \in \{D,C\}, h=3, t=3$ to estimate ATT$(D,3)$. An aggregate ATT can then be formed using the desired weighted average of the ATT$(g,t)$ or  ATT$(s,t)$  terms.  Population based weights also necessitate storing silo-year specific cell counts.
 
\subsection{UN-DID Software Packages}\label{sec:software}

We have developed a set of software packages to aid in the implementation of the UN-DID estimator.  There are versions available in \texttt{R}, \texttt{Stata}, \texttt{Python} and \texttt{Julia}.  At present, the \texttt{R} version is the most polished and has been accepted on CRAN. The documentation for the R-package can be found at \cite{undidr}.  The \texttt{Stata} package is available here: \url{https://github.com/ebjamieson97/undid}. 
The \texttt{Julia} program can be found here: \url{https://github.com/ebjamieson97/Undid.jl}.
The \texttt{Python} version works as a wrapper to call the \texttt{Julia} program, and can be found here: 
\url{https://github.com/ebjamieson97/undidPyjl}. Work is in progress to fully document using the software with worked examples.  As mentioned, these software packages estimate the treatment effects using the two stage procedure.  They also estimate cluster robust $P$-values using both a cluster jackknife, and a Randomization Inference routine. Which one of these procedures is preferred depends on many things, such as how many silos there are, and how many are treated \citep{karim2025slides}.

\section{Monte Carlo Design} \label{section: MC}

We conduct a set of Monte Carlo experiments with synthetic data to demonstrate the unbiasedness of the UN-DID and its relative performance compared to the conventional DID estimator. The overall experimental design is depicted in Figure \ref{fig:scheme}, where we apply both conventional and UN-DID regression methods to the same simulated dataset to estimate ATTs and assess their properties.  Note that, with poolable data, we can apply both the conventional and UN-DID methods to estimate the ATT. With unpoolable data, we are not able to run the conventional regression. Because these are synthetic data, they are poolable. Therefore, we can run both methods on the dataset.

Our generated dataset is intended to replicate samples from census data for Ontario and Quebec. To be more precise, we aligned the age and gender distributions to closely resemble those observed in the populations of these provinces in 2001 for individuals aged 65 years and older. We keep the gender constant for all individuals and age them by one year for the next nine years, resulting in our final dataset. Our outcome is a continuous variable. In our initial simulation, we created panels of 50 individuals each in both Ontario and Quebec, covering the period from 2000 to 2009. This results in a total of 1000 observations evenly distributed over 10 years. We chose a small sample size to ensure the robustness of the results even when dealing with small datasets. For the simulation, we assume that Quebec is treated by some arbitrary policy in the year 2005.  This approach ensures a balanced and symmetrical dataset, maintaining an equal number of observations in both pre- and post-treatment periods and in the treatment and control groups. We then repeat the simulations with roughly two thirds of the total observations in Ontario  and one third of the total observations in Quebec. 

\begin{figure}[ht!]
    \centering
    \includegraphics[width=0.75\textwidth]{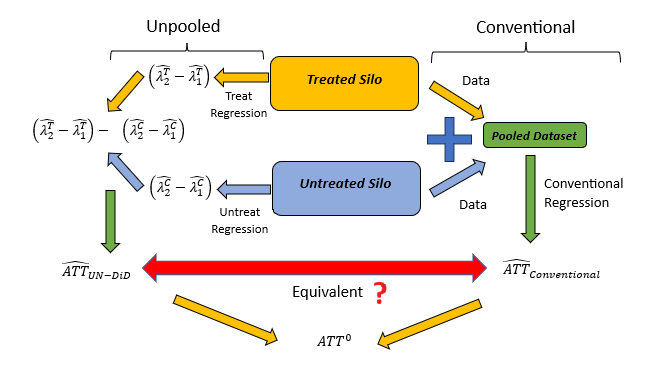}
    \caption{Schematic for Simulations}
    \label{fig:scheme}
\end{figure}

We also consider four types of data-generating processes (DGPs) to assess whether the estimates of the ATT align between conventional  and UN-DID regressions: one without covariates, one with time-invariant covariates, one with time-varying covariates when Assumption \eqref{assumption: stateinvariantccc} holds and one with time-varying covariates where Assumption \eqref{assumption: stateinvariantccc} does not hold. Within each DGP, we analyze two cases. In the first case, we set the true treatment effect for the treated group at 0. In the second case, the true treatment effect for the treated group is set at 0.1. After generating the data, we estimate ATTs using both conventional and UN-DID regressions and assess their unbiasedness properties. This is then repeated 1,000 times for each case. The DGPs for the two cases without covariates are as follows:

\begin{itemize}
	\item \textbf{Case 1}: No treatment effect, no covariates:\vspace{-20pt}

	\begin{equation*}
		Y^s_{it} = e^s_{it} \; \mbox{where:} \; s \in \{Treat,Control\}.
	\end{equation*}
	
	\item \textbf{Case 2:} Treatment effect, no covariates:\vspace{-20pt}
	
	\begin{equation*}
		Y^s_{it} = 0.1 * DID^s_{it} + e^s_{it} \; \mbox{where:} \; s \in \{Treat,Control\}.
	\end{equation*}
 \end{itemize}

In all the simulations in the paper, the error term is idiosyncratic and follows a normal distribution, denoted as $e \sim \mathcal{N}(0, 1)$. The $DID^s_{it}$ represents the interaction term between $P_t$ and $D_s$, corresponding to the post and treat dummy variables.  The coefficient associated with this interaction term signifies the ``true" effect of the treatment in both cases. 

With time-invariant covariates, we introduce a modification to our DGP, where the outcome is now influenced by a time-invariant covariate, denoted by $X_i$. In our simulations, the time-invariant covariate is represented by a binary variable indicating gender. Specifically, it assumes a value of 1 if the individual is female, and 0 otherwise. Additionally, we assign a common slope parameter of 0.5 to the covariate in both the treated and the control silos. As proven in Theorem \ref{theorem: theorem11}, the ATTs from both methods are equivalent regardless of whether Assumption \eqref{assumption: stateinvariantccc} holds with a balanced panel. Therefore, we just explore one scenario.

\begin{itemize}
	\item \textbf{Case 1}: No treatment effect, time-invariant covariate \vspace{-20pt}

	\begin{equation*}
		Y^s_{it} = 0.5 X^s_{i} + e^s_{it} \; \mbox{where:} \; s \in \{Treat,Control\}.
	\end{equation*}
	
	\item \textbf{Case 2:} Treatment effect, time-invariant covariate \vspace{-20pt}
	
	\begin{equation*}
		Y^s_{it} = 0.1 * DID^s_{it} + 0.5 X^s_{i} + e^s_{it} \; \mbox{where:} \; s \in \{Treat,Control\}.
	\end{equation*}
 \end{itemize}
 
With time-varying covariates, we replace the time-invariant covariate with a time-varying covariate, denoted by $X_{it}$. In our simulations, the time-varying covariate is age, which is unaffected by treatment even though it changes over time.
We consider two types of time-varying covariates: one where Assumption \eqref{assumption: stateinvariantccc} holds, and one where Assumption \eqref{assumption: stateinvariantccc} is violated. 

\begin{itemize}
	\item \textbf{Case 1}: No treatment effect, time-varying covariate, Assumption \eqref{assumption: stateinvariantccc} holds \vspace{-20pt}

	\begin{equation*}
		Y^s_{it} = \varsigma X^s_{it} + e^s_{it} \; \mbox{where:} \; s \in \{Treat,Control\}.
	\end{equation*}
	
	\item \textbf{Case 2:} Treatment effect, time-varying covariate, Assumption \eqref{assumption: stateinvariantccc} violated \vspace{-20pt}
	
	\begin{equation*}
		Y^s_{it} = 0.1 * DID^s_{it} + \varsigma^s X^s_{it} + e^s_{it} \; \mbox{where:} \; s \in \{Treat,Control\}.
	\end{equation*}
 \end{itemize}

When Assumption \eqref{assumption: stateinvariantccc} holds, $\varsigma$ is 1.5 for both silos, and when Assumption \eqref{assumption: stateinvariantccc} is violated, $\varsigma^T = 0.5$ and $\varsigma^C = 2$. 

In this simulation study, the generated dataset is a balanced panel. The results displayed in the following sections also apply to repeated cross-sectional data. Furthermore, the results for time invariant covariate with repeated cross-sectional data are comparable to the results with time-varying covariates with panel data, as the distribution of the covariate varies over time due to different individuals being present in the data in each year. After generating the datasets and estimating both the conventional and the UN-DID ATTs, we create kernel density plots to evaluate the bias in both estimators. 

\subsection{Results} \label{sec:results}

We now present the results for time-varying covariates. Additional Monte Carlo results for simpler DGPs are shown in Appendix \ref{sec:moremc}.

In this section, we examine two DGPs. Assumptions \eqref{assumption: binary}, \eqref{assumption: spt} and \eqref{assumption: na} are satisfied for both DGPs, whereas Assumption \eqref{assumption: stateinvariantccc} holds in the first DGP but is violated in the second. Theorem \ref{theorem: theorem4} demonstrates that when Assumption \eqref{assumption: stateinvariantccc} holds, the two population parameters are equivalent. Therefore, both estimators are unbiased when Assumption \eqref{assumption: stateinvariantccc} holds. In contrast, Theorem \ref{lemma: equivtimevaryingcov} shows that the two population parameters are not equivalent when Assumption \eqref{assumption: stateinvariantccc} is violated. In this section, we show that the conventional regression is biased when Assumption \eqref{assumption: stateinvariantccc} is violated, whereas the UN-DID estimator is unbiased using a kernel density plot. The results are shown in Figure \ref{fig: noCCC}. In Panel (a) and Panel (c), both the UN-DID and the conventional estimators are centered around the true value of 0 when Assumption \eqref{assumption: stateinvariantccc} holds. However, when Assumption \eqref{assumption: stateinvariantccc} is violated, we observe that the kernel density for the UN-DID (Panel (b)) is centered around the true value of 0, but the kernel density of the conventional estimator (Panel (d)) is centered roughly around -7.5. This shows that the UN-DID estimator is unbiased when Assumption \eqref{assumption: stateinvariantccc} is violated, but the conventional estimator is biased. 

\begin{figure}[ht!]
    \centering
    \includegraphics[width=0.75\textwidth]{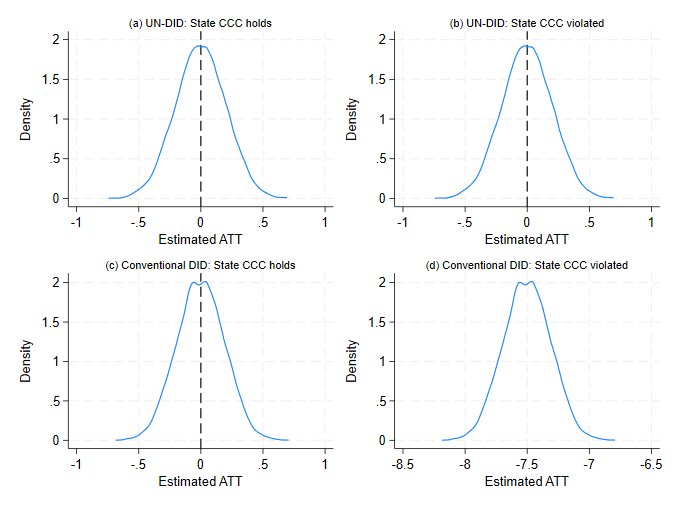}
    \caption{\centering Time-Varying Covariate with Unequal Sample Sizes: Kernel Density}
    \label{fig: noCCC}
\end{figure}

\FloatBarrier

\section{Staggered Adoption Monte Carlo}
\label{sec:mc_stag}

We next evaluate UN-DID in a staggered-adoption setting using individual-level data from the Merged Outgoing Rotation Group of the Current Population Survey (CPS) for 1979--1999, following the placebo-law design in \citet{karim2026improvedinferencecsdidusing}, inspired by \cite{bertrand2004much}. We treat each state as a separate silo. The outcome is the log of women's weekly earnings, excluding observations with weekly earnings below \$20, and the specifications condition on age, age squared, and education indicators. The resulting sample contains 547{,}818 observations.

For each replication, we randomly draw $T=8$ consecutive years and $S$ states. Placebo treatment is assigned to two timing groups: $J$ states are first treated in period 5 and $L$ states are first treated in period 7, with $J=L$. Treatment is absorptive. Because no treatment effect is added to the outcome, the true aggregate ATT is zero. We use this design both to compare the sampling distributions of UN-DID and CSDID and to evaluate cluster-robust inference for UN-DID.

\subsection{ATT Comparison with CSDID}
\label{sec:cps_att}

Figure~\ref{fig:undid_csdid_density} compares the sampling distributions of the aggregate UN-DID and CSDID estimators for $J=L=2$ and $S\in\{8,16,24\}$. For each value of $S$, we use 10{,}000 matched Monte Carlo replications, so that UN-DID and CSDID are estimated using exactly the same CPS observations and placebo treatment assignments.

UN-DID first constructs state-time effects, $ATT(s,t)$, and then aggregates them to obtain the overall ATT. CSDID is implemented using its doubly robust repeated-cross-section estimator and its standard simple aggregation of cohort-time effects, $ATT(g,t)$. Thus, the comparison uses CSDID as it would ordinarily be implemented. The covariates in this example are time-varying by nature of the CPS being a repeated cross section.

\begin{figure}[htbp]
\centering
\includegraphics[width=\textwidth]{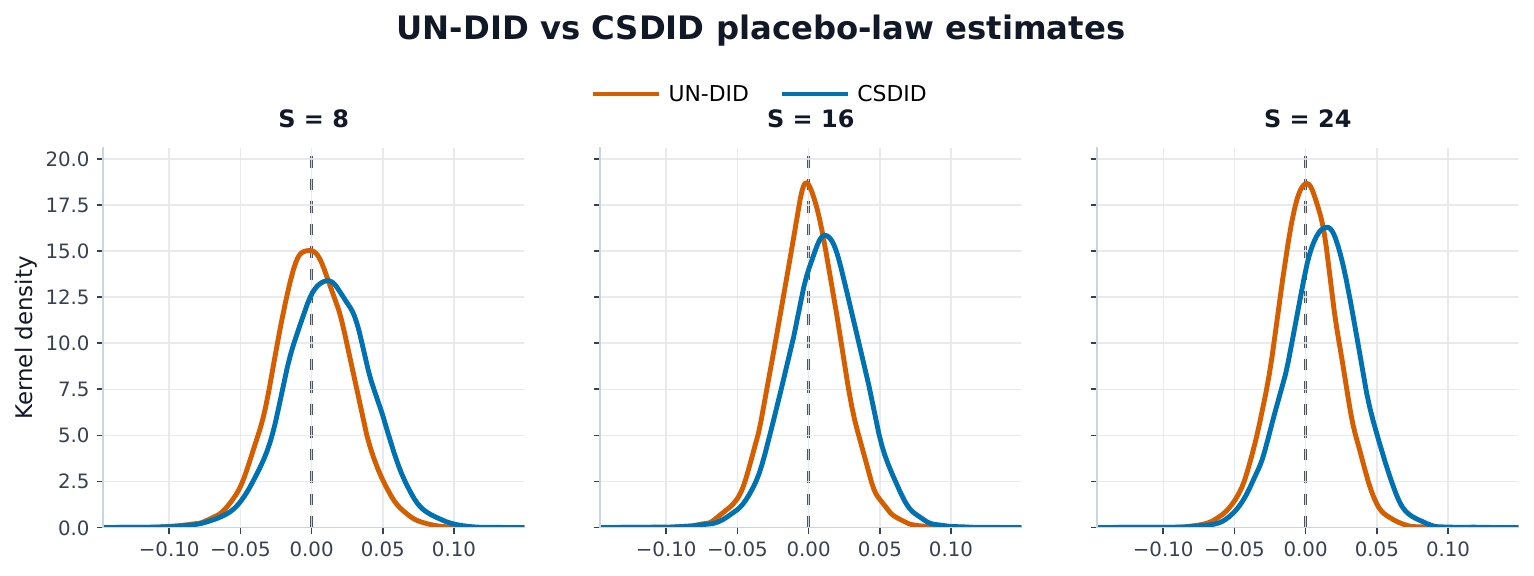}
\caption{UN-DID and CSDID placebo-law estimates}
\label{fig:undid_csdid_density}
\end{figure}

The two estimators produce broadly similar sampling distributions across all three designs, although their centers and dispersion differ modestly. UN-DID is more concentrated and remains centered very close to zero, while the CSDID distribution is slightly shifted to the right, particularly for $S=16$ and $S=24$. Overall, Figure~\ref{fig:undid_csdid_density} shows that aggregate UN-DID has sampling behavior broadly similar to that of CSDID in this placebo-law design, while also exhibiting modest finite-sample differences in location and precision.

\subsection{Cluster-Robust Inference}
\label{sec:undid_inference_results}

Our inference experiments consider
$S\in\{8,12,16,20,24,28,32,36,40\}$. For each value of $S$, we set
$J=L\in\{1,\ldots,S/2-1\}$, yielding 99 designs that range from one
treated state in each treatment-timing group to only two comparison states.
We use 50{,}000 replications for each design. Figure~\ref{fig:undid_att_rejection}
displays results for six representative values of $S$.

The jackknife tests are two-sided tests based on
Equation~\eqref{eq:jk_var}. Randomization inference holds fixed the
adoption dates and the number of states assigned to each treatment-timing
group and reassigns states across the admissible treatment configurations.
Let $M$ denote the number of admissible alternative assignments. When
$M<99$, we enumerate all assignments and use the kernel-smoothed
randomization $p$-value of \citet{racine2007inference}. When
$99\leq M<999$, we enumerate all assignments and compute
\[
p=\frac{1+\sum_{m=1}^{M}
\mathbf{1}\{|\widehat{\tau}^{(m)}|\geq|\widehat{\tau}^{obs}|\}}
{1+M}.
\]
When $M\geq999$, we randomly draw 999 admissible alternative assignments
and use the same finite-sample correction with the number of valid
reassignments in the denominator. Both the jackknife and randomization
inference procedures test at the five-percent level.

\begin{center}
\begin{minipage}{\textwidth}

\include{cps_fig}
\captionsetup{skip=2pt}
\vspace{-30pt}
\captionof{figure}{Finite-Sample Size of UN-DID Inference Procedures.}
\label{fig:undid_att_rejection}

\end{minipage}
\end{center}

Figure~\ref{fig:undid_att_rejection} reports null rejection frequencies for
the aggregate $\widehat{ATT}$. Jackknife size is U-shaped in $J=L$. When each
treatment-timing group contains only one treated state, rejection frequencies
range from 12.0 to 16.8 percent across the displayed values of $S$; with two
treated states per group, they range from 7.8 to 10.7 percent. Rejection
frequencies fall toward the nominal level over intermediate allocations but
rise again as $J=L$ approaches $S/2-1$ and only two comparison states remain.
At that boundary, rejection frequencies range from 9.0 to 13.4 percent. Thus,
increasing the total number of states does not by itself eliminate the size
distortion. Reliable jackknife inference requires adequate treated-state and
comparison-state support. Similar jackknife rejection frequencies for CSDID
are found in \cite{karim2026improvedinferencecsdidusing}.

Randomization inference remains close to nominal size throughout the full
99-design grid, with rejection frequencies between 4.4 and 5.2 percent. This
stability extends to both boundaries: designs with only one treated state in
each treatment-timing group and designs with only two comparison states. In
contrast, inference for more disaggregated targets can be affected by the
small number of admissible assignments; \citet{karim2025slides} studies
inference for cohort-, state-, and policy-level effects in greater detail. The two procedures also test different null hypotheses. Randomization
inference is a design-based test of a sharp null, whereas the cluster
jackknife relies on a large-sample approximation. Their rejection frequencies
should therefore not be interpreted as comparisons between interchangeable
inferential procedures.

\section{Empirical Examples} \label{sec:examples}

We now consider two empirical examples, one with common adoption timing and one with staggered adoption. We 
include these examples to highlight the properties of the UN-DID and conventional DID methods. In contrast to our Monte Carlo simulations, we do not know the true effects of the policies we analyze here. Using survey data that is pooled across US states, we can estimate the ATT using both conventional DID and UN-DID (treating the data as if it were siloed by state), and then compare the ATTs and SEs estimated using these two approaches. Our purpose here is to compare the two methods with real-world data - not to answer the substantive questions underlying the specific regressions.  Accordingly, we diverge somewhat from the designs in the original analyses, and our analysis falls well short of complete empirical analyses of the substantive treatment effects.   

\subsection{The Impact of Medical Marijuana Legalization on Body Mass Index}
We revisit the question asked in \cite{sabia2017effect} of how access to medical marijuana impacts body mass index.  The original study investigates the role of medical marijuana laws in 21 US states and finds that access to medical marijuana decreases the likelihood of being obese.  To assess the performance of UN-DID in a 2$\times$T setting with common treatment timing, we explore the impact of the District of Columbia's (DC) medical marijuana law, using New York State as the control.  

Data come from the Behavioral Risk Factor Surveillance System (BRFSS). This repeated cross-sectional survey is pooled across US states but can also be analyzed as if the data were siloed by state.  We use data from 2006--2013 from New York and DC, which enacted and enforced their medical marijuana law in 2010. We therefore have a $2 \times 8$ design with an equal number of pre-treatment and post-treatment years.  We restrict our attention to respondents aged 25 or older, resulting in a final sample of 84,783 observations.

With these data, we estimate treatment effects using four different models: a conventional DID regression and UN-DID, each with and without controls.  The control variables are a binary variable female and categorical variables for age, race, and education, all of which are time-varying covariates in a repeated cross-sectional dataset.  The UN-DID estimates are calculated by artificially siloing the data by state.  The conventional DID model is shown in Equation \eqref{eq:maryj}.    
{\footnotesize
\begin{equation}
    \label{eq:maryj}
    \text{BMI}_{ist} =  \beta_0 + \beta_1\text{post}_{t}  + \beta_2\text{treat}_{s} + \beta_3\text{did}_{st} + \delta_1\text{female}_{ist} + \delta_2 \text{RACE}_{ist} + \delta_3 \text{AGE}_{ist} + \delta_4 \text{EDUC}_{ist} +  \epsilon_{ist}.
\end{equation}
}
We also calculate the UN-DID ATT by first estimating the following model separately for New York and DC. 
{\footnotesize
\begin{equation}
    \label{eq:maryjundid}
    \text{BMI}_{ist} =  \beta_1 \text{pre}_{t} + \beta_2\text{post}_{t} + \delta_1\text{female}_{ist} + \delta_2 \text{RACE}_{ist} + \delta_3 \text{AGE}_{ist} + \delta_4 \text{EDUC}_{ist} +  \epsilon_{ist}.
\end{equation}
}
Table \ref{tab:maryj} presents the ATT and standard error estimates.  As expected, without covariates the conventional DID and UN-DID yield the same estimates: an estimated ATT of 0.2560 and a hetero-robust standard error of 0.0853. When we add covariates, the estimated coefficient drops to 0.1322 for the conventional DID and 0.1204 for UN-DID. These estimates are neither substantively nor statistically different from one another. The small difference between them is expected given that the models include 25 covariates which are restricted to common slopes with conventional DID and allowed to vary as silo-specific slopes with UN-DID.  The standard errors are also slightly different, but substantively equivalent: 0.0807 for the conventional DID and 0.0811 for UN-DID.

This empirical example demonstrates the strong performance of UN-DID relative to conventional DID in a 2$\times$T setting with common or single treatment timing. We reiterate that our effect estimates in this simplified example should not be interpreted as responding to the substantive question. However, this analysis using real-world, artificially siloed data demonstrates that UN-DID can accurately estimate the ATT and SE as well as conventional DID, both with and without covariates.

\begin{table} 
\begin{center}
    \caption{Effect of DC's Medical Marijuana Legalization on BMI}

\begin{tabular}{lcccc}
\hline
             & Conventional DID   & UN-DID  & Conventional DID & UN-DID \\
\hline
\hline
Coefficient      & 0.2560          & 0.2560 & 0.1322       & 0.1204 \\
Standard error     & 0.0853          & 0.0853 & 0.0807       & 0.0811 \\
\hline
Covariates   &       No        &  No    &  Yes         & Yes   \\
\hline
\end{tabular}
\label{tab:maryj}
\end{center}

\end{table}

\subsection{The Impact of Merit Scholarships on College Attendance}
To highlight the use of UN-DID in a staggered adoption setting, we revisit the analysis of merit scholarships from \cite{conley_2011}.  These authors estimate the average impact of 10 state-level merit scholarship programs with different adoption dates. These scholarships vary in their design, generally awarding scholarships to high-achieving high school students to enroll in a college within their home state (see \cite{deming_2010} and the references therein for details).  

Data come from the Current Population Survey (CPS). There are 51 ``states" (including DC) and 12 years (1989-2000) in the dataset, for a total of 42,161 individual-level observations. We estimate the ATT of a state's adoption of a merit scholarship program on the likelihood of a student in that state attending college.  We first estimate this effect individually for each of the 10 treated states, and then combine those for an aggregate ATT.  This allows us to compare 11 distinct treatment effect estimates using both UN-DID and conventional DID.  

Ten states adopt a merit scholarship across varying adoption dates. Four states adopt the scholarship in 1991, 1993, 1996, and 1999,  respectively, while two states each adopt in 1997, 1998, and 2000. For each treated state, one control state was chosen from the set of  41 untreated states to roughly match the number of treated and untreated observations in each model. We use only one treated state to highlight the properties of the UN-DID estimator rather than to estimate the best possible counterfactual. We estimate the treatment effects both with and without covariates.  Specifically, using only data from treated state $T$ and control state $C$, we estimate the model:
\begin{equation}
    \label{eq:merit}
    \text{college}_{ist} =  \beta_0 + \beta_1\text{post}_{t}  + \beta_2\text{treat}_{s} + \beta_3\text{did}_{st} + \delta_1\text{black}_{ist} + \delta_2 \text{asian}_{ist} + \delta_3\text{male}_{ist}  +  \epsilon_{ist}.
\end{equation}
The coefficients $\delta_1, \delta_2$, and $\delta_3$ are estimated only when the covariates are included.  The remaining
variables are defined the standard way: $treat_s$ is set equal to 1 for observations in state $T$ and is set equal to 0 for observations in state $C$. Similarly, $post_t$ is set equal to 1 for years greater than or equal to the year when state $T$ first offered a merit scholarship. Finally, $did_{st}$ is the product of   $treat_s$ and  $post_t$. We estimate treatment effects, or $\widehat{\beta}_{3}$, and standard errors using the pooled dataset. We then artificially silo the data to estimate treatment effects and standard errors using UN-DID.  Specifically, for each treated  silo we estimate 
\begin{equation}
    \label{eq:meritun}
    \text{college}_{ist} =  \beta_1\text{pre}_{st}  + \beta_2\text{post}_{st} + \delta_1\text{black}_{ist} + \delta_2 \text{asian}_{ist} + \delta_3\text{male}_{ist}  +  \epsilon_{ist} \text{ for } s \in \{T,C\}.
\end{equation}
Here, $\text{pre}_{t}$ is equal to one in the years before the merit program was adopted, and equal to zero afterwards.  Conversely, $\text{post}_{t}$ is equal to zero in the years before the merit program was adopted, and equal to one afterwards.  We estimate the same regression for the control state used for each treated state.  We then calculate the treatment effect using Equation \eqref{equation: undidattmain}. 

The results are found in Table \ref{tab:merit}. Panel A shows the estimates for the models without covariates.  The treatment effect estimates are identical to the fourth decimal place, and the standard error estimates differ by at most 0.0001. These results match the simulation results presented in Appendix \ref{sec:nocovar}.

Panel B of Table \ref{tab:merit} shows the treatment effect estimates conditional on covariates.  Here, some of the individual state treatment effect estimates are very similar using conventional DID and UN-DID.  For instance, for state 88, the conventional DID effect estimate is 0.1811 and the UN-DID estimate is 0.1819.  However, some of the estimates are not as close. The estimated treatment effect for state 61 is 0.0370 using conventional DID and 0.0314 with UN-DID. Even though the difference between the estimates is larger, they are not statistically different and may not be substantively different from each other. The estimates of the hetero-robust standard errors differ
more between conventional DID and UN-DID. Again, these differences are usually not large enough to impact inference in this example, with state 85 being an exception.  These findings are consistent with the results in Section \ref{section: MC}.

Lastly, we estimate average treatment effects using data from  all groups and all necessary years per group.  We do this 
using two methods, both including covariates. We estimate the ATT using the \texttt{CSDID} command, which implements the \cite{callaway2021difference} procedure. We again pretend that each state's data are siloed, and estimate the ATT using UN-DID.  We do this using our \texttt{UNDID} command, which excludes forbidden comparisons. For both CSDID and UN-DID we estimate one ATT by taking a weighted average of the ATT$({s,t})$ (or ATT$({g,t})$) and one by averaging over the ATT$s$ (or ATT$g$) terms. Additionally, we calculate cluster-robust standard errors, assuming state level clustering, for all ATT estimates.  The UN-DID standard errors use the jackknife discussed in Section \ref{sec:cluster}, and the CSDID uses the jackknife proposed in \citet{karim2026improvedinferencecsdidusing}.

The aggregate ATT estimates are shown in Panel C of Table \ref{tab:merit}.  We find the simple average ATT estimates from the two procedures to be quite similar.  Using simple weights, the UN-DID ATT is 0.0485 and the CSDID ATT is 0.0464.  The two standard errors are also very close to one another, at 0.0110 for UN-DID and 0.0133 for CSDID. The group-averaged ATTs are not as similar to one another, but still close.  The UN-DID ATT is 0.0459, remaining close to the simple average ATT. The CSDID standard errors are cluster-jackknife standard errors using the \texttt{Stata} package \texttt{csdidjack}. Using group weights, the UNDID ATT is 0.0459 while the  CSDID ATT drops to 0.0339.  The standard errors are also similar, with the UNDID standard error equal to 0.0188 and the CSDID one equal to 0.0211.  Note that the CSDID ATT is no longer significant at the 10\% level.  

Finally, we use the DID-INT estimator to estimate ATTs.  These results are very comparable to the UN-DID results.  With group aggregation, the UN-DID estimate is 0.0459 and the DID-INT estimate (allowing for state level CCC violations) is 0.0458.  The standard errors are more different, where the DID-INT standard error is 0.0084 compared to the UN-DID 0.0188. Results are very similar to those with simple aggregation. We also investigate the sensitivity to the remaining CCC assumptions.  These results can be found in Appendix \ref{sec:merit_ccc}.  Generally, the ATTs broadly agree with one another, with the smallest being 0.041 when allowing for time-varying parameters, and the largest being 0.0511 when allowing for two-way varying parameters. 

\begin{table} 
\renewcommand{\arraystretch}{1.4}
\caption{Estimates from Merit Example}
\begin{tabular}{lclrrcc}
 \multicolumn{3}{l}{Panel A: Without Covariates}   & & & &        \\
\hline 
Treated & Control & Year &     &    &  &  \\
State & State & Treated & DID    & UN-DID   & DID SE & UN-DID SE \\
\hline
71      & 73      & 1991 & $-0$.0242 & $-0$.0242 & 0.0755  & 0.0754  \\
58      & 46      & 1993 & 0.1386  & 0.1386  & 0.0590  & 0.0590  \\
64      & 54      & 1996 & 0.0182  & 0.0182  & 0.0619  & 0.0619  \\
59      & 23      & 1997 & $-0$.0073 & $-0$.0073 & 0.0362  & 0.0362  \\
85      & 86      & 1997 & 0.1318  & 0.1318  & 0.0601  & 0.0601  \\
57      & 32      & 1998 & $-0$.0291 & $-0$.0291 & 0.0747  & 0.0747  \\
72      & 55      & 1998 & $-0$.0645 & $-0$.0645 & 0.0658  & 0.0658  \\
61      & 62      & 1999 & 0.0335  & 0.0335  & 0.0816  & 0.0816  \\
34      & 33      & 2000 & $-0$.0305 & $-0$.0305 & 0.0666  & 0.0666  \\
88      & 11      & 2000 & 0.1595  & 0.1595  & 0.1337  & 0.1337  \\
\hline
 \multicolumn{3}{l}{Panel B: With Covariates}   & & & &        \\
 \hline
 Treated & Control & Year &     &    &  &  \\
State & State & Treated & DID    & UN-DID   & DID SE & UN-DID SE \\
\hline
71      & 73      & 1991 & $-0$.0141 & $-0$.0137 & 0.0750  & 0.0871  \\
58      & 46      & 1993 & 0.1501  & 0.1523  & 0.0582  & 0.0734  \\
64      & 54      & 1996 & 0.0077  & 0.0028  & 0.0614  & 0.0820  \\
59      & 23      & 1997 & $-0$.0179 & $-0$.0139 & 0.0360  & 0.0441  \\
85      & 86      & 1997 & 0.1309  & 0.1277  & 0.0597  & 0.0726  \\
57      & 32      & 1998 & $-0$.0480 & $-0$.0417 & 0.0733  & 0.0876  \\
72      & 55      & 1998 & $-0$.0704 & $-0$.0685 & 0.0655  & 0.0793  \\
61      & 62      & 1999 & 0.0370  & 0.0314  & 0.0806  & 0.0918  \\
34      & 33      & 2000 & $-0$.0282 & $-0$.0303 & 0.0655  & 0.0709  \\
88      & 11      & 2000 & 0.1811  & 0.1819  & 0.1335  & 0.1428 \\
\hline
 \multicolumn{3}{l}{Panel C: Full Sample}   & & & &        \\
 \hline
 & UN-DID & UN-DID & CSDID & CSDID & DID-INT & DID-INT \\
Agg.     & ATT    & SE     & ATT   & SE    & ATT     & SE      \\
simple & 0.0485       & 0.0110       & 0.0464     & 0.0133 & 0.0464   & 0.0102    \\
group  & 0.0459     & 0.0188       & 0.0339        & 0.0211 & 0.0458  & 0.0084  \\

 \hline
 \hline
\end{tabular}
\label{tab:merit}
\end{table}

\section{Conclusion} \label{sec:concl}

Difference-in-differences is commonly used to estimate treatment effects, particularly for policies or other interventions that vary at the jurisdictional (e.g., state, country) level. However, applying standard pooled DID methods can be difficult or impossible when data are siloed by jurisdiction and the data-governance regime restricts the information that can be shared across silos. In this paper, we propose the unpoolable DID (UN-DID) method, assess its performance, and demonstrate its utility in settings with siloed data. We also formalize data poolability as a data-access condition in the DID context. This is an implementation feature that has generally been taken for granted rather than explicitly articulated. We show that, under the stated identifying assumptions, UN-DID provides an unbiased regression-based approach for estimating the ATT when the standard pooled implementation is unavailable.

We begin by presenting analytical proofs demonstrating the equivalence of UN-DID and conventional DID with no covariates. With covariates, we show that the two estimators converge in probability to the ATT, provided that the state-invariant CCC assumption holds. When the state-invariant CCC assumption is violated, however, the conventional regression that imposes a common coefficient on the covariates is biased and does not converge to the ATT. In addition, we show that the two estimators are not algebraically equivalent in finite samples, except when we use time-invariant covariates and have access to panel data.  

We then extended our assessment to simulated panel data using a series of Monte Carlo simulations (incorporating a range of DGPs and sample sizes). These experiments permitted a direct assessment of UN-DID’s performance and unbiasedness compared to both conventional DID and the true treatment effects. Across all assessments, we find that the estimated ATT from UN-DID remains unbiased. However, the conventional estimator is biased when time-varying covariates are used and the state-invariant CCC is violated. The estimated standard errors are also not numerically equal, but converge to the true value as the sample size increases. These findings suggest that UN-DID is a valid approach to DID when data are unpoolable. Finally, we compared the performance of UN-DID and conventional DID using pooled data from existing empirical examples, treating these data as unpoolable. 

Our work contributes to a growing body of literature developing DID methods for use in increasingly complex policy and data environments. UN-DID offers an important methodological tool for researchers assessing multi-jurisdictional (or cross-jurisdictional) policy effects. By enabling the use of DID estimation in siloed data settings, UN-DID effectively opens the door to new research questions and new evidence on policy impacts. 

The UN-DID software packages in  \texttt{Stata}, \texttt{R}, \texttt{Python}, and \texttt{Julia} will facilitate the practical application of UN-DID by a range of different end-users; the accompanying user guide will further support this aim while additionally showcasing the use of UN-DID in a truly siloed setting. Taken together, we expect our development of UN-DID (both as an analytical method and associated software packages) to catalyze innovative cross-jurisdictional policy evaluation in siloed data settings, thereby eliminating a common barrier for an otherwise widely used and powerful method.

\bibliography{mybib.bib}

\bibliographystyle{econ}

\clearpage

\appendix 
\setcounter{section}{0} 
\setcounter{equation}{0}  
\setcounter{assumption}{0}  

\numberwithin{equation}{section}
\renewcommand{\theequation}{\thesection\arabic{equation}}

\counterwithin{assumption}{section}
\renewcommand{\theassumption}{\Alph{section}\arabic{assumption}}

\section{Alternative Covariate-Adjusted DID Estimators}
\label{appendix:matching}

In this section, we discuss three alternative approaches for estimating the ATT in the presence of covariates: Inverse Probability Weighting (IPW), Regression Adjustment (RA), and Doubly Robust Difference-in-Differences (DR-DID). In particular, we explain why their standard implementations are unavailable under the data-access restrictions considered here, where Data Assumption \eqref{assumption:pooled} does not hold and Data Assumption \eqref{assumption: noentry} does hold. IPW is a semiparametric approach that reweights observations in the control group according to their propensity scores so that the weighted covariate distribution more closely resembles that of the treated group \citep{abadie2005semiparametric}. The propensity score gives the probability of treatment conditional on observed covariates and therefore must be estimated using information on treatment status and covariates. However, the standard pooled-data implementation of IPW requires jointly accessible information on treatment status and covariates from the treated and control groups in order to estimate the propensity score. Consequently, this standard implementation is unavailable when Data Assumption \eqref{assumption:pooled} is violated. DR-DID combines propensity-score weighting with outcome regression \citep{sant2020doubly}. Its standard implementation likewise requires cross-silo information that is unavailable under the data-access restrictions considered here.

We next consider RA in greater detail. In the canonical two-group and two-period repeated-cross-sectional setting, RA models untreated outcomes using observations from the control group. Since the control group remains untreated in both periods, its observed outcomes coincide with its untreated potential outcomes. In particular, separate outcome regressions are estimated for the control group in the pre-treatment and post-treatment periods and subsequently fitted values are estimated using the covariate values of treated observations. For simplicity, consider the following linear outcome regressions:

\begin{equation}
\label{equation: RAfit}
Y^{Control,p}_{i,t}
=
\varrho^{Control,p}_0
+
\varrho^{Control,p}_1 X_{i,t}
+
\varepsilon^{Control,p}_{i,t},
\qquad
p \in \{pre,post\}.
\end{equation}

The conditional outcome functions may alternatively be estimated using other parametric or nonparametric specifications. Under the linear specification in Equation \eqref{equation: RAfit}, define the fitted conditional outcome for the control group in period $p$ as

\begin{equation}
\label{equation: RAfitvalue}
\widehat{m}_{Control,p}(X_{i,t})
=
\widehat{\varrho}^{Control,p}_0
+
\widehat{\varrho}^{Control,p}_1 X_{i,t},
\qquad
p \in \{pre,post\}.
\end{equation}

The fitted conditional outcomes in Equation \eqref{equation: RAfitvalue} are evaluated at the covariate values of treated observations. Because RA does not impose the common covariate specification used by the conventional pooled regression, it can accommodate violations of the state-invariant CCC assumption defined in Section \ref{sec3:ccc}, provided that the relevant identifying assumptions hold and the outcome regressions are correctly specified. For repeated cross-sectional data, the RA estimator can be written as:

\begin{gather}
\label{equation: RARC}
\widehat{ATT}_{RA}
=
\frac{1}{N^{Treat,post}}
\sum_{i:S_i=Treat,\,P_t=post}
\left[
Y^{Treat,post}_{i,t}
-
\widehat{m}_{Control,post}(X_{i,t})
\right]
\\ -
\frac{1}{N^{Treat,pre}}
\sum_{i:S_i=Treat,\,P_t=pre}
\left[
Y^{Treat,pre}_{i,t}
-
\widehat{m}_{Control,pre}(X_{i,t})
\right] \nonumber.
\end{gather}

In Equation \eqref{equation: RARC}, $N^{Treat,post}$ and $N^{Treat,pre}$ denote the numbers of treated-group observations in the post-treatment and pre-treatment periods, respectively. The control-group outcome regressions are evaluated at the covariate values of treated observations in the corresponding period. The estimator therefore compares the average difference between the observed treated outcome and its fitted control-group counterpart across the post-treatment and pre-treatment periods.

When Data Assumption \eqref{assumption:pooled} is violated, the control-group outcome regressions can still be estimated locally within the control silo. Implementing standard RA, however, requires evaluating these regressions at the covariate values of treated observations. This would require either transferring the estimated control-silo regression coefficients into the treated silo or transferring treated-unit covariates into the control silo. Both operations are ruled out by Data Assumption \eqref{assumption: noentry}, which prohibits information from being brought into the silos. Consequently, the standard repeated-cross-sectional RA implementation is unavailable in the data environment considered here.  Moreover, RA can struggle with time-varying covariates.

With balanced panel data, the individual-level change in outcomes is observed directly. RA can therefore be implemented by modeling the conditional change in untreated outcomes using observations from the control group. For simplicity, consider the following linear regression:

\begin{equation}
\label{equation: RAdeltafit}
Y^{Control,post}_{i,t}
-
Y^{Control,pre}_{i,t}
=
\varrho^{Control,\Delta}_0
+
\varrho^{Control,\Delta}_1 X_i
+
\varepsilon^{Control,\Delta}_{i,t}.
\end{equation}

Define the corresponding fitted conditional change in untreated outcomes as

\begin{equation}
\label{equation: RAdelta}
\widehat{m}_{Control,\Delta}(X_i)
=
\widehat{\varrho}^{Control,\Delta}_0
+
\widehat{\varrho}^{Control,\Delta}_1 X_i.
\end{equation}

The panel-data RA estimator can then be written as

\begin{equation}
\label{equation: RA}
\widehat{ATT}_{RA}
=
\frac{1}{N^{Treat}}
\sum_{i:S_i=Treat}
\left[
\left(
Y^{Treat,post}_{i,t}
-
Y^{Treat,pre}_{i,t}
\right)
-
\widehat{m}_{Control,\Delta}(X_i)
\right].
\end{equation}

As in the repeated cross-sectional setting, implementation requires evaluating a model estimated in the control silo at the covariate values of treated observations. The RA can be feasible with unpoolable data using exported regression coefficients and treated group summary statistics. With multiple control silos, constructing the point estimate may become more cumbersome, since we must obtain the relevant summary statistics from each silo and aggregate them appropriately. However, valid inference may be challenging, since we require additional information such as the variance-covariance information for the exported quantities within each silo. This additional information may not be available under restrictive data-sharing restrictions. 

\section{Proof of Theorem \eqref{lemma: equalATT}} 
\label{appendix: nocovarpopequivalence}

\begin{proof}[Proof: Theorem \eqref{lemma: equalATT}]
Taking conditional expectations on both sides of Equation \eqref{equation: didpopulationnocov} for each silo and period, we get the following equations. 
For the treated group in the post-intervention period:
\begin{equation}
\label{equation: theorem1first} \mathbb{E}[Y_{i,s,t}|D_s = 1, P_t = 1] = \beta_0^* + \beta_1^* + \beta_2^* + \beta_3^* 
\end{equation}
For the treated group in the pre-intervention period:
\begin{equation}
\label{equation: theorem1second} \mathbb{E}[Y_{i,s,t}|D_s = 1, P_t = 0] = \beta_0^* + \beta_1^* 
\end{equation}
For the control group in the post-intervention period:
\begin{equation}
\label{equation: theorem1third} \mathbb{E}[Y_{i,s,t}|D_s = 0, P_t = 1] = \beta_0^* + \beta_2^*  
\end{equation}
For the control group in the pre-intervention period:
\begin{equation}
\label{equation: theorem1fourth} \mathbb{E}[Y_{i,s,t}|D_s = 0, P_t = 0] = \beta_0^*  
\end{equation}
Taking a difference between Equations \eqref{equation: theorem1first} and \eqref{equation: theorem1second}, and a difference between Equations \eqref{equation: theorem1third} and \eqref{equation: theorem1fourth}, and subsequently taking a difference of these equations and rearranging, we get:
\begin{equation}
\label{equation: theorem1pooled}
\begin{aligned}
\beta_3^* = &
\Big[\mathbb{E}\left[Y_{i,s,t}\mid D_s=1,P_t=1\right] -
\mathbb{E}\left[Y_{i,s,t}\mid D_s=1,P_t=0\right]\Big]
\\\quad- & \Big[\mathbb{E}\left[Y_{i,s,t}\mid D_s=0,P_t=1\right] - \mathbb{E}\!\left[Y_{i,s,t}\mid D_s=0,P_t=0\right]\Big]
\end{aligned}
\end{equation}
In comparing Equations \eqref{equation: theorem1pooled} and \eqref{equation: attuncond}, we see that their right-hand sides are identical. Therefore,
\begin{equation}
    \beta_3^* = ATT
\end{equation}
To demonstrate that the UN-DID identifies the ATT, we derive the estimator's probability limit. By taking conditional expectations on both sides of Equation \eqref{equation: undidpopulationnocov} for each period, we obtain the subsequent equations. For the treated group in the post-intervention period:
\begin{equation}
\label{equation: theorem1undid1}
\mathbb{E}[Y^{Treat}_{i,t}|P_t=1] = 
\lambda_2^{Treat*}.
\end{equation}
For the treated group in the pre-intervention period:
\begin{equation}
\label{equation: theorem1undid2}
\mathbb{E}[Y^{Treat}_{i,t}|P_t=0] = 
\lambda_1^{Treat*}.
\end{equation}
For the control group in the post-intervention period:
\begin{equation}
\label{equation: theorem1undid3}
\mathbb{E}[Y^{Control}_{i,t}|P_t=1] = 
\lambda_2^{Control*}.
\end{equation}
For the control group in the pre-intervention period:
\begin{equation}
\label{equation: theorem1undid4}
\mathbb{E}[Y^{Control}_{i,t}|P_t=0] = 
\lambda_1^{Control*}.
\end{equation}
Assuming that the conventional DID and UN-DID estimators are computed using the same sample, the corresponding conditional expectations are identical:
\begin{equation}
\label{equation: condexptheorem1}
\begin{aligned}
\mathbb{E}[Y_{i,s,t}\mid D_s=1,P_t=1]
&=
\mathbb{E}[Y^{Treat}_{i,t}\mid P_t=1],\\
\mathbb{E}[Y_{i,s,t}\mid D_s=1,P_t=0]
&=
\mathbb{E}[Y^{Treat}_{i,t}\mid P_t=0],\\
\mathbb{E}[Y_{i,s,t}\mid D_s=0,P_t=1]
&=
\mathbb{E}[Y^{Control}_{i,t}\mid P_t=1],\\
\mathbb{E}[Y_{i,s,t}\mid D_s=0,P_t=0]
&=
\mathbb{E}[Y^{Control}_{i,t}\mid P_t=0].
\end{aligned}
\end{equation}
Replacing the conditional expectations with those in Equation \eqref{equation: condexptheorem1}, then subtracting Equations \eqref{equation: theorem1undid1} and \eqref{equation: theorem1undid2}, as well as Equations \eqref{equation: theorem1undid3} and \eqref{equation: theorem1undid4}, and further subtracting and rearranging these results, yields:
\begin{equation}
\label{equation: theorem1unpooled}
\begin{aligned}
\tau_{UN}^* = &
\Big[\mathbb{E}\left[Y_{i,s,t}\mid D_s=1,P_t=1\right] -
\mathbb{E}\left[Y_{i,s,t}\mid D_s=1,P_t=0\right]\Big]
\\\quad- & \Big[\mathbb{E}\left[Y_{i,s,t}\mid D_s=0,P_t=1\right]-\mathbb{E}\!\left[Y_{i,s,t}\mid D_s=0,P_t=0\right]\Big] = ATT
\end{aligned}
\end{equation}
Hence,
\begin{equation}
\beta_3^* = \tau_{UN}^* = ATT.
\end{equation}
\end{proof}

\section[ATT Equivalence (no covariates)]{Without covariates, the ATT estimates from the conventional and the UN-DID estimators are algebraically equivalent}
\label{ssection: Lemma1}

In this section, we show that the UN-DID and the conventional DID recover the same estimate of the ATT in the absence of covariates (Theorem \eqref{theorem: theorem2}). Without covariates, the ATT under strong parallel trends and no anticipation is as follows:

\begin{theorem}[ATT under strong parallel trends and no anticipation]
\label{theorem: attstrongpt}
    {\footnotesize
\begin{equation}
    \label{equation: attstrongpt}
    \begin{gathered}
        ATT = \biggl[E[Y_{i,s,t}|D_{s} = 1,P_t = 1] - E[Y_{i,s,t}|D_{s} = 1, P_{t} = 0]\biggr] - \biggl[E[Y_{i,s,t}|D_{s} = 0, P_{t} = 1] - E[Y_{i,s,t}|D_{s} = 0, P_t = 0]\biggr].
    \end{gathered}
    \end{equation}
}
\end{theorem}

Strong parallel trends is the unconditional parallel trends restriction obtained when covariates are not used, and is stated below:
\begin{assumption}[Strong Parallel Trends]
\label{assumption: cpt}
    The evolution of untreated potential outcomes is the same between treated and control groups.
    \begin{align}
        \label{equation: spt}
        \begin{split}
           & \biggr[E[Y_{i,s,t}(0)|D_s = 1,P_t=1] - E[Y_{i,s,t}(0)|D_s = 1,P_t=0]\biggr] \\
              = & \ \biggr[E[Y_{i,s,t}(0)|D_s = 0,P_t=1] - E[Y_{i,s,t}(0)|D_s = 0, P_t= 0] \biggr].
        \end{split}
    \end{align}
\end{assumption}

In order to prove Theorem \eqref{theorem: theorem2}, we will first prove the following Lemmas:

\begin{lemma}[Estimate of the ATT from the conventional regression]
\label{lemma: conventionalnocov}
\begin{equation}
 \hat{\beta}_3 = \left( \overline{Y}_{D_s = 1, P_t = 1} - \overline{Y}_{D_s = 1, P_t = 0} \right) - \left( \overline{Y}_{D_s = 0, P_t = 1} - \overline{Y}_{D_s = 0, P_t = 0} \right).   
\end{equation}
\end{lemma}
\begin{proof}[Proof: Lemma \eqref{lemma: conventionalnocov}] We begin by re-writing Equation \eqref{equation: didsimple} without the $Pooled$ superscripts. In Equation \eqref{equation: didsimple}, the $Pooled$ superscript indicates that
observations from the treated and control groups are analyzed in a common estimation environment.  
Under Data Assumption \eqref{assumption:pooled}, this regression can be estimated using the standard pooled-data implementation. The algebraic result in Lemma \eqref{lemma: conventionalnocov}, however, requires only Assumption \eqref{assumption: binary}; Data Assumption \eqref{assumption:pooled} concerns whether this pooled implementation is available rather than identification of the ATT. Assumptions \eqref{assumption: spt} and \eqref{assumption: na} are required only when interpreting the resulting double difference as the ATT in Theorem \eqref{lemma: equalATT}.

\begin{equation}
\label{equation: didsimplewithouts}
    Y_{i,s,t} = \beta_0 + \beta_1 D_s + \beta_2 P_t + \beta_3 P_t * D_s + \epsilon_{i,s,t}.
\end{equation}

According to the FWL theorem, $\hat{\beta_3}$ will be the same as the coefficient of $(P_t*D_s - \widehat{P_t*D_s})$ from the regression shown in Equation \eqref{equation: didconventionalFWL}. Here, $(P_t*D_s - \widehat{P_t*D_s})$ are the residuals from the regression shown in Equation \eqref{equation: didconventionalFWLresiduals}.
\begin{equation}
\label{equation: didconventionalFWLresiduals}
    P_t * D_s = \alpha_0 + \alpha_1 P_t + \alpha_2 D_s + u_{i,s,t}.
\end{equation}
\begin{equation}
\label{equation: didconventionalFWL}
    Y_{i,s,t} = \widehat{\beta_3} (P_t*D_s - \widehat{P_t*D_s}) + u'_{i,s,t}. 
\end{equation}

From Equation \eqref{equation: didconventionalFWLresiduals}, we compute the residuals ($P_t * D_s - \widehat{P_t * D_s}$) for each observation. For a treated observation in the pre-intervention period, the residual is given by $- \hat{\alpha_0} - \hat{\alpha_2}$.  Conversely, for a treated observation in the post-intervention period, the residual becomes $1 - \hat{\alpha_0} - \hat{\alpha_1} - \hat{\alpha_2}$. Similarly, an untreated observation in the pre-intervention period has a residual of $- \hat{\alpha_0}$, while an untreated observation in the post-intervention period yields a residual of $- \hat{\alpha_0} - \hat{\alpha_1}$. Since Equation \eqref{equation: didconventionalFWLresiduals} contains only group and period indicators, the residual is constant within each group-period cell. We then proceed to derive $\widehat{\beta}_3$ from Equation \eqref{equation: didconventionalFWL} by using the following OLS formula:
\begin{align}
\label{equation: didconventionalols}
\begin{split}
\widehat{\beta}_3= & \frac{ \sum_i \sum_t Y_{i,s, t}\left(P_t*D_s - \widehat{P_t*D_s}\right)}{\sum_i \sum_t \left(P_t*D_s - \widehat{P_t*D_s}\right)^2}.
\end{split}
\end{align}
We can substitute the computed residuals into the numerator and simplify the expression from Equation \eqref{equation: didconventionalols} as outlined below:
\begin{align*}
    \begin{split}
        & \sum_{i \in \{ t \in pre, d \in {Treat} \}} Y_{i, t}^{pre,{Treat}}\left(-\widehat{\alpha}_0-\widehat{\alpha}_2\right) + \sum_{i \in \{ t \in post, d \in {Treat} \}} Y_{i, t}^{post,{Treat}} \left(1-\widehat{\alpha}_0-\widehat{\alpha}_1-\widehat{\alpha}_2\right) \\ & + \sum_{i \in \{ t \in pre, d \in {Control} \}} Y_{i, t}^{pre,{Control}}\left(- \widehat{\alpha}_0\right) + \sum_{i \in \{ t \in post, d \in {Control} \}} Y^{post,{Control}}_{i, t}\left(-\widehat{\alpha}_0-\widehat{\alpha}_1\right).
    \end{split}
\end{align*}
Likewise, by substituting the residuals and simplifying, the denominator of Equation \eqref{equation: didconventionalols} can be expressed as
\begin{align*}
    \begin{split}
        A = & \sum_{i \in \{ t \in pre, d \in {Treat} \}} \left(-\widehat{\alpha}_0-\widehat{\alpha}_2\right)^2 + \sum_{i \in \{ t \in post, d \in {Treat} \}}\left(1-\widehat{\alpha}_0-\widehat{\alpha}_1-\widehat{\alpha}_2\right)^2 \\ & + \sum_{i \in \{ t \in pre, d \in {Control} \}}\left(-\widehat{\alpha}_0\right)^2 + \sum_{i \in \{ t \in post, d \in {Control} \}}\left(-\widehat{\alpha}_0-\widehat{\alpha}_1\right)^2.
    \end{split}
\end{align*}
The denominator is the sum of squared residuals from the regression in Equation \eqref{equation: didconventionalFWLresiduals}. Combining the numerator and the denominator, $\widehat{\beta}_3$ can be written as
\begin{align}
\label{equation: didconventionalbeta3}
\begin{split}
 \widehat{\beta}_3= & \frac{\left(1-\widehat{\alpha}_0-\widehat{\alpha}_1-\widehat{\alpha}_2\right)\sum_{i \in\{ t \in post, d \in {Treat} \}} Y_{i, t}^{post,{Treat}}}{A} - \frac{\left(\widehat{\alpha}_0+\widehat{\alpha}_2\right)\sum_{i \in \{ t \in pre, d \in Treat \}} Y_{i, t}^{pre,{Treat}}}{A} \\ & - \frac{\left(\widehat{\alpha}_0+\widehat{\alpha}_1\right)\sum_{i \in \{ t \in post, d \in {Control} \}} Y_{i, t}^{post,{Control}}}{A} + \frac{\left(-\widehat{\alpha}_0\right) \sum_{i \in \{ t \in pre, d \in {Control} \}} Y_{i, t}^{pre,{Control}}}{A}.
\end{split}
\end{align}
In Appendix \ref{ssection: samplesizes}, we show that  $\frac{\left(1-\widehat{\alpha}_0-\widehat{\alpha}_1-\widehat{\alpha}_2\right)}{A}$ is the reciprocal of the number of observations in the treated group and the post-intervention period ($N_{1,1})$ and $\frac{\left(\widehat{\alpha}_0+\widehat{\alpha}_2\right)}{A}$ is the reciprocal of the number of observations in the treated group in the pre-intervention period ($N_{1,0}$). Similarly, $\frac{\left(\widehat{\alpha}_0 + \widehat{\alpha}_1\right)}{A}$ is the reciprocal of the number of observations in the control group in the post-intervention period $(N_{0,1})$  and $\frac{\left(-\widehat{\alpha}_0\right)}{A}$ is the reciprocal of the number of observations in the control group in the pre-intervention period ($N_{0,0})$. Therefore, Equation \eqref{equation: didconventionalbeta3} can be further simplified as:

\begin{align*}
\begin{split}
\hat{\beta}_3= \frac{1}{N_{1,1}}\sum_{i \in \{ t \in post, d \in {Treat} \}} Y_{i, t}^{post,{Treat}} - \frac{1}{N_{1,0}} \sum_{i \in \{ t \in pre, d \in {Treat} \}} Y_{i, t}^{pre,{Treat}} \\ - \frac{1}{N_{0,1}} \sum_{i \in \{ t \in post, d \in {Control} \}} Y_{i, t}^{post,{Control}} + \frac{1}{N_{0,0}} \sum_{i \in \{ t \in pre, d \in {Control} \}} Y_{i, t}^{pre,{Control}}.
\end{split}
\end{align*}

In simpler notation, the above expression can be rewritten as: \begin{align}
\label{equation: didconventionalbeta3simpler}
\begin{split}
\hat{\beta}_3 = & \left( \overline{Y}_{D_s = 1, P_t = 1} - \overline{Y}_{D_s = 1, P_t = 0} \right) - \left( \overline{Y}_{D_s = 0, P_t = 1} - \overline{Y}_{D_s = 0, P_t = 0} \right).
\end{split}
\end{align}
\end{proof}

\begin{lemma}
\label{lemma: UNDIDnocov}
\begin{equation}
(\widehat{\lambda}_2^{Treat} - \widehat{\lambda}_1^{Treat}) - (\widehat{\lambda}_2^{Control} - \widehat{\lambda_1}^{Control}) = \left( \overline{Y}_{D_s = 1, P_t = 1} - \overline{Y}_{D_s = 1, P_t = 0} \right) - \left( \overline{Y}_{D_s = 0, P_t = 1} - \overline{Y}_{D_s = 0, P_t = 0} \right).   
\end{equation}
\end{lemma}

\begin{proof}[Proof: Lemma \eqref{lemma: UNDIDnocov}] To prove Lemma \eqref{lemma: UNDIDnocov}, we may assume that the data is no longer poolable. It is important to note that the unpooled regression can be run with both poolable and unpoolable datasets. Therefore, we can compare the results between the conventional and UN-DID with a poolable dataset. However, the standard pooled regression is unavailable when Data Assumption \eqref{assumption:pooled} does not hold.  Without covariates and with some modifications, Equations \eqref{equation: undid1} and \eqref{equation: undid2} can be rewritten in the following way:
\begin{equation}
\label{equation: undidnocovmodified}
     Y^j_{i,t} = \gamma_0^j + \gamma_1^j post_t^j + \nu^j_{i,t}, \;\;\;\;\;\; \mbox{where} \; j\in \{ Treat, Control \}.
\end{equation}

The proof of this proposition is shown in Appendix \ref{section: transformedequation}. From this transformed equation, $(\widehat{\lambda}_2^{Treat} - \widehat{\lambda}_1^{Treat}) - (\widehat{\lambda}_2^{Control} - \widehat{\lambda}_1^{Control}) = \hat{\gamma}_1^{Treat} - \hat{\gamma}_1^{Control}$.

In Equation \eqref{equation: undidnocovmodified}, the coefficient of interest is $\hat{\gamma}^j_1$. According to the FWL theorem, $\hat{\gamma}^j_1$ will be the same as the coefficient of $(post_t^j - \widehat{post}_t^j)$ from the regression shown in Equation \eqref{equation: undidfwl}. Here, $(post_t^j - \widehat{post}_t^j)$ are the residuals from the regression shown in Equation \eqref{equation: undidfwlresiduals}.
\begin{equation}
\label{equation: undidfwlresiduals}
    post_t^j = \eta_0^j + \omega^j_{i,t}, \;\;\;\;\;\; \mbox{where} \;\;\; \widehat{\eta}_0^j= \overline{post}_t^j.
\end{equation}
\begin{equation}
\label{equation: undidfwl}
    Y_{i,t}^j = \widehat{\gamma_1^j} (post^j_t -\widehat{post}^j_t) + \nu^{j'}_{i,t}.
\end{equation}
Similar to the preceding proof, we calculate the residuals $(post_t^j - \widehat{post}_t^j)$ for each observation using Equation \eqref{equation: undidfwlresiduals}. For an observation in group $j$ during the pre-intervention period, the residual $is$ given by $(post^j_t -\widehat{post}^j_t) = - \widehat{\eta}_0^j$. Likewise, for an observation in group j during the post-intervention period, the residual becomes $(post^j_t -\widehat{post}^j_t) = 1 - \widehat{\eta}_0^j$. Following this, we proceed to obtain $\widehat{\gamma}_1^j$ from Equation \eqref{equation: undidfwl} using the OLS formula:
\begin{align*}
\begin{split}
\widehat{\gamma}_1^j= & \frac{\sum_i \sum_{t} Y_{i, t}^j\left(post^j_t-\widehat{post}^j_t\right)}{\sum_i \sum_t\left(post^j_t-\widehat{post}^j_t\right)^2} \\
\Rightarrow \widehat{\gamma}_1^j= & \frac{\sum_{i \in \{ t \in pre, d \in j\}} Y_{i, t}^{pre,j}\left(-\widehat{\eta}_0^j\right)+\sum_{i \in \{ t \in post, d \in j\}} Y_{i, t}^{post,j}\left(1-\widehat{\eta}_0^j\right)}{\sum_{i \in \{ t \in pre, d \in j\}}\left(-\widehat{\eta}_0^j\right)^2+\sum_{i \in \{ t \in post, d \in j\}}\left(1-\widehat{\eta}_0^j\right)^2} \\
\Rightarrow \widehat{\gamma}_1^j= & \frac{1-\widehat{\eta}_0^j}{\sum_{i \in \{ t \in pre, d \in j\}}\left(-\widehat{\eta}_0^j\right)^2+\sum_{i \in \{ t \in post, d \in j\}}\left(1-\widehat{\eta}_0^j\right)^2} \sum_{i \in \{ t \in post, d \in j\}} Y^{post,j}_{i,t} \\ &-\frac{\widehat{\eta}_0^j}{\sum_{i \in \{ t \in pre, d \in j\}}\left(-\widehat{\eta}_0^j\right)^2+\sum_{i \in \{ t \in post, d \in j\}}\left(1-\widehat{\eta}_0^j\right)^2} \sum_{i \in \{ t \in pre, d \in j\}} Y^{pre,j}_{i,t},
\end{split}
\end{align*}
\begin{equation}
\label{equation: undidgamma}
\widehat{\gamma}_1^j = \frac{1 - \widehat{\eta}_0^j}{B} \sum_{i \in \{t \in post, d \in j \}} Y^{post,j}_{i,t} - \frac{\widehat{\eta}_0^j}{B} \sum_{i \in \{t \in pre, d \in j \}} Y^{pre,j}_{i,t}. \quad 
\end{equation}
Here, $B = \sum_{i \in t=0}\left(-\widehat{\eta}_0^j\right)^2+\sum_{i \in t=1}\left(1-\widehat{\eta}_0^j\right)^2$. Following the proof in  Appendix \ref{ssection: samplesizes}, $\frac{1-\widehat{\eta}_0^j}{B}$ is the reciprocal of the number of observations in group $j$ in the post-intervention period and $\frac{\widehat{\eta}_0^j}{B}$ is the reciprocal of the number of observations in group j for the pre-intervention period. Therefore, Equation \eqref{equation: undidgamma} can be rewritten as:
\begin{equation*}
\widehat{\gamma}_1^j= \frac{1}{N_{j,1}} \sum_{i \in \{t \in post, d \in j \}} Y^{post,j}_{i,t} -\frac{1}{N_{j,0}} \sum_{i \in \{t \in pre, d \in j \}} Y^{pre,j}_{i,t}
\end{equation*}
\begin{equation}
\label{equation: undidgammasimplified}
\Rightarrow \widehat{\gamma}_1^j= \left( \overline{Y}^{post,j}_{i,t} - \overline{Y}^{pre,j}_{i,t} \right).
\end{equation}
From Equation \eqref{equation: undidgammasimplified}, we can take $\widehat{\gamma}^T_1$ and $\widehat{\gamma}^C_1$ from the treated and untreated regressions, respectively, and then take the difference between them.
\begin{align*}
\begin{split}
\widehat{\gamma}_1^{Treat} - \widehat{\gamma}_1^{Control} = \left( \overline{Y}^{post,{Treat}}_{i,t} - \overline{Y}^{pre,{Treat}}_{i,t} \right) - \left( \overline{Y}^{post,{Control}}_{i,t} - \overline{Y}^{pre,{Control}}_{i,t}\right).
\end{split}
\end{align*}
\begin{align}
\label{equation: undidatt}
\begin{split}
    \Rightarrow (\widehat{\lambda}_2^{Treat} - \widehat{\lambda}_1^{Treat}) - (\widehat{\lambda}_2^{Control} - \widehat{\lambda}_1^{Control}) = \left( \overline{Y}^{post,{Treat}}_{i,t} - \overline{Y}^{pre,{Treat}}_{i,t} \right) - \left( \overline{Y}^{post,{Control}}_{i,t} - \overline{Y}^{pre,{Control}}_{i,t} \right).
\end{split}
\end{align}
In simpler notation, Equation \eqref{equation: undidatt} can be rewritten as:
\begin{align}
\label{equation: undidattsimplified}
\begin{split}
(\widehat{\lambda}_2^{Treat} - \widehat{\lambda}_1^{Treat}) - (\widehat{\lambda}_2^{Control} - \widehat{\lambda}_1^{Control}) = & \left( \overline{Y}_{D_s = 1, P_t = 1} - \overline{Y}_{D_s = 1, P_t = 0} \right) - \left( \overline{Y}_{D_s = 0, P_t = 1} - \overline{Y}_{D_s = 0, P_t = 0} \right).
\end{split}
\end{align}
\end{proof}

\begin{proof}[Proof: Theorem \eqref{theorem: theorem2}]
The proof follows the proof of Lemma \eqref{lemma: conventionalnocov} and Lemma \eqref{lemma: UNDIDnocov}. Comparing Equations \eqref{equation: didconventionalbeta3simpler} and \eqref{equation: undidattsimplified}, we observe that, provided we are using the same poolable dataset, both the conventional and the UN-DID will provide numerically equivalent results.
    \begin{equation}
        \label{equation: lemma1proved}
            \widehat{\beta_3} =      (\widehat{\lambda}_2^{Treat} - \widehat{\lambda}_1^{Treat}) - (\widehat{\lambda}_2^{Control} - \widehat{\lambda}_1^{Control}).
    \end{equation}
\end{proof}

\section{Consistency of the UN-DID estimator without covariates}
\label{appendix:consistency}

\begin{proof}[Proof: Theorem \eqref{theorem: consistencywithoutcovariates}] Since the probability limits of both the conventional and UN-DID estimators are identical, it suffices to demonstrate that the UN-DID estimator is consistent. Under the standard OLS regularity conditions, each coefficient estimator in $\widehat{\tau}_{UN}$, as defined in Equation \eqref{equation: undidattmain}, converges in probability to its corresponding population parameter. 
\begin{equation}
\label{equation: parameterconvergence}
    \widehat{\lambda}^j_2 \xrightarrow{p} \lambda^{j*}_2,    \widehat{\lambda}^j_1 \xrightarrow{p} \lambda^{j*}_1 \;\;\;\;\; \mbox{where} \; j \in \{Treat,Control\}
\end{equation}
Taking the difference between $\widehat{\tau}_{UN}$ and $\tau^*_{UN}$ and rearranging terms yields:
\begin{equation}
\begin{aligned}
    \widehat{\tau}_{UN} - \tau^*_{UN} = \bigr( \widehat{\lambda}_2^{Treat} - \lambda_2^{Treat*} \bigr) -& \bigr( \widehat{\lambda}_1^{Treat} - \lambda_1^{Treat*} \bigr) \\ -  \biggr\{ \bigr( \widehat{\lambda}_2^{Control} - \lambda_2^{Control*} \bigr) -& \bigr( \widehat{\lambda}_1^{Control} - \lambda_1^{Control*} \bigr) \biggr\}
\end{aligned}
\end{equation}
By the triangle inequality,
\begin{equation}
\begin{aligned}
\left|\widehat{\tau}_{UN}-\tau_{UN}^{*}\right|\leq{}&\left|\widehat{\lambda}_2^{Treat}-\lambda_2^{Treat*}\right|+ \left| \widehat{\lambda}_1^{Treat}-\lambda_1^{Treat*}\right|\\&+\left|\widehat{\lambda}_2^{Control}-\lambda_2^{Control*}\right|+\left|\widehat{\lambda}_1^{Control}-\lambda_1^{Control*}\right|.
\end{aligned}
\end{equation}
For any $\epsilon > 0$, this inequality implies the following event inclusion:
\begin{equation}
\begin{aligned}
\biggr\{|\widehat{\tau}_{UN}-\tau_{UN}^{*}| > \epsilon\biggr\} \subseteq  & \biggr\{ \bigr(|\widehat{\lambda}_2^{Treat}-\lambda_2^{Treat*}|+ | \widehat{\lambda}_1^{Treat}-\lambda_1^{Treat*}|\\ & +|\widehat{\lambda}_2^{Control}-\lambda_2^{Control*}|+|\widehat{\lambda}_1^{Control}-\lambda_1^{Control*}|\bigr)> \epsilon \biggr\} .
\end{aligned}
\end{equation}
By the averaging principle and the monotonicity property of probability, we get:
\begin{equation}
\begin{aligned}
&\Pr\left(|\widehat{\tau}_{UN}-\tau_{UN}^{*}|>\epsilon
\right) \\ \leq & \Pr\left(|\widehat{\lambda}_2^{Treat}-\lambda_2^{Treat*}|>\frac{\epsilon}{4}\right)+\Pr\left(|\widehat{\lambda}_1^{Treat} -\lambda_1^{Treat*}| > \frac{\epsilon}{4} \right) \\ & + \Pr \left(|\widehat{\lambda}_2^{Control} - \lambda_2^{Control*}|>\frac{\epsilon}{4}\right) + \Pr\left(|\widehat{\lambda}_1^{Control} - \lambda_1^{Control*}|>\frac{\epsilon}{4}\right).
\end{aligned}
\end{equation}
Since each coefficient estimator in $\widehat{\tau}_{UN}$ is consistent, as shown in Equation \eqref{equation: parameterconvergence}, each probability on the right-hand side converges to zero. Therefore,
\begin{equation}
\Pr\left(|\widehat{\tau}_{UN}-\tau_{UN}^{*}|>\epsilon
\right)\longrightarrow 0.
\end{equation}
Hence,
\begin{equation}
\widehat{\tau}_{UN} \xrightarrow{p} \tau_{UN}^{*}.
\end{equation}

By the same standard OLS argument, $\widehat{\beta}_3 \xrightarrow{p} \beta_3^*$. Since the conventional and UN-DID population parameters are equal in the absence of covariates, $\beta_3^* = \tau_{UN}^* = ATT$, so the conventional DID estimator is also consistent.

\end{proof}

\section{Covariate-Adjusted Mean-Independence Result}
\label{conditionalmeanindependenceresult}
\begin{proof}
    Subtracting $\lambda^{j*}_3 x$ from both sides of Equation \eqref{equation: linearity}, and using the condition that $\mathbb{E}[X^j_{i,t} \mid S_i=s,P_t=p,X^j_{i,t}=x ] = x$ , we get:
    \begin{equation}
            \mathbb{E}\left[Y^j_{i,t}(0)\mid S_i=s,P_t=p,X^j_{i,t}=x\right] - \lambda^{j*}_3 \mathbb{E}[X^j_{i,t} \mid S_i=s,P_t=p,X^j_{i,t}=x ]  = \lambda^{j*,\circ}_{p+1}.
    \end{equation}
By the linearity of conditional expectations, we get:
    \begin{equation}
            \mathbb{E}\left[Y^j_{i,t}(0) - \lambda^{j*}_3 X^j_{i,t} \mid S_i=s,P_t=p,X^j_{i,t}=x\right] = \lambda^{j*, \circ}_{p+1}.
    \end{equation}
Applying the Law of Iterated Expectations over the  distribution of $X^j_{i,t}$ gives: 
   \begin{equation}
   \begin{aligned}
        & \mathbb{E}\left[Y^j_{i,t}(0) - \lambda^{j*}_3 X^j_{i,t} \mid S_i=s,P_t=p\right] \\ = & \mathbb{E} \biggr[ \mathbb{E}\left[Y^j_{i,t}(0) - \lambda^{j*}_3 X^j_{i,t} \mid S_i=s,P_t=p,X^j_{i,t} \right] \mid S_i = s, P_t = p \biggr] \\ = & \mathbb{E}\left[\lambda^{j*}_{p+1,0} \mid S_i=s,P_t=p\right] = \lambda^{j*,,\circ}_{p+1}.
    \end{aligned}
    \end{equation}
    Therefore, 
    \begin{equation}
            \mathbb{E}\left[Y^j_{i,t}(0) - \lambda^{j*}_3 X^j_{i,t} \mid S_i=s,P_t=p,X^j_{i,t} = x\right] = \mathbb{E}\left[Y^j_{i,t}(0) - \lambda^{j*}_3 X^j_{i,t} \mid S_i=s,P_t=p\right].
    \end{equation}
\end{proof}

\section{Proof of Theorem \eqref{theorem: theorem3}}
\label{theorem:UNDIDtau}
\begin{proof}[Proof: Theorem \eqref{theorem: theorem3}]
To show that the UN-DID identifies the ATT, we derive the estimator's probability limit. Let $x$ be an arbitrary value in the support of $X_{i,s,t}$ for the treated population. The OLS objective function for the UN-DID regression for group $j$ is given by:
\begin{equation}
    (\lambda_1^{j*},\lambda_2^{j*},\lambda_3^{j*}) = \arg \min_{\lambda_1^{j},\lambda_2^{j},\lambda_3^{j}} \mathbb{E}\bigr[(Y^j_{i,t} - \lambda_1^{j}pre_t^j
- \lambda_2^{j}post_t^j - \lambda_3^{j}X^j_{i,t})^2\bigr].
\end{equation}
The first-order condition with respect to $\lambda^j_1$ is given by:
\begin{equation}
    \mathbb{E} \bigr[ pre_t^j (Y^j_{i,t} - \lambda_1^{j*}pre_t^j
- \lambda_2^{j*}post_t^j - \lambda_3^{j*}X^j_{i,t}) \bigr] = 0.
\end{equation}
Since $pre_t^j$ is a dummy variable with a value of 1 in the pre-intervention period, and $post_t^j = 0$ for the same period, the pre-intervention conditional mean depends only on $P_t = 0$. Rearranging, we get:
\begin{equation}
\label{equation: Lambda1pop}
    \lambda_1^{j*} = \mathbb{E} \bigr[Y^j_{i,t} \mid P_t = 0 \bigr] - \lambda_3^{j*} \mathbb{E}\bigr[X^j_{i,t} \mid P_t = 0 \bigr].
\end{equation}
Similarly, the first-order condition with respect to $\lambda^j_2$ is given by:
\begin{equation}
    \mathbb{E} \bigr[ post_t^j (Y^j_{i,t} - \lambda_1^{j*}pre_t^j
- \lambda_2^{j*}post_t^j - \lambda_3^{j*}X^j_{i,t}) \bigr] = 0.
\end{equation}
Since $post_t^j$ is a dummy variable with a value of 1 in the post-intervention period, and $pre_t^j = 0$ for the same period, the post-intervention conditional mean depends only on $P_t = 1$. Rearranging, we get:
\begin{equation}
\label{equation: Lambda2pop}
    \lambda_2^{j*} = \mathbb{E} \bigr[Y^j_{i,t} \mid P_t = 1 \bigr] - \lambda_3^{j*} \mathbb{E}\bigr[X^j_{i,t} \mid P_t = 1 \bigr].
\end{equation}
Taking a difference between Equations \eqref{equation: Lambda2pop} and \eqref{equation: Lambda1pop}, we get:
\begin{equation}
\label{equation: lambda1lambda2diff}
    \lambda_2^{j*} - \lambda_1^{j*} = \mathbb{E} \bigr[Y^j_{i,t} \mid P_t = 1 \bigr] - \mathbb{E} \bigr[Y^j_{i,t} \mid P_t = 0 \bigr] -\lambda_3^{j*} \biggr( \mathbb{E}\bigr[X^j_{i,t} \mid P_t = 1 \bigr] - \mathbb{E} \bigr[X^j_{i,t} \mid P_t = 0 \bigr] \biggr) . 
\end{equation}
Assuming that the conventional DID and UN-DID estimators are computed using the same sample, the corresponding conditional expectations are identical:
\begin{equation}
\label{equation: condexptheorem4}
\begin{aligned}
\mathbb{E}[Y_{i,s,t}\mid D_s=1,P_t=1]
&=
\mathbb{E}[Y^{Treat}_{i,t}\mid P_t=1],\\
\mathbb{E}[Y_{i,s,t}\mid D_s=1,P_t=0]
&=
\mathbb{E}[Y^{Treat}_{i,t}\mid P_t=0],\\
\mathbb{E}[Y_{i,s,t}\mid D_s=0,P_t=1]
&=
\mathbb{E}[Y^{Control}_{i,t}\mid P_t=1],\\
\mathbb{E}[Y_{i,s,t}\mid D_s=0,P_t=0]
&=
\mathbb{E}[Y^{Control}_{i,t}\mid P_t=0].
\end{aligned}
\end{equation}
Therefore, we can rewrite Equation \eqref{equation: lambda1lambda2diff} for the treated silo as:
\begin{equation}
\begin{aligned}
\label{equation: lambda1lambda2diffmod}
    \lambda_2^{Treat*} - \lambda_1^{Treat*} = & \mathbb{E} \bigr[Y_{i,s,t} \mid D_s= 1, P_t = 1 \bigr] - \mathbb{E} \bigr[Y_{i,s,t} \mid D_s= 1, P_t = 0 \bigr] \\ - &\lambda_3^{Treat*} \biggr( \mathbb{E}\bigr[X_{i,s,t} \mid D_s= 1, P_t = 1 \bigr] - \mathbb{E} \bigr[X_{i,s,t} \mid D_s= 1, P_t = 0 \bigr] \biggr). 
\end{aligned}
\end{equation}
Under Assumptions \eqref{assumption: sutva} and \eqref{assumption: na}, the above expression for the treated silo becomes:
\begin{equation}
\begin{aligned}
\label{equation: lambda1lambda2diffmod2}
    \lambda_2^{Treat*} - \lambda_1^{Treat*} = & \mathbb{E} \bigr[Y_{i,s,t}(1) \mid D_s= 1, P_t = 1 \bigr] - \mathbb{E} \bigr[Y_{i,s,t}(0) \mid D_s= 1, P_t = 0 \bigr] \\ - &\lambda_3^{Treat*} \biggr( \mathbb{E}\bigr[X_{i,s,t} \mid D_s= 1, P_t = 1 \bigr] - \mathbb{E} \bigr[X_{i,s,t} \mid D_s= 1, P_t = 0 \bigr] \biggr). 
\end{aligned}
\end{equation}
Adding and subtracting $\mathbb{E} \bigr[Y_{i,s,t}(0) \mid D_s= 1, P_t = 1 \bigr]$, and setting $ATT = \mathbb{E} \bigr[Y_{i,s,t}(1) - Y_{i,s,t}(0) \mid D_s= 1, P_t = 1 \bigr]$ by its definition, we get:
\begin{equation}
\begin{aligned}
\label{equation: lambda1lambda2diffmod3}
    \lambda_2^{Treat*} - \lambda_1^{Treat*} = ATT + & \mathbb{E} \bigr[Y_{i,s,t}(0) \mid D_s= 1, P_t = 1 \bigr] - \mathbb{E} \bigr[Y_{i,s,t}(0) \mid D_s= 1, P_t = 0 \bigr] \\ - &\lambda_3^{Treat*} \biggr( \mathbb{E}\bigr[X_{i,s,t} \mid D_s= 1, P_t = 1 \bigr] - \mathbb{E} \bigr[X_{i,s,t} \mid D_s= 1, P_t = 0 \bigr] \biggr). 
\end{aligned}
\end{equation}
By the linearity of conditional expectations, we get:
\begin{equation}
\begin{aligned}
\label{equation: lambda1lambda2diffmod4}
    \lambda_2^{Treat*} - \lambda_1^{Treat*} = ATT \quad + & \mathbb{E} \bigr[Y_{i,s,t}(0) - \lambda_3^{Treat*} X_{i,s,t}  \mid D_s= 1, P_t = 1 \bigr]  \\ - & \mathbb{E} \bigr[Y_{i,s,t}(0)- \lambda_3^{Treat*} X_{i,s,t} \mid D_s= 1, P_t = 0 \bigr]. 
\end{aligned}
\end{equation}
Using the mean-independence result established in Equation \eqref{equation: meanindependence}, we can further simplify the above expression.
\begin{equation}
\begin{aligned}
\label{equation: lambda1lambda2diffmod5}
    \lambda_2^{Treat*} - \lambda_1^{Treat*} = ATT \quad + & \mathbb{E} \bigr[Y_{i,s,t}(0) - \lambda_3^{Treat*} X_{i,s,t}  \mid D_s= 1, P_t = 1, X_{i,s,t} = x \bigr]  \\ - & \mathbb{E} \bigr[Y_{i,s,t}(0)- \lambda_3^{Treat*} X_{i,s,t} \mid D_s= 1, P_t = 0, X_{i,s,t} = x \bigr]. 
\end{aligned}
\end{equation}
Again using the linearity of conditional expectations, and using the condition that $\mathbb{E}[X_{i,s,t} \mid S_i=s,P_t=p,X_{i,s,t} = x] = x$, we can further simplify the above expression as:
{\small
\begin{equation}
\begin{aligned}
\label{equation: lambda1lambda2diffmod6}
    \lambda_2^{Treat*} - \lambda_1^{Treat*} = ATT \quad + & \mathbb{E} \bigr[Y_{i,s,t}(0) \mid D_s= 1, P_t = 1, X_{i,s,t} = x \bigr]  \\ - & \mathbb{E} \bigr[Y_{i,s,t}(0) \mid D_s= 1, P_t = 0, X_{i,s,t} = x \bigr] - \lambda_3^{Treat*} x + \lambda_3^{Treat*} x \\ = ATT \quad + & \mathbb{E} \bigr[Y_{i,s,t}(0) \mid D_s= 1, P_t = 1, X_{i,s,t} = x\bigr]  - \mathbb{E} \bigr[Y_{i,s,t}(0) \mid D_s= 1, P_t = 0, X_{i,s,t} = x \bigr]. 
\end{aligned}
\end{equation}
}
Similarly, for the control group, we can derive the following expression:
{\small
\begin{equation}
\begin{aligned}
\label{equation: lambda1lambda2diffcontrol}
    \lambda_2^{Control*} - \lambda_1^{Control*} = \mathbb{E} \bigr[Y_{i,s,t}(0) \mid D_s= 0, P_t = 1, X_{i,s,t} = x \bigr]  - \mathbb{E} \bigr[Y_{i,s,t}(0) \mid D_s= 0, P_t = 0, X_{i,s,t} = x \bigr]. 
\end{aligned}
\end{equation}
}
Taking a difference between Equations \eqref{equation: lambda1lambda2diffcontrol} and \eqref{equation: lambda1lambda2diffmod6}, we get:
\begin{equation}
\label{equation: taualmost}
    \begin{aligned}
        \tau^*_{UN} = & ATT + \mathbb{E} \bigr[Y_{i,s,t}(0) \mid D_s= 1, P_t = 1, X_{i,s,t} = x \bigr]  - \mathbb{E} \bigr[Y_{i,s,t}(0) \mid D_s= 1, P_t = 0, X_{i,s,t} = x \bigr] \\ - & \biggr[\mathbb{E} \bigr[Y_{i,s,t}(0) \mid D_s= 0, P_t = 1, X_{i,s,t} = x \bigr]  - \mathbb{E} \bigr[Y_{i,s,t}(0) \mid D_s= 0, P_t = 0, X_{i,s,t} = x \bigr]  \biggr].
    \end{aligned}
\end{equation}
Under Assumption \eqref{assumption: spt}, we know:
\begin{equation}
\begin{aligned}
    & \mathbb{E} \bigr[Y_{i,s,t}(0) \mid D_s= 1, P_t = 1, X_{i,s,t} = x \bigr]  - \mathbb{E} \bigr[Y_{i,s,t}(0) \mid D_s= 1, P_t = 0, X_{i,s,t} = x \bigr] \\ - & \biggr[\mathbb{E} \bigr[Y_{i,s,t}(0) \mid D_s= 0, P_t = 1, X_{i,s,t} = x \bigr]  - \mathbb{E} \bigr[Y_{i,s,t}(0) \mid D_s= 0, P_t = 0, X_{i,s,t} = x \bigr]  \biggr] = 0.
\end{aligned}
\end{equation}
Therefore, the expression shown in Equation \eqref{equation: taualmost} reduces to:
\begin{equation}
    \begin{aligned}
        \tau^*_{UN} = ATT.
    \end{aligned}
\end{equation}
\end{proof}

\section{Proof of Theorem \eqref{lemma: equivtimevaryingcov}}
\label{appendix:nonequivconvundid}

\begin{proof}[Proof: Theorem \eqref{lemma: equivtimevaryingcov}]

The OLS objective function for the conventional regression is:
\begin{equation}
    (\beta_0^*,\beta_1^*, \beta_2^*,\beta_3^*, \beta_4^*) = \arg \min_{\beta_0,\beta_1, \beta_2,\beta_3, \beta_4} \mathbb{E} \bigr[(Y_{i,s,t} - \beta_0 - \beta_1D_{s} - \beta_2P_{t} -
\beta_3P_{t}*D_{s} - \beta_4X_{i,s,t})^2 \bigr].
\end{equation}
The first-order condition with respect to $\beta_3$ is given by:
\begin{equation}
    \mathbb{E}\bigr[ P_t*D_s (Y_{i,s,t} - \beta_0^* - \beta_1^*D_{s} - \beta_2^*P_{t} -
\beta_3^*P_{t}*D_{s} - \beta_4^{*}X_{i,s,t}) \bigr] = 0.
\end{equation}
By distributing the $(P_t*D_s)$ across the terms, using the linearity of expectations and noting that $(P_t*D_s)^2 = (P_t^2*D_s) = (P_t*D_s^2) = P_t*D_s$  because it is binary, we obtain:
{\small
\begin{equation}
    \mathbb{E} \bigr[ (P_t*D_s) Y_{i,s,t}  \bigr] - \mathbb{E} \bigr[ (P_t*D_s) \beta_0^*  \bigr] - \mathbb{E} \bigr[ (P_t*D_s) \beta_1^*  \bigr] - \mathbb{E} \bigr[ (P_t*D_s) \beta_2^*  \bigr] - \mathbb{E} \bigr[ (P_t*D_s)\beta_3^*\bigr] - \mathbb{E} \bigr[ (P_t*D_s) X_{i,s,t} \beta_4^*  \bigr] = 0.
\end{equation}
}
Re-arranging the above equation, $\beta^*_3$ can be expressed as:
{\small
\begin{equation}
\label{eq: beta3maindecomp}
    \beta^*_3 = \frac{\mathbb{E} \bigr[ (P_t*D_s) Y_{i,s,t}  \bigr] - \mathbb{E} \bigr[ (P_t*D_s) \beta_0^*  \bigr] - \mathbb{E} \bigr[ (P_t*D_s)  \beta_1^* \bigr] - \mathbb{E} \bigr[ (P_t*D_s)  \beta_2^* \bigr] - \mathbb{E} \bigr[ (P_t*D_s) X_{i,s,t} \beta_4^* \bigr]}{\mathbb{E}[P_t*D_s]}.
\end{equation}
}
Under Assumption \eqref{assumption: sutva}, the observed outcome is given by:
\begin{equation}
\label{eq: sutva}
    Y_{i,s,t} = Y_{i,s,t}(0) + \bigr(P_t * D_s\bigr) \bigr(Y_{i,s,t}(1) - Y_{i,s,t}(0)\bigr).
\end{equation}
Plugging in Equation \eqref{eq: sutva} into Equation \eqref{eq: beta3maindecomp} and simplifying, we get:
\begin{equation}
    \beta^*_3 = \frac{\mathbb{E} \biggr[ \bigr(P_t*D_s\bigr) \bigr(Y_{i,s,t}(1) - Y_{i,s,t}(0) \bigr)  \biggr]}{\mathbb{E}[(P_t*D_s)]} + \text{Bias} 
\end{equation}
where,

{\small
\begin{equation}
\label{eq: bias_ccc_viol}
    \text{Bias} = \frac{\mathbb{E} \bigr[ (P_t*D_s) (Y_{i,s,t}(0)) \bigr] - \mathbb{E} \bigr[ (P_t*D_s) \beta_0^*  \bigr] - \mathbb{E} \bigr[ (P_t*D_s)  \beta_1^* \bigr] - \mathbb{E} \bigr[ (P_t*D_s)  \beta_2^* \bigr] - \mathbb{E} \bigr[ (P_t*D_s) X_{i,s,t} \beta_4^* \bigr]}{\mathbb{E}[(P_t*D_s)]}
\end{equation}
}
By the definition of the ATT, we can express its probability limit as:
\begin{equation}
    ATT = \frac{\mathbb{E} \biggr[ \bigr(P_t*D_s\bigr) \bigr(Y_{i,s,t}(1) - Y_{i,s,t}(0) \bigr)  \biggr]}{\mathbb{E}[(P_t*D_s)]} = \mathbb{E}\biggr[Y_{i,s,t}(1) - Y_{i,s,t}(0) \mid P_t*D_s = 1  \biggr]. 
\end{equation}
Therefore, Equation \eqref{eq: beta3maindecomp} can be further simplified as:
\begin{equation}
\label{equation: theorem5mainresult}
    \beta^*_3 = ATT + \text{Bias} 
\end{equation}
The Bias term shown in Equation \eqref{eq: bias_ccc_viol} arises from model mis-specification. Now let us analyze this term in more detail. Under Assumptions \eqref{assumption: spt} and \eqref{as: linearity}, allowing the covariate coefficient to vary across silos when Assumption \eqref{assumption: stateinvariantccc} is violated, the conditional mean of the untreated potential outcome can be expressed as:
\begin{equation}
\label{eq: truepoconv}
    \mathbb{E}[Y_{i,s,t}(0) \mid D_s, P_t, X_{i,s,t} ] = \beta^{\circ}_0 + \beta^{\circ}_1 D_s + \beta^{\circ}_2 P_t + \beta^{\circ, s}_4 X_{i,s,t}.
\end{equation}
Here, $\beta_4^{\circ,s}$ denotes the silo-specific covariate coefficient, corresponding to $\lambda_\circ^{Treat}$ in the treated silo and $\lambda_\circ^{Control}$ in the control silo.
$\beta^{\circ}_1$ captures the unobserved characteristics of the treated region which do not change over time, and $\beta^{\circ}_2$ captures the unobserved factors which affect both the treated and untreated regions in the post-intervention period. To express the untreated potential outcome component of the bias in terms of the conditional mean in Equation \eqref{eq: truepoconv}, we apply the Law of Iterated Expectations.
\begin{equation}
\mathbb{E}\left[(P_t * D_s)Y_{i,s,t}(0)\right] = \mathbb{E}\biggr[
(P_t* D_s) \mathbb{E}\left[Y_{i,s,t}(0)\mid D_s,P_t,X_{i,s,t}\right]\biggr].
\end{equation}
Using the Law of Iterated Expectations as shown above, and substituting Equation $\eqref{eq: truepoconv}$ for the inner conditional expectation, the untreated potential outcome component of the bias can be written in terms of the coefficients of the underlying model. Substituting this expression into Equation $\eqref{eq: bias_ccc_viol}$ and simplifying, we obtain:
{\small
\begin{equation}
\label{eq: biasdecomposed}
    \text{Bias} = \frac{\mathbb{E} \bigr[ (P_t*D_s) (\beta_0^{\circ} - \beta_0^*)  \bigr] + \mathbb{E} \bigr[ (P_t*D_s)  (\beta_1^{\circ} - \beta_1^{*}) \bigr] + \mathbb{E} \bigr[ (P_t*D_s) (\beta_2^{\circ} - \beta_2^*) \bigr] + \mathbb{E} \bigr[ (P_t*D_s) X_{i,s,t} (\beta_4^{\circ, s} - \beta_4^*) \bigr]}{\mathbb{E}[(P_t*D_s)]}.
\end{equation}
}
It is important to note that, when Assumption \eqref{assumption: stateinvariantccc} is violated, the population model underlying the conventional DID regression in Equation \eqref{equation: didpopulation} is mis-specified. Consequently, the true population regression coefficients generally differ from their corresponding coefficients in the underlying model: $(\beta_0^{\circ} \neq \beta_0^*)$, $(\beta_1^{\circ} \neq \beta_1^*)$, $(\beta_2^{\circ} \neq \beta_2^*)$, and $(\beta_4^{\circ,s} \neq \beta_4^*)$. 
Provided that the misspecification components arising from violations of the state-invariant CCC assumption do not exactly offset in the conventional DID, the estimand is generally subject to nonzero misspecification bias, so $\text{Bias} \neq 0$.
It follows that:
\begin{equation}
    \beta^*_3 = ATT + \text{Bias}. 
\end{equation}
The state-invariant CCC violation makes the conventional specification misspecified. Provided the misspecification components do not offset each other, 
\begin{equation}
    \text{Bias} \neq 0. 
\end{equation}
Therefore,
\begin{equation}
\beta_3^* \neq ATT =  \tau_{UN}^*.
\end{equation}
\end{proof}

\section{Proof of Theorem \eqref{theorem: theorem4}}
\label{appendix:Equivalenceundidconvccc}
\begin{proof}[Proof: Theorem \eqref{theorem: theorem4}] Under Assumptions \eqref{assumption: stateinvariantccc} and \eqref{as: linearity}, the population regression model for the conventional DID regression shown in Equation \eqref{equation: didpopulation} is correctly specified. Consider the population regression model in Equation \eqref{equation: didinteracted}:
\begin{equation}
\label{equation: didinteracted2}
    Y_{i,s,t} = \iota^*_0 + \iota^*_1 D_{s} + \iota^*_2 P_{t} + \iota^*_3 P_{t} * D_{s} + \iota^*_4 D_s*X_{i,s,t} + \iota^*_5 (1-D_s)*X_{i,s,t} + \varepsilon_{i,s,t}.
\end{equation}
Under Assumption \eqref{assumption: stateinvariantccc}, the coefficients of the covariates are equal across treated and control silos. 
\begin{equation}
\label{equation: ccccondition}
   \iota^*_4 = \iota^*_5. 
\end{equation}
 Substituting this condition in Equation \eqref{equation: ccccondition} into Equation \eqref{equation: didinteracted2} and simplifying, we get:
 \begin{equation}
\label{equation: didinteracted3}
    Y_{i,s,t} = \iota^*_0 + \iota^*_1 D_{s} + \iota^*_2 P_{t} + \iota^*_3 P_{t} * D_{s} + \iota^*_4 X_{i,s,t} + \varepsilon_{i,s,t}.
\end{equation}
Comparing Equations \eqref{equation: didinteracted3} and \eqref{equation: didpopulation}, we observe that they have the same functional form. Consequently, it follows that their corresponding population coefficients are also equal: $\iota^*_0 = \beta^*_0$, $\iota^*_1 = \beta^*_1$, $\iota^*_2 = \beta^*_2$, $\iota^*_3 = \beta^*_3$ and $\iota^*_4 = \beta^*_4$. Since the model is correctly specified, $\iota^*_0 = \beta^*_0 = \beta^{\circ}_0$, $\iota^*_1 = \beta^*_1 = \beta^{\circ}_1$, $\iota^*_2 = \beta^*_2 = \beta^{\circ}_2$ and $\iota^*_4 = \beta^*_4 = \beta^{\circ}_4$. Therefore, Equation \eqref{equation: theorem5mainresult} reduces to: 
\begin{equation}
\label{equation: theorem2pooledsimplified}
\begin{aligned}
\beta_3^* = ATT.
\end{aligned}
\end{equation}
 It follows that:
\begin{equation}
    \beta_3^* = \tau_{UN}^*.
\end{equation}
\end{proof}

\section[Interpretation Equation (A.5)]{Interpretation of the Weights in Equation \eqref{equation: didconventionalbeta3}}
\label{ssection: samplesizes}

We begin with the auxiliary regression used to residualize the interaction term in Equation \eqref{equation: didconventionalFWLresiduals}:
\begin{equation}
    P_t * D_s = \alpha_0 + \alpha_1 P_t + \alpha_2 D_s + u_{i,s,t}.
\end{equation}
The fitted residuals from the above equation can be expressed as:
\begin{equation}
    \widehat{u}_{i,s,t} = P_t*D_s - \widehat{\alpha}_0 - \widehat{\alpha}_1 P_t - \widehat{\alpha}_2 D_s.
\end{equation}
Since $P_t$ and $D_s$ are binary indicators, the fitted value and the residuals are constant within each group-period cell. In a simple two-silo, two-period case, the residuals for each group-period cell can be expressed as:
\begin{align}
    & \widehat{u}_{Treat,post} =  1-\widehat{\alpha}_0- \widehat{\alpha}_1-\widehat{\alpha}_2, \\
    & \widehat{u}_{Treat,pre} =  -\widehat{\alpha}_0 -\widehat{\alpha}_2, \\ & \widehat{u}_{Control,post} =  -\widehat{\alpha}_0-\widehat{\alpha}_1, \\ & \widehat{u}_{Control,pre} = -\widehat{\alpha}_0.
\end{align}
Now we will show that $P_t$, $D_s$ and the constant are sample orthogonal to these residuals. The OLS objective function for the auxiliary regression is given by:
\begin{equation}
    (\widehat{\alpha_0}, \widehat{\alpha_1}, \widehat{\alpha_2}) = \arg \min_{\alpha_0, \alpha_1, \alpha_2} \sum_{i,s,t}(P_t*D_s - \alpha_0 - \alpha_1 P_t - \alpha_2 D_s)^2
\end{equation}
The first-order conditions with respect to $\alpha_0$, $\alpha_1$ and $\alpha_2$ yield:
\begin{align}
    & \sum_{i,s,t} 1.\widehat{u}_{i,s,t} = 0, \label{eq:orthconstant} \\
    & \sum_{i,s,t} P_t.\widehat{u}_{i,s,t} = 0, \label{eq:orthpost}\\
    & \sum_{i,s,t} D_s.\widehat{u}_{i,s,t} = 0 \label{eq:orthtreat}.
\end{align}
Let $N_{s,t}$ denote the number of observations for silo $s$ in period $t$. Since the fitted residuals remain the same within each group-period cell, we can rewrite Equation \eqref{eq:orthpost} as:
\begin{equation}
\label{eq: b11}
    \sum_{i,s,t} P_t.\widehat{u}_{i,s,t} = N_{Treat,post} \widehat{u}_{Treat,post} + N_{Control,post} \widehat{u}_{Control,post} = 0.
\end{equation}
Rearranging, we get:
\begin{equation}
\label{eq: b12}
     \widehat{u}_{Control,post} = - \frac{N_{Treat,post}}{N_{Control,post}}\widehat{u}_{Treat,post} .
\end{equation}
Similarly, we can rewrite Equation \eqref{eq:orthtreat} as:
\begin{equation}
\label{eq: b13}
    \sum_{i,s,t} D_s.\widehat{u}_{i,s,t} = N_{Treat,post} \widehat{u}_{Treat,post} + N_{Treat,pre} \widehat{u}_{Treat,pre} = 0.
\end{equation}
Rearranging, we get:
\begin{equation}
\label{eq: b14}
    \widehat{u}_{Treat,pre} = - \frac{N_{Treat,post}}{N_{Treat,pre}} \widehat{u}_{Treat,post}.
\end{equation}
We can also rewrite Equation \eqref{eq:orthconstant} as:
\begin{equation}
\label{eq: b15}
    \sum_{i,s,t} 1.\widehat{u}_{i,s,t} = N_{Treat,post} \widehat{u}_{Treat,post} + N_{Treat,pre} \widehat{u}_{Treat,pre} + N_{Control,post} \widehat{u}_{Control,post} + N_{Control,pre} \widehat{u}_{Control,pre} = 0.
\end{equation}
Plugging in Equations \eqref{eq: b12} and \eqref{eq: b14} into Equation \eqref{eq: b15} and simplifying, we get:
\begin{equation}
\label{eq: b16}
    \widehat{u}_{Control,pre} = \frac{N_{Treat,post}}{N_{Control,pre}} \widehat{u}_{Treat,post}.
\end{equation}
Now we will show that a ``difference-in-differences" of the four residuals is equal to one. 
\begin{align*}
    & \bigr(\widehat{u}_{Treat,post} - \widehat{u}_{Treat,pre} \bigr) - \bigr(\widehat{u}_{Control,post} - \widehat{u}_{Control,pre}  \bigr) \\
    = & \bigr( 1 - \widehat{\alpha}_0 - \widehat{\alpha}_1 - \widehat{\alpha}_2 + \widehat{\alpha}_0 +\widehat{\alpha}_2  \bigr) - \bigr(- \widehat{\alpha}_0 - \widehat{\alpha}_1 + \widehat{\alpha}_0 \bigr) \\
    = & 1 - \widehat{\alpha}_1 + \widehat{\alpha}_1 = 1.  
\end{align*}
Therefore,
\begin{equation}
\label{eq: b17}
    \bigr(\widehat{u}_{Treat,post} - \widehat{u}_{Treat,pre} \bigr) - \bigr(\widehat{u}_{Control,post} - \widehat{u}_{Control,pre}  \bigr) = 1.
\end{equation}
Plugging in Equations \eqref{eq: b12}, \eqref{eq: b14} and \eqref{eq: b16} into Equation \eqref{eq: b17} and simplifying, we get:
\begin{equation}
\label{eq: b18}
    \widehat{u}_{Treat,post}  \biggr( 1 + \frac{N_{Treat,post}}{N_{Treat,pre}} + \frac{N_{Treat,post}}{N_{Control,post}}  +\frac{N_{Treat,post}}{N_{Control,pre}}  \biggr) = 1.
\end{equation}
Dividing both sides of Equation \eqref{eq: b18} by $N_{Treat,post}$ and re-arranging yields:
\begin{equation}
\label{eq: b19}
    \widehat{u}_{Treat,post} \biggr( \frac{1}{N_{Treat,post}} + \frac{1}{N_{Treat,pre}} + \frac{1}{N_{Control,post}}  +\frac{1}{N_{Control,pre}}  \biggr) = \frac{1}{N_{Treat,post}}
\end{equation}
To simplify the notation, define $H$ as:
    \begin{equation}
        H = \biggr( \frac{1}{N_{Treat,post}} + \frac{1}{N_{Treat,pre}} + \frac{1}{N_{Control,post}}  +\frac{1}{N_{Control,pre}}  \biggr).
    \end{equation}
Hence, Equation \eqref{eq: b19} becomes:
\begin{equation}
\label{eq: b21}
      \widehat{u}_{Treat,post} = \frac{1}{N_{Treat,post}H}  
\end{equation}
Plugging in Equation \eqref{eq: b21} into Equations \eqref{eq: b12}, \eqref{eq: b14} and \eqref{eq: b16}, we get: 
\begin{align}
      & \widehat{u}_{Treat,pre} = - \frac{1}{N_{Treat,pre}H}, \label{eq: b22} \\
      & \widehat{u}_{Control,post} = - \frac{1}{N_{Control,post}H}, \label{eq: b23} \\ & \widehat{u}_{Control,pre} = \frac{1}{N_{Control,pre}H} \label{eq: b24}.
\end{align}
Now the sum of squared residuals is given by:

\begin{equation}
\label{eq: b25}
    A = N_{Treat,post} (\widehat{u}_{Treat,post})^2 + N_{Treat,pre} (\widehat{u}_{Treat,pre})^2 + N_{Control,post} (\widehat{u}_{Control,post})^2 + N_{Control,pre} (\widehat{u}_{Control,pre})^2.
\end{equation}
Substituting Equations \eqref{eq: b21}, \eqref{eq: b22}, \eqref{eq: b23} and \eqref{eq: b24} into Equation \eqref{eq: b25}, we get:
\begin{align}
    A &= \frac{1}{H^2} \left(\frac{1}{N_{Treat,post}} + \frac{1} {N_{Treat,pre}} + \frac{1}{N_{Control,post}} + \frac{1}{N_{Control,pre}}\right) \\ & = \frac{1}{H^2} * H = \frac{1}{H}.
\end{align}
Therefore, Equations \eqref{eq: b21}, \eqref{eq: b22}, \eqref{eq: b23} and \eqref{eq: b24} can be expressed as:
\begin{align}
    \frac{\widehat{u}_{Treat,post}}{A} = \frac{1}{N_{Treat,post}},
    \\ \frac{\widehat{u}_{Treat,pre}}{A} = - \frac{1}{N_{Treat,pre}}, \\
    \frac{\widehat{u}_{Control,post}}{A} = - \frac{1}{N_{Control,post}}, \\ \frac{\widehat{u}_{Control,pre}}{A} = 
    \frac{1}{N_{Control,pre}}.
\end{align}

Using the definitions of the residuals above, this implies
\begin{align}
    \frac{
    1-\widehat{\alpha}_0-\widehat{\alpha}_1-\widehat{\alpha}_2
    }{A} = \frac{1}{N_{Treat,post}}, \\
    \frac{\widehat{\alpha}_0+\widehat{\alpha}_2}{A} = \frac{1}{N_{Treat,pre}}, \\
    \frac{\widehat{\alpha}_0+\widehat{\alpha}_1}{A} = \frac{1}{N_{Control,post}}, \\ \frac{-\widehat{\alpha}_0}{A} = \frac{1}{N_{Control,pre}}.
\end{align}
Thus, the weights appearing in Equation \eqref{equation: didconventionalbeta3} are the reciprocals of the corresponding group-period cell sizes. 
\clearpage

\section[Proof: Lemma 1]{Proof of Theorem \eqref{theorem: nonequivestimatorsassumption9}} \label{sec:proof_conventionaltimevaryingatt}

\begin{proof}[Proof: Theorem \eqref{theorem: nonequivestimatorsassumption9}]

In the conventional regression shown in Equation \eqref{equation: didsimple}, $\hat{\beta_3}$ is the parameter of interest. According to the FWL theorem, $\hat{\beta_3}$ will be the same as the coefficient of $(P_t*D_s - \widehat{P_t*D_s})$ from the regression shown in Equation \eqref{equation: didcovfwl}. Here, $(P_t*D_s - \widehat{P_t*D_s})$ are the residuals from Equation \eqref{equation: didcovfwlresiduals}:
\begin{equation}
\label{equation: didcovfwlresiduals}
    P_t * D_s = \alpha_0 + \alpha_1 P_t + \alpha_2 D_s + \alpha_3 X_{i,s,t} + u_{i,s,t},
\end{equation}
\begin{equation}
\label{equation: didcovfwl}
    Y_{i,s,t} = \widehat{\beta_3} (P_t*D_s - \widehat{P_t*D_s}) + u'_{i,s,t}.
\end{equation}
From Equation \eqref{equation: didcovfwlresiduals}, we calculate residuals ($P_t * D_s - \widehat{P_t * D_s}$) for each observation based on their treatment status and time period. For treated observations in the pre-intervention period, the residual is given by $- \widehat{\alpha}_0 - \widehat{\alpha}_2 - \widehat{\alpha}_3 X_{i,s,t}$, while for the post-intervention period, it is $1 - \widehat{\alpha}_0 - \widehat{\alpha}_1 - \widehat{\alpha}_2 - \widehat{\alpha}_3 X_{i,s,t}$. Similarly, untreated observations have residuals of $- \widehat{\alpha}_0 - \widehat{\alpha}_3 X_{i,s,t}$ in the pre-intervention period and $- \widehat{\alpha}_0 - \widehat{\alpha}_1 - \widehat{\alpha}_3 X_{i,s,t}$ in the post-intervention period. Subsequently, $\widehat{\beta}_3$ from Equation \eqref{equation: didcovfwl} can be estimated using the following OLS formula:
\begin{align}
\label{equation: didcovols}
\begin{split}
\widehat{\beta}_3= & \frac{\sum_i \sum_t Y_{i,s,t}\left(P_t * D_s-\widehat{P_t * D_s}\right)}{\sum_i \sum_t \left(P_t * D_s-\widehat{P_t * D_s}\right)^2}.
\end{split}
\end{align}
The expression from Equation \eqref{equation: didcovols} is simplified by substituting the computed residuals into the numerator, as shown below:
\begin{align*}
    \begin{split}
        & \sum_{i \in \{t \in pre,{d \in Treat}\} } Y_{i,t}^{pre,Treat}\left(-\widehat{\alpha}_0-\widehat{\alpha}_2 - \widehat{\alpha}_3 X_{i,t}^{pre,Treat} \right) \\ + & \sum_{i \in \{t \in  post,{d \in Treat}\}} Y_{i,t}^{post,Treat}\left(1-\widehat{\alpha}_0-\widehat{\alpha}_1-\widehat{\alpha}_2 - \widehat{\alpha}_3 X_{i,t}^{post,Treat} \right) \\ + & \sum_{i \in \{t \in  pre,{d \in Control}\}} Y_{i,t}^{pre,Control}\left(-\widehat{\alpha}_0- \widehat{\alpha}_3 X^{pre,Control}_{i,t}\right) \\ + & \sum_{i \in \{t \in  post,{d \in Control}\}} Y_{i,t}^{post,Control}\left(-\widehat{\alpha}_0-\widehat{\alpha}_1 - \widehat{\alpha}_3 X_{i,t}^{post,Control}\right).
    \end{split}
\end{align*}
In the above expression, $Y_{i,t}^{p,d}$ and $X_{i,t}^{p,d}$ are the observed outcomes and the observed covariates for units in group $d$ in period $p$, respectively. Similar to the expression above, we substitute the residuals and simplify in the denominator of Equation \eqref{equation: didcovols} and recover the following expression:
\begin{align*}
    \begin{split}
        G = & \sum_{i \in\{t \in pre,d \in Treat\}} \left(-\widehat{\alpha}_0-\widehat{\alpha}_2 - \widehat{\alpha}_3 X^{pre,Treat}_{i,t}\right)^2 + \sum_{i \in\{t \in post,d \in Treat\}}\left(1-\widehat{\alpha}_0-\widehat{\alpha}_1-\widehat{\alpha}_2 - \widehat{\alpha}_3 X^{post,Treat}_{i,t}\right)^2 \\ & + \sum_{i \in\{t \in pre,d \in Control\}}\left(-\widehat{\alpha}_0 - \widehat{\alpha}_3 X^{pre,Control}_{i,t}\right)^2 + \sum_{i \in\{t \in post,d \in Control\}}\left(-\widehat{\alpha}_0-\widehat{\alpha}_1 - \widehat{\alpha}_3 X^{post,Control}_{i,t}\right)^2. 
    \end{split}
\end{align*}
Combining the numerator and the denominator, $\widehat{\beta}_3$ can be written as:
\begin{align}
\label{equation: didcovbeta3}
\begin{split}
\widehat{\beta}_3= &  \frac{\sum_{i \in \{t \in  post,{d \in Treat}\}} Y_{i,t}^{post,Treat}\left(1-\widehat{\alpha}_0-\widehat{\alpha}_1-\widehat{\alpha}_2 - \widehat{\alpha}_3 X_{i,t}^{post,Treat} \right)}{G} \\ - & \frac{\sum_{i \in \{t \in pre,{d \in Treat}\}} Y_{i,t}^{pre,Treat}\left(\widehat{\alpha}_0 +\widehat{\alpha}_2 + \widehat{\alpha}_3 X_{i,t}^{pre,Treat} \right)}{G} \\ - & \frac{\sum_{i \in \{t \in  post,{d \in Control}\}} Y_{i,t}^{post,Control} \left(\widehat{\alpha}_0+\widehat{\alpha}_1 + \widehat{\alpha}_3 X_{i,t}^{post,Control}\right)}{G} \\ - & \frac{\sum_{i \in \{t \in  pre,{d \in Control}\}} Y_{i,t}^{pre,Control}\left(\widehat{\alpha}_0+ \widehat{\alpha}_3 X_{i,t}^{pre,Control}\right)}{G} .
\end{split}
\end{align}

 With covariates and with some modifications, Equations \eqref{equation: undid1} and \eqref{equation: undid2} can be rewritten in the following way:
	\begin{equation}
		\label{equation: undidcovmodified}
		  \; Y^j_{i,t} = \gamma_0^j + \gamma_1^j post_t^j + \gamma_2^j X^j_{i,t} + \nu^j_{i,t}, \;\;\;\;\;\; \mbox{where} \; j \in \{ Treat, Control \}.
	\end{equation}

	A proof of this equivalence is shown in Appendix \ref{section: transformedequation}. From this transformed equation, $(\widehat{\lambda}_2^{Treat} - \widehat{\lambda}_1^{Treat}) - (\widehat{\lambda}_2^{Control} - \widehat{\lambda}_1^{Control}) = \hat{\gamma}_1^{Treat} - \hat{\gamma}_1^{Control}$. In Equation \eqref{equation: undidcovmodified}, the coefficient of interest is $\hat{\gamma}^j_1$. It is worthwhile to note that the coefficient of the covariate $X^j_{i,t}$, $\gamma_2^j$ may vary between the treated and control silos if Assumption \eqref{assumption: stateinvariantccc} does not hold. Using the FWL theorem, it follows that $\widehat{\gamma}_1^j$ is equal to the coefficient of $(post^j_t -\widehat{post}^j_t)$ in the regression presented in Equation \eqref{equation: undidcovfwl}. Here, $(post^j_t -\widehat{post}^j_t)$ are the residuals from the model shown in Equation \eqref{equation: undidcovfwlresiduals}:
	\begin{equation}
		\label{equation: undidcovfwlresiduals}
		post_t^j = \eta_0^j + \eta_1^j X^j_{i,t} + \omega^j_{i,t}.
	\end{equation}
	\begin{equation}
		\label{equation: undidcovfwl}
		Y_{i,t}^j = \gamma_1^j (post^j_t -\widehat{post}^j_t) + \nu^{j'}_{i,t}.
	\end{equation}
	We now compute the residuals $(post_t^j - \widehat{post}_t^j)$ for each observation using Equation \eqref{equation: undidcovfwlresiduals}. For an observation in group $j$ during the pre-intervention period, the residual is given by $- \widehat{\eta}_0^j - \widehat{\eta}_1^j X^j_{i,t}$. Similarly, for an observation in group $j$ during the post-intervention period, the residual is expressed as $1 - \widehat{\eta}_0^j - \widehat{\eta}_1^j X^j_{i,t}$.
	Following this, we move on to derive $\widehat{\gamma}_1^j$ from Equation \eqref{equation: undidcovfwl} using the OLS formula:
	\begin{align*}
		\begin{split}
			\widehat{\gamma}_1^j= & \frac{\sum_i \sum_{t} Y_{i, t}^j \left(post^j_t-\widehat{post}^j_t\right)}{\sum_i \sum_t\left(post^j_t-\widehat{post}^j_t\right)^2} \\
		\end{split}
	\end{align*}
	\begin{equation}
		\label{equation: undidcovreduced}
		\Rightarrow \widehat{\gamma}_1^j= \frac{\sum_{i \in \{t \in post, x \in X^j\}} Y_{i,t}^{post,j}(1 - \widehat{\eta}_0^j - \widehat{\eta}_1^j X^{post,j}_{i,t}) + \sum_{i \in \{t \in pre, x \in X^j\}} Y_{i,t}^{pre,j} (- \widehat{\eta}_0^j - \widehat{\eta}_1^j X^{pre,j}_{i,t}) } {F_j}.
	\end{equation}
	Here, 
	\begin{equation*}
		F_j = \sum_{i \in \{t \in post\}}\left(1 -\widehat{\eta}_0^j- \widehat{\eta}_1^j X^{post,j}_{i,t}\right)^2+\sum_{i \in \{t \in pre\}}\left(-\widehat{\eta}_0^j- \widehat{\eta}_1^j X^{pre,j}_{i,t}\right)^2.
	\end{equation*}

Comparing Equations \eqref{equation: didcovbeta3} and \eqref{equation: undidcovreduced}, we observe that the conventional DID and the UN-DID estimators depend on different residualization procedures. In the conventional regression, the treatment interaction is residualized assuming a common coefficient on the covariates. In contrast, UN-DID residualizes the period indicator separately within each silo. Consequently, $\widehat{\eta}^{treat}_1$ and  $\widehat{\eta}^{control}_1$ need not equal $ \widehat{\alpha}_3$ in finite samples. Therefore, the two expressions are not algebraically equivalent, although they may coincide for some sample realizations. 

\end{proof}
\clearpage

\section[Proof of Theorem 11]{Proof of Theorem \eqref{theorem: theorem11}}
\label{sec: theorem11}

\begin{proof}[Proof: Theorem \eqref{theorem: theorem11}]
 
In the context of time-invariant covariates, the regression coefficients for $X_{i}$ in  Equations \eqref{equation: didcovfwlresiduals} and \eqref{equation: undidcovfwlresiduals} are observed to be 0. Let us consider a balanced panel, where each individual is observed in both the pre-treatment and post-treatment periods. Since the covariates are time-invariant, $X_{i,0}=X_{i,1}=X_i$. We first show that the coefficient of $X_i$ ($\widehat{\alpha_3}$) in Equation \eqref{equation: didcovfwlresiduals} is zero under the pooled sample orthogonality condition:
\begin{equation}
\label{equation: orthogonalitycondition1}
    \sum_{i,s,t}X_i(\widetilde{P_t * D_s}) = 0.
\end{equation}

We will now show that this condition holds in a balanced panel. Using the FWL theorem to identify the coefficient on $P_t*D_s$, we begin with the full auxiliary regression:
\begin{equation}
\label{equation: orthogonalityderivation3}
    P_t*D_s = \alpha_0 + \alpha_1 P_t + \alpha_2 D_s + \alpha_3 X_i + u_{i,s,t}.
\end{equation}
To derive $\alpha_3$ using the FWL theorem from the above regression, we can proceed in three steps. First, we regress $D_s * P_t$ on $P_t$ and $D_s$:
\begin{equation}
\label{equation: orthogonalityderivation1}
    P_t*D_s = \xi_0 + \xi_1 P_t + \xi_2 D_s + \widetilde{P_t*D_s}.
\end{equation}
Second, we regress $X_i$ on $P_t$ and $D_s$:
\begin{equation}
\label{equation: orthogonalityderivation2}
    X_i = \upsilon_0 + \upsilon_1 P_t + \upsilon_2 D_s + \widetilde{X_i}.
\end{equation}
In the third step, we regress $\widetilde{P_t*D_s}$ on $\widetilde{X_i}$, which are the residuals from Equations \eqref{equation: orthogonalityderivation1} and \eqref{equation: orthogonalityderivation2} respectively. According to the FWL theorem, the coefficient of $X_i$ from this regression will be the same as the coefficient of $X_i$ from Equation \eqref{equation: orthogonalityderivation3}. Therefore,
\begin{equation}
\label{equation: orthogonalityderivation4}
    \widehat{\alpha_3} = \frac{\sum_{i,s,t}\widetilde{X_{i,s,t}}(\widetilde{P_t * D_s}) }{\sum_{i,s,t} \widetilde{X}_{i,s,t}^2} 
\end{equation}
From Equation \eqref{equation: orthogonalityderivation4}, $\widehat{\alpha_3}$ will be zero under the following \textbf{pooled sample orthogonality condition:}
\begin{equation}
\label{equation: orthogonalitycondition1repeated}
    \sum_{i,s,t}X_i(\widetilde{P_t * D_s}) = 0.
\end{equation}

Now we will prove that this holds in a balanced panel. We can rewrite Equation \eqref{equation: orthogonalitycondition1repeated} as:
\begin{equation}
    \sum_{i,s,t}X_i(D_s-\overline{D})(P_t - \overline{P}) = 0.
\end{equation}

In a balanced two-period panel, every individual is observed once when $P_t = 0$ and again when $P_t = 1$. So, $\overline{P} = 1/2$. For each individual:
\begin{equation}
    \sum_{t=0}^1 X_i(D_s - \overline{D})(P_t - \overline{P}) = X_i(D_s - \overline{D})\bigr[(0 - 1/2) + (1 - 1/2)\bigr] = X_i(D_s - \overline{D}) (0) = 0
\end{equation}

Therefore, 
\begin{equation}
    \sum_{i,s,t}X_i(\widetilde{P_t * D_s}) = 0.
\end{equation}

Now, we show that the coefficient of $X^j_i$ ($\widehat{\eta_1^j}$) in Equation \eqref{equation: undidcovfwlresiduals} is zero under the silo-specific sample orthogonality condition:
\begin{equation}
    \label{equation: orthogonalitycondition2repeated}
        \sum_{i,t \in j} (X_i^j - \overline{X^j})(P_t^j - \overline{P^j}) = 0.
\end{equation}
Equation \eqref{equation: orthogonalitycondition2repeated} can be derived by applying the FWL theorem in the same manner as in the derivation of Equation \eqref{equation: orthogonalitycondition1repeated}. For each silo, $\eta^j_1$ is given by:
\begin{equation}
\label{equation: orthogonalityderivation5}
    \widehat{\eta^j_1} = \frac{\sum_{i,t}(X^j_i - \overline{X}^j)(P^j_t-\overline{P^j}) }{\sum_{i,t} (X^j_i- \overline{X}^j)^2} 
\end{equation}
From Equation \eqref{equation: orthogonalityderivation5}, $\widehat{\eta}^j_1$ will be zero under the silo-specific orthogonality condition shown in Equation \eqref{equation: orthogonalitycondition2repeated}. For each individual in silo $j$ in a balanced panel:
\begin{equation}
    \sum_{t=0}^1 (X^j_i - \overline{X}^j)(P^j_t - \overline{P}^j) = (X^j_i - \overline{X}^j)\bigr[(0 - 1/2) + (1 - 1/2)\bigr] = \left(X_i^j - \overline{X}^j\right)(0) = 0.
\end{equation}
Therefore,
\begin{equation}
        \sum_{i,t \in j} (X_i^j - \overline{X}^j)(P_t^j - \overline{P}^j) = 0.
\end{equation}
Imposing $\widehat{\alpha_3}=0$ in Equation \eqref{equation: didcovbeta3}, we get: 
\begin{align}
\label{equation: orthogonalityderivation6}
\begin{split}
\widehat{\beta}_3= &  \frac{\sum_{i \in \{t \in  post,{d \in Treat}\}} Y_{i,t}^{post,Treat}\left(1-\widehat{\alpha}_0-\widehat{\alpha}_1-\widehat{\alpha}_2 \right)}{G} \\ - & \frac{\sum_{i \in \{t \in pre,{d \in Treat}\}} Y_{i,t}^{pre,Treat}\left(\widehat{\alpha}_0 +\widehat{\alpha}_2 \right)}{G} \\ - & \frac{\sum_{i \in \{t \in  post,{d \in Control}\}} Y_{i,t}^{post,Control} \left(\widehat{\alpha}_0+\widehat{\alpha}_1 \right)}{G} \\ - & \frac{\sum_{i \in \{t \in  pre,{d \in Control}\}} Y_{i,t}^{pre,Control}\left(\widehat{\alpha}_0\right)}{G} .
\end{split}
\end{align}
Similarly, imposing $\widehat{\eta}^j_1 = 0$ in Equation \eqref{equation: undidcovreduced}, we get:
	\begin{equation}
		\label{equation: orthogonalityderivation7}
		\Rightarrow \widehat{\gamma}_1^j= \frac{1-\widehat{\eta}_0^j}{F_j} \sum_{i \in \{t \in post, x \in X^j\}} Y^{post,j}_{i,t} -\frac{\widehat{\eta}_0^j}{F_j} \sum_{i \in \{t \in pre,x \in X^j\}} Y^{pre,j}_{i,t}.
	\end{equation}
Since $\widehat{\alpha}_3 = 0$ and $\widehat{\eta}_1^j = 0$, the FWL residuals for the pooled and silo-specific regressions reduce to their corresponding no-covariate residuals. Therefore, following the same approach as the proof shown in Appendix \ref{ssection: Lemma1}, we get
\begin{equation}
\label{equation: orthogonalityderivation8}
    \widehat{\beta_3} = (\widehat{\lambda}_2^{Treat} - \widehat{\lambda}_1^{Treat}) - (\widehat{\lambda}_2^{Control} - \widehat{\lambda}_1^{Control}).
\end{equation}

\end{proof}

\section{Proof of Theorem \eqref{theorem: DIDINT}}
\label{appendix:DIDINT}

\begin{proof}[Proof: Theorem \eqref{theorem: DIDINT}] The population regression model for the state-varying DID-INT is given by:
\begin{equation}
    \label{equation: DIDINTpopulation}
        Y_{i,s,t} = \psi^*_1 pre_t^{Treat} + \psi^*_2 post_t^{Treat} + \psi^*_3 pre_t^{Control} + \psi^*_4 post_t^{Control} + \psi^*_5 X^{Treat}_{i,s,t} + \psi^*_6 X^{Control}_{i,s,t} + e^*_{i,s,t}
\end{equation}

For the treated group in the post-intervention period:
\begin{equation}
\label{equation: theorem12first} \mathbb{E}[Y_{i,s,t}|D_s = 1, P_t = 1, X_{i,s,t} = x] = \psi_2^* + \psi^*_5 x 
\end{equation}
For the treated group in the pre-intervention period:
\begin{equation}
\label{equation: theorem12second} \mathbb{E}[Y_{i,s,t}|D_s = 1, P_t = 0, X_{i,s,t} = x] = \psi_1^* + \psi^*_5 x
\end{equation}
For the control group in the post-intervention period:
\begin{equation}
\label{equation: theorem12third} \mathbb{E}[Y_{i,s,t}|D_s = 0, P_t = 1, X_{i,s,t} = x] = \psi_4^* + \psi^*_6 x 
\end{equation}
For the control group in the pre-intervention period:
\begin{equation}
\label{equation: theorem12fourth} \mathbb{E}[Y_{i,s,t}|D_s = 0, P_t = 0, X_{i,s,t}= x] = \psi_3^* + \psi^*_6 x
\end{equation}
Since the state-varying DID-INT regression is correctly specified, $\mathbb{E}[e^{*}_{i,s,t}|D_s,P_t,X_{i,s,t}] = 0$. Taking a difference between Equations \eqref{equation: theorem12first} and \eqref{equation: theorem12second}, and a difference between Equations \eqref{equation: theorem12third} and \eqref{equation: theorem12fourth}, and subsequently taking a difference of these equations and rearranging, we get:
\begin{equation}
\begin{aligned}
\label{g6}
\tau^*_{DT}
= & \Bigl(
\mathbb{E}[Y_{i,s,t}\mid D_s=1,P_t=1,X_{i,s,t}=x]
-
\mathbb{E}[Y_{i,s,t}\mid D_s=1,P_t=0,X_{i,s,t}=x]
\Bigr)
\\
\quad-&
\Bigl(
\mathbb{E}[Y_{i,s,t}\mid D_s=0,P_t=1,X_{i,s,t}=x]
-
\mathbb{E}[Y_{i,s,t}\mid D_s=0,P_t=0,X_{i,s,t}=x]
\Bigr)
\end{aligned}
\end{equation}
Under the identifying assumptions maintained in Theorem \eqref{theorem: DIDINT}, the right-hand side of Equation \eqref{g6} identifies the ATT. Therefore,
\begin{equation}
    \tau^*_{DT} = ATT
\end{equation}
\end{proof}

\section[Proof: Theorem 5]{Proof of Theorem \eqref{theorem: theorem13DIDINTequiv}}
\label{section:proof_RPeqUNDID}

\begin{proof}[Proof: Theorem \eqref{theorem: theorem13DIDINTequiv}]
    Let us begin by deriving $\hat{\psi_1}$ from the regression shown in Equation \eqref{equation: DIDINT}. Let $\widehat{pre}_t^{Treat}$ be the fitted value from the following regression:
        \begin{equation}
            \label{equation: DIDINTpre1}
                pre_t^{Treat} = \pi_2^{Treat} post_t^{Treat} + \pi_3^{Control} pre_t^{Control} + \pi_4^{Control} post_t^{Control} + \pi_5^{Treat} X^{Treat}_{i,t} + \pi_6^{Control} X^{Control}_{i,t} + e'_{i,t}
        \end{equation}

It can be shown that $pre_t^{Treat}$ is sample-orthogonal to $pre_t^{Control}$, $post_t^{Control}$ and $X_{i,t}^{Control}$ because the treated and control silo variables have disjoint nonzero support. A proof of this proposition is shown in Appendix \ref{ssection: independence}. Therefore, we can rewrite Equation \eqref{equation: DIDINTpre1} as:

\begin{equation}
\label{equation: DIDINTpre1simplified}
    pre_t^{Treat} = \pi_2^{Treat} post_t^{Treat} + \pi_5^{Treat} X^{Treat}_{i,t} + e'_{i,t}
\end{equation}

We proceed to calculate the residuals for each observation based on their treatment status and time period. For a treated observation in the pre-intervention period, the residual is given by $1 - \hat{\pi^{Treat}_5} X^{Treat}_{i,0}$. Similarly, for a treated observation in the post-intervention period, the residual is given by $-\hat{\pi^{Treat}_2} - \hat{\pi^{Treat}_5} X^{Treat}_{i,1}$. According to the FWL theorem, the coefficient $\hat{\psi_1}$ will be the same as the coefficient of $(pre^{Treat}_t -\widehat{pre}^{Treat}_t)$ from the regression shown in Equation \eqref{equation: DIDINTpre1residuals}. $(pre^{Treat}_t -\widehat{pre^{Treat}_t})$ are the residuals from the regression shown in Equation \eqref{equation: DIDINTpre1simplified}.

\begin{equation}
\label{equation: DIDINTpre1residuals}
    Y_{i,t} = \widehat{\psi}_1 (pre^{Treat}_t -\widehat{pre}^{Treat}_t) + e''_{i,t} 
\end{equation}

From Equation \eqref{equation: DIDINTpre1residuals}, we will derive $\widehat{\psi}_1$ using the OLS formula:

\begin{align*}
\begin{split}
    \hat{\psi_1} = & \frac{\sum_i \sum_t Y_{i,t}(pre^{Treat}_t -\widehat{pre}^{Treat}_t)}{\sum_i \sum_t (pre^{Treat}_t -\widehat{pre}^{Treat}_t)^2}
\end{split}
\end{align*}

\begin{equation}
\label{equation: ATTDTfirst}
\resizebox{\textwidth}{!}{$
\begin{aligned}
    \Rightarrow \hat{\psi_1} = & \frac{
        \sum_{i \in \{t \in pre, x \in X^{Treat}\}} 
        \left( 1 - \hat{\pi}^{Treat}_5 X^{pre,{Treat}}_{i,t} \right) Y^{pre,{Treat}}_{i,t} + 
        \sum_{i \in \{t \in post, x \in X^{Treat}\}} 
        \left( -\hat{\pi}^{Treat}_2 - \hat{\pi}^{Treat}_5 X^{post,{Treat}}_{i,t} \right) Y^{post,{Treat}}_{i,t}
    }{
        \sum_{i \in \{ t \in pre, x \in X^{Treat} \}} 
        \left( 1 - \hat{\pi}^{Treat}_5 X^{pre,{Treat}}_{i,t} \right)^2 + 
        \sum_{i \in \{ t \in post, x \in X^{Treat} \}} 
        \left( -\hat{\pi}^{Treat}_2 - \hat{\pi}^{Treat}_5 X^{post,{Treat}}_{i,t} \right)^2
    }
\end{aligned}
$}
\end{equation}

Similarly, we can show that $\hat{\psi_3}$ is obtained by residualizing $pre_t^{Control}$ with respect to $post_t^{Control}$ and $X^{Control}_{i,t}$. Let $\widehat{pre}_t^{Control}$ be the fitted value from the following regression:

\begin{equation}
\label{equation: DIDINTpreControl}
    pre_t^{Control}
    =
    \kappa_2^{Control} post_t^{Control}
    +
    \kappa_5^{Control} X^{Control}_{i,t}
    +
    u_{i,t}.
\end{equation}

For a control observation in the pre-intervention period, the residual is given by $1-\hat{\kappa}_5^{Control}X^{pre,Control}_{i,t}$. Similarly, for a control observation in the post-intervention period, the residual is given by $-\hat{\kappa}_2^{Control}-\hat{\kappa}_5^{Control}X^{post,Control}_{i,t}$. Therefore,

\begin{equation}
\resizebox{\textwidth}{!}{$
\label{equation: ATTDTsecond}
\begin{aligned}
    \hat{\psi_3} = & \frac{\sum_{i \in \{ t \in pre, x \in X^{Control}\}} \left( 1 - \hat{\kappa}^{Control}_5 X^{pre,{Control}}_{i,t} \right) Y^{pre,{Control}}_{i,t} + \sum_{i \in \{ t \in post, x \in X^{Control} \}} \left( -\hat{\kappa}^{Control}_2 - \hat{\kappa}^{Control}_5 X^{post,{Control}}_{i,t} \right) Y^{post,{Control}}_{i,t}}{\sum_{i \in \{ t \in pre, x \in X^{Control} \}} \left( 1 - \hat{\kappa}^{Control}_5 X^{pre,{Control}}_{i,t} \right)^2 + \sum_{i \in \{t \in post, x \in X^{Control}\}} \left( -\hat{\kappa}^{Control}_2 - \hat{\kappa}^{Control}_5 X^{post,{Control}}_{i,t} \right)^2}
\end{aligned}
$}
\end{equation}

Now, we will derive $\hat{\psi_2}$ from the regression shown in Equation \eqref{equation: DIDINT}. Let $\widehat{post}_t^{Treat}$ be the fitted values from the following regression:
\begin{equation}
\label{equation: DIDINTpre2}
    post_t^{Treat} = \rho_1^{Treat} pre_t^{Treat} + \rho_3^{Control} pre_t^{Control} + \rho_4^{Control} post_t^{Control} + \rho_5^{Treat} X^{Treat}_{i,t} + \rho_6^{Control} X^{Control}_{i,t} + v_{i,t}.
\end{equation}

Based on the proof of Lemma \eqref{lemma: independence} in Appendix \ref{ssection: independence},  $post_t^{Treat}$ is sample-orthogonal to  $pre_t^{Control}$, $post_t^{Control}$ and $X_{i,t}^{Control}$. Consequently, $\rho_3^{Control} = \rho_4^{Control} = \rho_6^{Control} = 0$. We can rewrite Equation \eqref{equation: DIDINTpre2} as:
\begin{equation}
\label{equation: DIDINTpre2simplified}
    post_t^{Treat} = \rho_1^{Treat} pre_t^{Treat} + \rho_5^{Treat} X^{Treat}_{i,t} + v_{i,t}
\end{equation}

From Equation \eqref{equation: DIDINTpre2simplified}, we again derive the residuals for each observation based on their treatment status and time period. For a treated observation in the pre-intervention period, the residual can be written as $- \hat{\rho_1}^{Treat} - \hat{\rho}^{Treat}_5 X^{Treat}_{i,0}$. Similarly, for a treated observation in the post-intervention period, the residual is $1 - \hat{\rho}^{Treat}_5 X^{Treat}_{i,1}$. According to the FWL theorem, the coefficient $\hat{\psi_2}$ will be the same as the coefficient of $(post^{Treat}_t -\widehat{post}^{Treat}_t)$ from the regression shown in Equation \eqref{equation: psi2residuals}. Here, $(post^{Treat}_t -\widehat{post}^{Treat}_t)$ are the residuals from the regression shown in Equation \eqref{equation: DIDINTpre2simplified}.
\begin{equation}
\label{equation: psi2residuals}
    Y_{i,t} = \hat{\psi_2} (post^{Treat}_t -\widehat{post}^{Treat}_t) + v'_{i,t} 
\end{equation}
From Equation \eqref{equation: psi2residuals}:
\begin{align*}
\begin{split}
\widehat{\psi}_2 =
\frac{
  \sum_i \sum_t Y_{i,t}
  (post^{Treat}_t-\widehat{post}^{Treat}_t)
}{
  \sum_i \sum_t
  (post^{Treat}_t-\widehat{post}^{Treat}_t)^2
}.
\end{split}
\end{align*}
\begin{equation}
\label{equation: psi2}
\resizebox{\textwidth}{!}{$
\begin{aligned}
    \Rightarrow \hat{\psi}_2
    ={} &
    \frac{
    \sum_{i \in \{t \in pre, x \in X^{Treat}\}}
    \left(
    -\hat{\rho}_1^{Treat}
    -\hat{\rho}^{Treat}_5 X^{pre,Treat}_{i,t}
    \right)
    Y^{pre,Treat}_{i,t}
    +
    \sum_{i \in \{t \in post, x \in X^{Treat}\}}
    \left(
    1-\hat{\rho}^{Treat}_5 X^{post,Treat}_{i,t}
    \right)
    Y^{post,Treat}_{i,t}
    }{
    \sum_{i \in \{t \in pre, x \in X^{Treat}\}}
    \left(
    -\hat{\rho}_1^{Treat}
    -\hat{\rho}^{Treat}_5 X^{pre,Treat}_{i,t}
    \right)^2
    +
    \sum_{i \in \{t \in post, x \in X^{Treat}\}}
    \left(
    1-\hat{\rho}^{Treat}_5 X^{post,Treat}_{i,t}
    \right)^2
    }.
\end{aligned}
$}
\end{equation}

Similarly, we can show that $\hat{\psi_4}$ is obtained by residualizing $post_t^{Control}$ with respect to $pre_t^{Control}$ and $X^{Control}_{i,t}$. Let $\widehat{post}_t^{Control}$ be the fitted value from the following regression:

\begin{equation}
\label{equation: DIDINTpostControl}
    post_t^{Control}
    =
    \zeta_1^{Control} pre_t^{Control}
    +
    \zeta_5^{Control} X^{Control}_{i,t}
    +
    w_{i,t}.
\end{equation}

For a control observation in the pre-intervention period, the residual is given by $-\hat{\zeta}_1^{Control}-\hat{\zeta}_5^{Control}X^{pre,Control}_{i,t}$. Similarly, for a control observation in the post-intervention period, the residual is given by $1-\hat{\zeta}_5^{Control}X^{post,Control}_{i,t}$. Therefore,

\begin{equation}
\resizebox{\textwidth}{!}{$
\label{equation: psi4}
\begin{aligned}
    \hat{\psi_4} = & \frac{\sum_{i \in \{ t \in pre, x \in X^{Control} \}} \left( - \hat{\zeta}_1^{Control} - \hat{\zeta}^{Control}_5 X^{pre,{Control}}_{i,t} \right) Y^{pre,{Control}}_{i,t} + \sum_{i \in \{ t \in post, x \in X^{Control}  \}} \left( 1 - \hat{\zeta}^{Control}_5 X^{post,{Control}}_{i,t} \right) Y^{post,{Control}}_{i,t} }{\sum_{i \in \{ t \in pre, x \in X^{Control}\}} \left( - \hat{\zeta}_1^{Control} - \hat{\zeta}^{Control}_5 X^{pre,{Control}}_{i,t} \right)^2 + \sum_{i \in \{ t \in post, x \in X^{Control}\}} \left( 1 - \hat{\zeta}^{Control}_5 X^{post,{Control}}_{i,t} \right)^2}
\end{aligned}
$}
\end{equation}

Combining Equations \eqref{equation: ATTDTfirst}, \eqref{equation: ATTDTsecond}, \eqref{equation: psi2} and \eqref{equation: psi4}, we get an expression for $\left[ (\hat{\psi_2} - \hat{\psi_1}) - (\hat{\psi_4} - \hat{\psi_3})\right]$. 

For the unpooled regression, we use the specifications in Equations \eqref{equation: undid1} and \eqref{equation: undid2}. Provided we are using the same underlying sample for the two regressions, $X^{Treat}_{i,t}$ from the treated UN-DID regression in Equation \eqref{equation: undid1} is the same as $X^{Treat}_{i,t}$ from the regression shown in Equation \eqref{equation: DIDINT}.  Using the FWL theorem, it can be shown that $\hat{\lambda}_1^{Treat}$ from the regression shown in Equation \eqref{equation: undid1} is the same as the coefficient of $(pre^{Treat}_t - \widehat{pre}^{Treat}_t)$ from the following regression:
\begin{equation}
\label{equation: undid1fwlcov}
    Y_{i,t} = \lambda^{Treat}_1 (pre^{Treat}_t -\widehat{pre}^{Treat}_t) + \nu^{Treat}_{i,t} 
\end{equation}
Here, $(pre^{Treat}_t - \widehat{pre}^{Treat}_t)$ are the residuals from the following regression:
\begin{equation}
\label{equation: undid1fwlcovresiduals2}
    pre_t^{Treat} = \Lambda_2^{Treat} post_t^{Treat} + \Lambda_3^{Treat} X^{Treat}_{i,t} + \nu'^{{Treat}}_{i,t}
\end{equation}

Using the same underlying sample, Equation \eqref{equation: undid1fwlcovresiduals2} is exactly the same as Equation \eqref{equation: DIDINTpre1simplified} from the DID-INT regression. Therefore the residuals from the two regressions will also be the same. This implies that $\hat{\lambda}_1^{Treat}$ from the unpooled regression of the treated group is equal to $\hat{\psi_1}$ from the DID-INT regression. Similarly, $\hat{\lambda}_1^{Control}$ from the unpooled regression of the untreated group is equal to $\hat{\psi_3}$ from the DID-INT regression.

Now, we will derive $\hat{\lambda}_2^{Treat}$ from the regression shown in Equation \eqref{equation: undid1}. Using the FWL theorem, it can be shown that $\hat{\lambda}_2^{Treat}$ from the regression shown in Equation \eqref{equation: undid1} is the same as the coefficient of $(post^{Treat}_t - \widehat{post}^{Treat}_t)$ from the following regression:
\begin{equation}
\label{equation: undid1fwlcov2}
    Y_{i,t} = \lambda^{Treat}_2 (post^{Treat}_t -\widehat{post}^{Treat}_t) + \nu^{{Treat}''}_{i,t} 
\end{equation}
Here, $(post^{Treat}_t - \widehat{post}^{Treat}_t)$ are the residuals from the following regression:
\begin{equation}
\label{equation: undid1fwlcovresiduals}
    post_t^{Treat} = \Theta_1^{Treat} pre_t^{Treat} + \Theta_3^{Treat} X^{Treat}_{i,t} + \nu_{i,t}^{{Treat}'''}
\end{equation}

As in the proof of $\widehat{\lambda}^{Treat}_1$, Equation \eqref{equation: undid1fwlcovresiduals} is exactly the same as Equation \eqref{equation: DIDINTpre2simplified} from the DID-INT regression provided we are using the same sample. Therefore the residuals from the two regressions will also be the same. This implies, $\hat{\lambda}_2^{Treat}$ from the unpooled regression of the treated group is equal to $\hat{\psi_2}$ from the DID-INT regression. We can similarly show, $\hat{\lambda}_2^{Control}$ from the unpooled regression of the untreated group is equal to $\hat{\psi_4}$ from the DID-INT regression. 

    \begin{equation}
        \left[ (\hat{\psi_2} - \hat{\psi_1}) - (\hat{\psi_4} - \hat{\psi_3})\right] = (\widehat{\lambda}_2^{Treat} - \widehat{\lambda}_1^{Treat}) - (\widehat{\lambda}_2^{Control} - \widehat{\lambda}_1^{Control})
    \end{equation}
\end{proof}

\section{Proof of equivalence between two versions of UN-DID}
\label{section: transformedequation}

\begin{lemma} For each silo $j \in \{Treat,Control\}$, the population regression \begin{equation} \label{equation: prop1eq1} Y^j_{i,t} = \lambda_1^j pre_t^j + \lambda_2^j post_t^j + \lambda_3^j X^j_{i,t} + \nu^j_{i,t} \end{equation} can equivalently be written as \begin{equation} \label{equation: prop1eq2} Y^j_{i,t} = \gamma_0^j + \gamma_1^j post_t^j + \gamma_2^j X^j_{i,t} + \nu^j_{i,t}, \end{equation}
\end{lemma}

\begin{proof}
    We know that $pre_t^j = (1 - post_t^j)$. Substituting this into Equation \eqref{equation: prop1eq1} gives
\begin{equation*}
Y^j_{i,t} = \lambda_1^j (1 - post_t^j) + \lambda_2^j post_t^j + \lambda_3^j X^j_{i,t} + \nu^j_{i,t}
\end{equation*}
\begin{equation}
\label{equation: prop1eq3}
    \Rightarrow {Y^j_{i,t}} = \lambda_1^j + (\lambda_2^j - \lambda_1^j) \text{post}_t^j + \lambda_3^j X^j_{i,t} + \nu^j_{i,t}.
\end{equation}
Comparing equations \eqref{equation: prop1eq2} and \eqref{equation: prop1eq3}, we can see that $\lambda_1^j = \gamma_0^j$, $(\lambda_2^j - \lambda_1^j) = \gamma_1^j$ and $\lambda_3^j = \gamma_2^j$. This proof also implies that  $(\widehat{\lambda}_2^{Treat} - \widehat{\lambda}_1^{Treat}) - (\widehat{\lambda}_2^{Control} - \widehat{\lambda}_1^{Control}) = \hat{\gamma}_1^{Treat} - \hat{\gamma}_1^{Control} $. 
\end{proof}

\section{Proof of Lemma 4}
\label{ssection: independence}

\begin{lemma}[Disjoint Support and Sample Orthogonality]
\label{lemma: independence}
            \begin{equation}
                \sum_{i,s,t} pre^{Treat}_{t}pre^{Control}_{t} = 0, \sum_{i,s,t} pre^{Treat}_{t}post^{Control}_{t} = 0, \sum_{i,s,t} pre^{Treat}_{t} X_{i,s,t}^{Control}=0  
            \end{equation}
\end{lemma}

\begin{proof}[Proof: Lemma \eqref{lemma: independence}.]
    \begin{equation}
        \label{equation: lemma8first}
    pre_t^{Treat} = 1\{D_s=1,P_t=0\}.
    \end{equation}
The above follows from the definition of $pre_t^{Treat}$, which is a dummy variable that takes on a value of 1 if the observation is from the treated group in the pre-intervention period. Similarly, 
    \begin{equation}
        \label{equation: lemma8first2}
    pre_t^{Control} = 1\{D_s=0,P_t=0\}.
    \end{equation}
$pre_t^{Control}$ is a dummy variable which has a value of one if the observation is from the control group in the pre-intervention period. For a given observation, 
    \begin{equation}
    \label{equation: lemma8second}
        pre_t^{Treat}pre_t^{Control} = 1\{D_s=1,P_t=0\}1\{D_s=0,P_t=0\} = 0.
    \end{equation}

Taking a sum over all observations, we observe that the inner product between $pre^{Treat}_{t}$ and $pre^{Control}_{t}$ is zero, which implies that $pre^{Treat}_{t}$ and $pre^{Control}_{t}$ are sample-orthogonal.
\begin{equation}
    \sum_{i,s,t} pre^{Treat}_{i,s,t}pre^{Control}_{i,s,t} = 0.
\end{equation}
Intuitively, the result follows because the two indicators have disjoint nonzero support, meaning that an observation cannot belong to the pre-treatment period in both the treated and control groups. As a result, if one indicator is equal to one, the other must be zero. Using the same intuition, we can show that the inner product of $pre_t^{Treat}$ and $post_t^{Control}$ as well as $pre_t^{Treat}$ and $X_{i,t}^{Control}$ is 0. Since every treated silo regressor is sample-orthogonal to every control silo regressor, the treated and control regressor blocks are block orthogonal. Suppose $\mathbb{Z}^{Treat}$ is a matrix containing the regressors from the treated silo, $\mathbb{Z}^{Control}$ is a matrix containing the regressors from the control silo, and $\mathbb{Z}$ is the full design matrix containing the regressors from both the treated and the control silos

\begin{equation}
    \mathbb{Z}^{Treat}
    =
    \left[
        pre^{Treat},
        \; post^{Treat},
        \; X^{Treat}
    \right],
\end{equation}
\begin{equation}
    \mathbb{Z}^{Control}
    =
    \left[
        pre^{Control},
        \; post^{Control},
        \; X^{Control}
    \right],
\end{equation}
\begin{equation}
    \mathbb{Z}
    =
    \left[
        \mathbb{Z}^{Treat},
        \; \mathbb{Z}^{Control}
    \right].
\end{equation}
Block orthogonality implies:
\[
\mathbb{Z}'\mathbb{Z}
=
\begin{bmatrix}
\left(\mathbb{Z}^{Treat}\right)'\mathbb{Z}^{Treat}
&
0
\\[4pt]
0
&
\left(\mathbb{Z}^{Control}\right)'\mathbb{Z}^{Control}
\end{bmatrix}.
\]

Therefore, $\widehat{\pi}^{Control}_3 = \widehat{\pi}^{Control}_4 = \widehat{\pi}^{Control}_6 = 0$ in Equation \eqref{equation: DIDINTpre1}.
\end{proof}

    \section{Additional Monte Carlo Results}\label{sec:moremc}

\subsection{Results: No Covariates}\label{sec:nocovar}

\begin{figure}[ht!]
    \centering
    \includegraphics[width=0.70\textwidth]{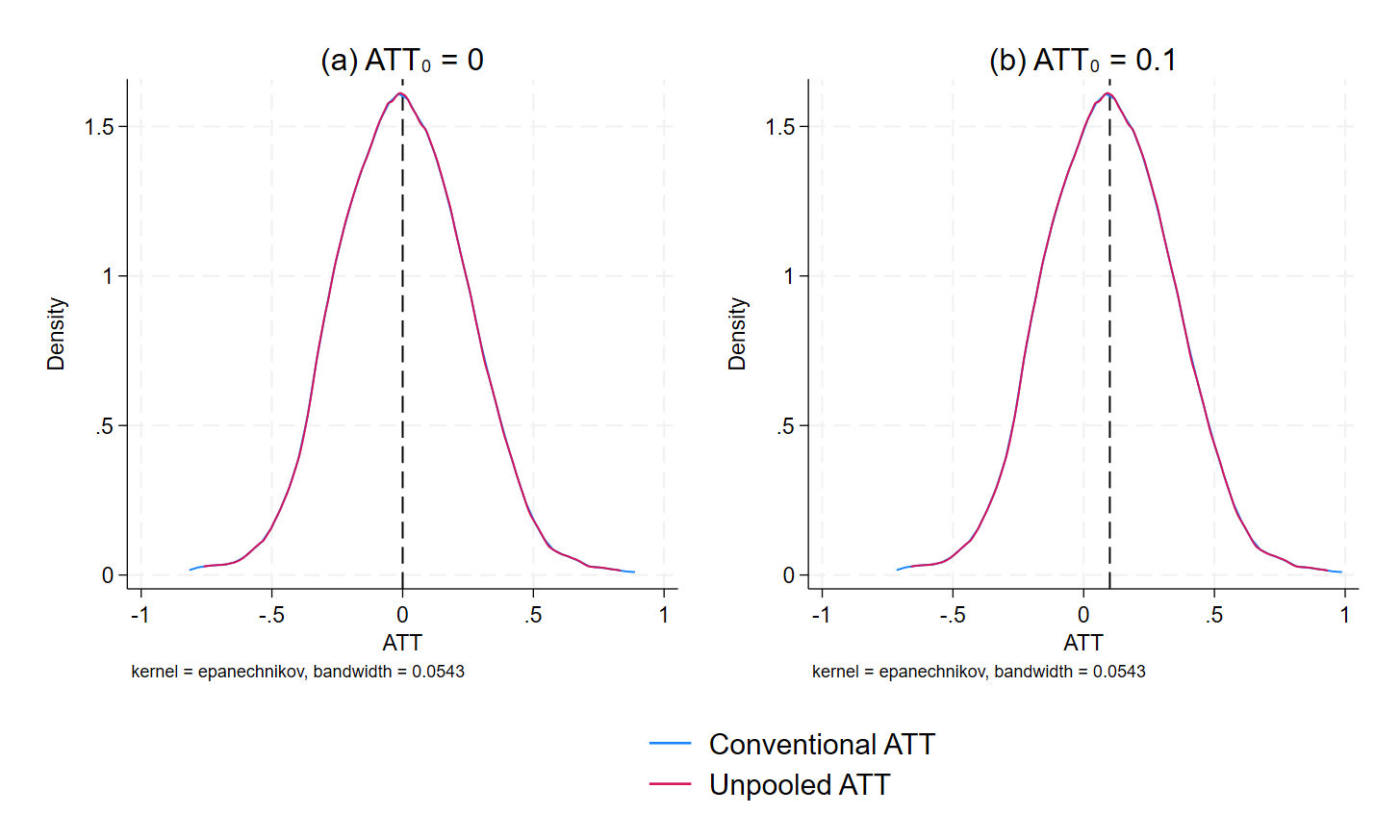}
    \caption{\centering No Covariates with Unequal Sample Sizes: Kernel Density}
    \label{fig:KDensity2}
\end{figure}

\FloatBarrier

In Theorem \ref{theorem: theorem2}, we have shown that the conventional and the UN-DID regressions are algebraically equivalent in finite samples and unbiased. To assess unbiasedness, we generate kernel density plots for both conventional DID and UN-DID estimates of the ATT on a shared axis. The kernel densities for unequal sample sizes are shown in Figure \ref{fig:KDensity2}. The DGPs are generated such that Assumptions \eqref{assumption: binary}, \eqref{assumption: na} and \eqref{assumption: pt} hold. Since we can run both the UN-DID and the conventional method on poolable data, Data Assumption \eqref{assumption:pooled} is not necessary. In Panel (a), the kernel density estimates for Case 1 are presented, and in Panel (b), the kernel density estimates for Case 2 are presented. Notably, for Case 1, both kernels are centered around 0, aligning with the true value of the ATT. Similarly, for Case 2, both kernels are centered around 0.1, which is the true value of the ATT for Case 2. The blue line representing the conventional ATT lies exactly beneath the red line representing the UN-DID ATT (except at the tails) and is therefore not visible. This demonstrates that the ATTs of the two methods are exactly numerically equivalent. This finding demonstrates that both conventional and UN-DID estimation methods exhibit unbiasedness, as evidenced by the distributions of estimates being centered around the true values of the ATT.

\FloatBarrier

\subsection{Results: Time-Invariant Covariate}
\label{subsec:timeinvarmc}
In Theorem \ref{theorem: theorem11}, we have shown that both the conventional DID and the UN-DID methods  can recover numerically equivalent values of the ATT when implemented on a poolable dataset and that both are unbiased. The DGPs for this section are generated such that Assumptions \eqref{assumption: binary}, \eqref{assumption: spt}, \eqref{assumption: na} and \eqref{assumption: stateinvariantccc} hold. Similar to the preceding sub-section, we assess the unbiasedness of the UN-DID estimator with unequal sample sizes and time-invariant covariates using kernel densities, illustrated in Figure \ref{fig:KDensity7}. The kernel densities show that the distribution of ATTs is centered on the true ATT value, implying unbiasedness with time-invariant covariates. As in Figure \ref{fig:KDensity2}, the blue line is not visible as it is perfectly underneath the red line (except at the tails), demonstrating numerical equivalence.

\begin{figure}[ht!]
    \centering
    \includegraphics[width=0.70\textwidth]{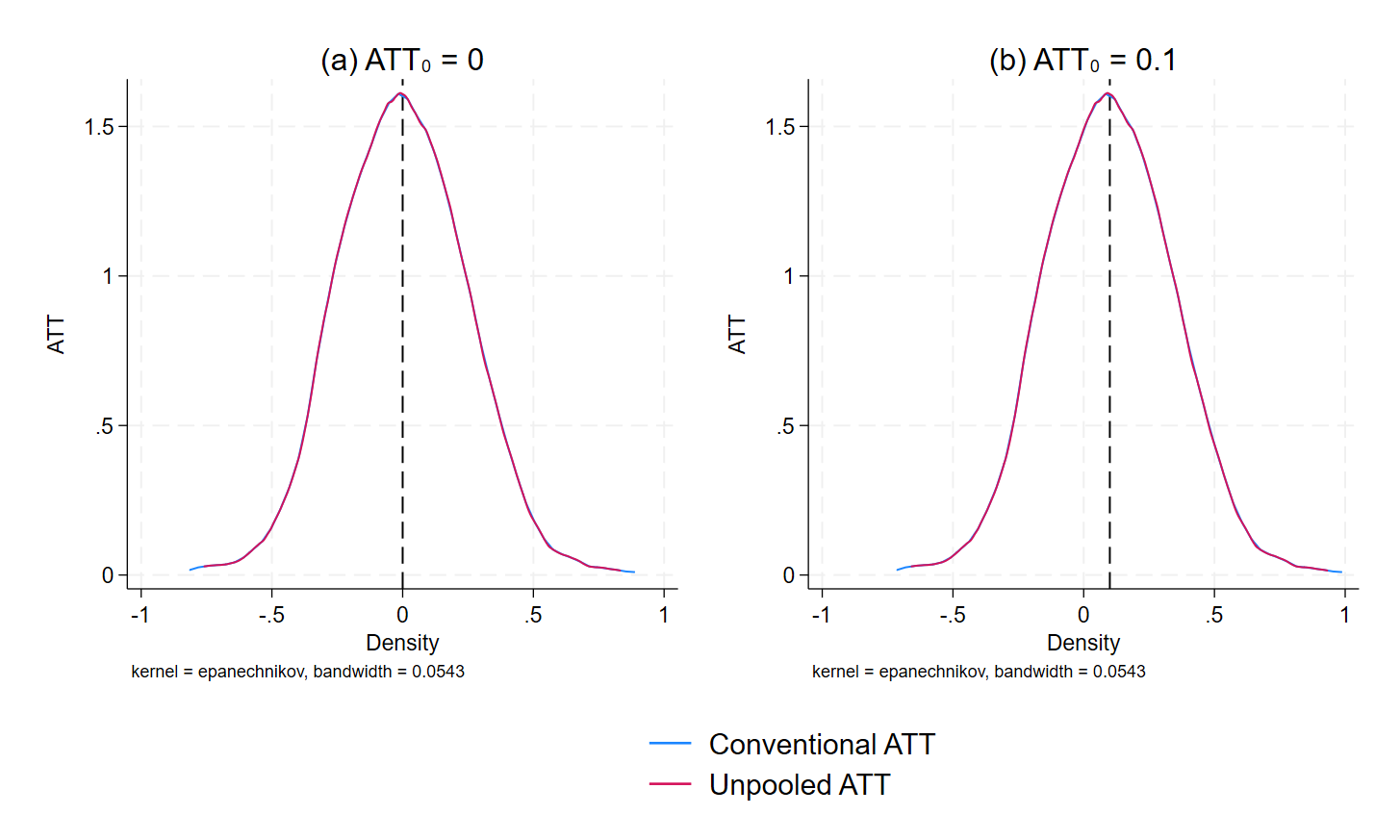}
    \caption{\centering Time-Invariant Covariates with Unequal Sample Sizes: Kernel Density}
    \label{fig:KDensity7}
\end{figure}

\section{Monte Carlo Simulations with Differing Covariates} \label{sec:diff_covar}

The main Monte Carlo experiments assume that covariates are defined in the same way across silos. In practice, this need not be true. Two jurisdictions may measure the same underlying characteristic using different categories or even different variables. This creates a problem for pooled covariate adjustment because the measures must first be expressed in a common form. It is also distinct from the Common Causal Covariates (CCC) restriction discussed in Section \ref{sec3:ccc}. Harmonization concerns whether different observed measures can be converted into a commonly measured covariate, whereas CCC concerns whether that covariate has the same relationship with untreated outcomes across jurisdictions. We use simulations calibrated to Ontario and Quebec to separate these two issues.

We compare three estimators. \emph{Oracle TWFE} pools the provincial data and estimates a conventional two-way fixed-effects regression with province and year effects, a Quebec-by-post-treatment indicator, and common covariate coefficients. It is ``oracle'' because it observes the desired common covariate for every observation, even when that variable would not actually be available in one of the silos. \emph{Harmonized TWFE} estimates exactly the same pooled regression, but when a province's data does not include the desired covariate, it constructs it using an empirical crosswalk between the measures observed in the two provinces. The difference between Oracle and Harmonized TWFE therefore isolates the consequences of harmonization while holding the pooled outcome specification fixed. \emph{UN-DID} does not require a common covariate representation. Each province's regression adjusts for the variables available in its own data before the resulting changes over time are compared across provinces, allowing the first-stage covariate relationships to differ across silos.

All simulations use the Statistics Canada 2001 Census Individuals Public Use Microdata File (PUMF) to approximate realistic Ontario and Quebec populations \citep{statisticscanada2005pumf}. We consider two applications. In both, Quebec is treated beginning in 2006, Ontario remains untreated, and treatment does not enter the outcome-generating process. The true aggregate $ATT$ is therefore zero.

The first application considers age adjustment for a binary heart-attack outcome. We draw panels of 10{,}000 men per province who are aged 55--75 in 2001, age them mechanically through 2010, and generate untreated outcomes from a stylized age-specific risk schedule. The simulations vary three features of the data. First, Ontario and Quebec either draw from a common age distribution or retain their province-specific Census distributions. Second, age is either recorded identically across provinces or using incompatible categories. In the latter design, Ontario records age using conventional five-year groups, 55--59, 60--64, and so on, while Quebec uses shifted, non-nested groups: 55--56, 57--61, 62--66, 67--71, 72--76, 77--81, and 82--84. Finally, CCC either holds, so the untreated relationship between age and heart-attack risk is the same in the two provinces, or fails, so Quebec's untreated risk at age $a$ corresponds to Ontario's risk at age $a+3$.

When the age categories differ, pooled estimation requires the Quebec categories to be translated into the Ontario categories. Harmonized TWFE uses the Census to estimate the distribution of actual ages within each Quebec category and replaces the unavailable Ontario-category indicators with their conditional expectations. For example, if a Quebec category overlaps two Ontario categories, an observation receives fractional values reflecting the probability of belonging to each. Oracle TWFE instead observes the individual's actual Ontario-style age category. UN-DID requires no such translation: Ontario and Quebec use their own age-category indicators in their respective first-stage regressions.

The second application considers education adjustment in repeated cross-sections, with log total income as the outcome. We draw 10{,}000 adults per province and year from 2001 through 2010. The Census provides two related but distinct measures of schooling. \texttt{TOTSCHP} measures total years of schooling, while \texttt{HLOSP} records the highest level of schooling. We combine the latter into five broad categories: less than high-school graduation, high-school graduation, trades, college, or university below a bachelor's degree, bachelor's or first professional degree, and post-bachelor qualifications. Because the two measures are related but not interchangeable, we draw Census donor records jointly rather than generating the measures independently. The simulated samples therefore preserve the empirical relationship between the two education measures and demographic characteristics such as age, sex, and marital status.

Untreated log income is generated from a Mincer regression using education indicators, age, sex, and marital status. The education component is defined using either total schooling or the five-category highest-schooling measure. CCC holds when the same covariate-outcome relationship generates untreated income in both provinces and fails when separately calibrated Ontario and Quebec relationships are used. We also vary whether the provinces observe the same education measure. In the incompatible-measure design, Ontario observes highest schooling while Quebec observes total schooling. If total schooling determines untreated income, for example, Harmonized TWFE reconstructs Ontario's unavailable total-schooling indicators using the empirical distribution
\[
P(\texttt{TOTSCHP}=k
  \mid \texttt{HLOSP5},\text{Ontario},\text{age group}),
\]
whereas Oracle TWFE observes each observation's actual total-schooling category. UN-DID instead adjusts Ontario outcomes using highest schooling and Quebec outcomes using total schooling.

\begin{figure}[htbp]
\centering
\includegraphics[width=0.98\textwidth]{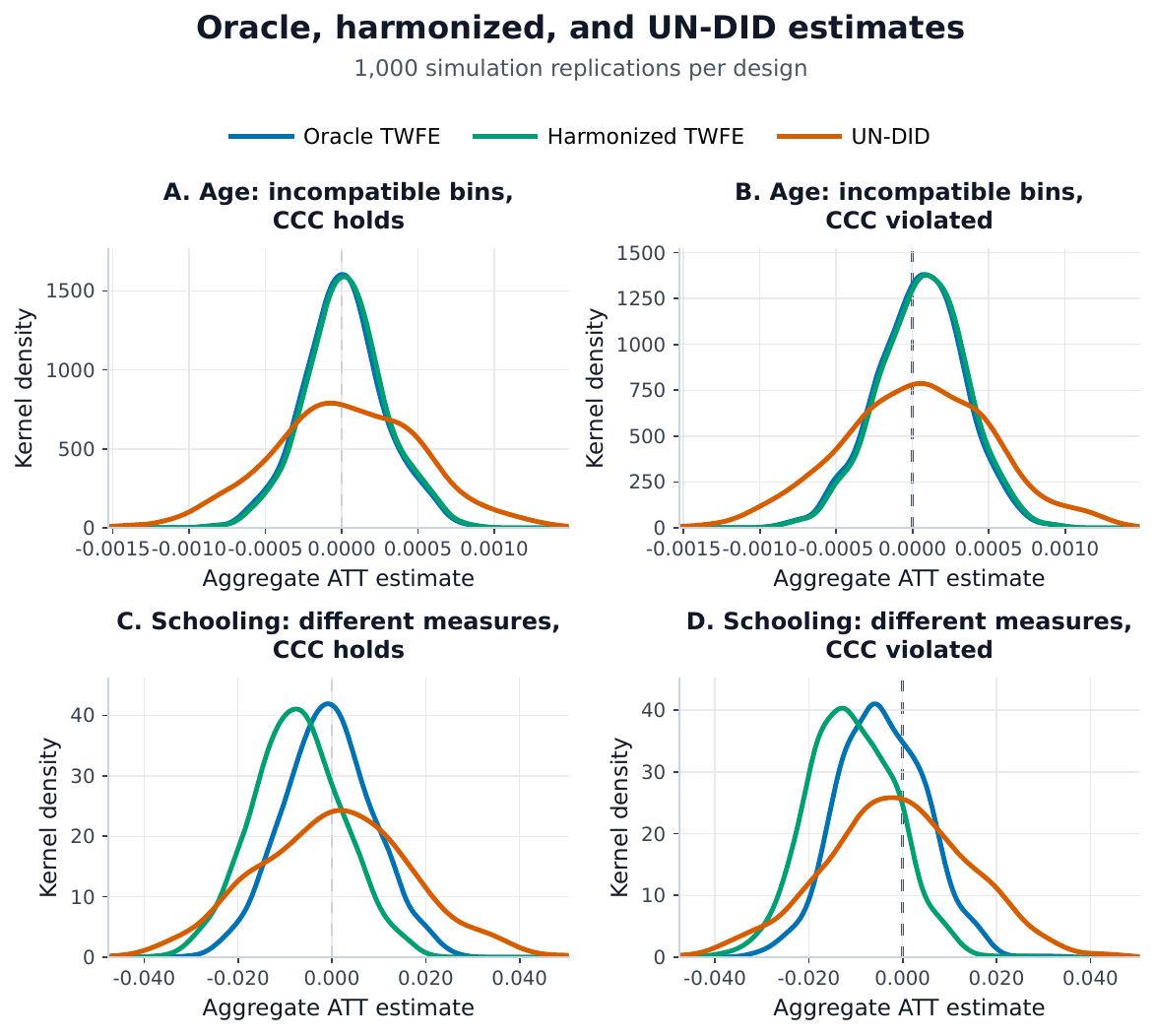}
\caption{\textbf{Harmonization and heterogeneous covariate effects.}
Each panel shows the distribution of 1{,}000 estimates of the aggregate $ATT$, whose true value is zero. Oracle TWFE observes the target covariate directly but imposes common covariate coefficients across provinces. Harmonized TWFE estimates the same pooled model after constructing an unavailable target covariate using an empirical Census crosswalk. UN-DID instead adjusts separately using the covariates observed within each province. Panels A and B use incompatible age categories and province-specific age distributions; Panels C and D use different education measures across provinces. CCC holds in Panels A and C and fails in Panels B and D.}
\label{fig:oncov_oracle_density}
\end{figure}

Figure \ref{fig:oncov_oracle_density} shows that the consequences of harmonization depend importantly on what is being measured. In the age application, translating the non-nested categories has almost no effect. When CCC holds, Oracle and Harmonized TWFE have biases of 0.000010 and 0.000028, respectively, and their sampling distributions are nearly indistinguishable in Panel A. When CCC fails, both move slightly to the right of zero, with biases of 0.000048 for Oracle TWFE and 0.000067 for Harmonized TWFE. The difference between them remains very small. UN-DID is also centered close to zero but has a wider sampling distribution. Its Monte Carlo standard deviation is approximately 70 percent larger than that of Oracle TWFE. Thus, in this calibration, different age categories create little harmonization cost, while the greater flexibility of UN-DID carries a noticeable precision cost.  Panel B is similar, except that violating CCC shifts both pooled estimators slightly away from zero, while UN-DID remains centered close to zero. UN-DID is again more variable.

The education application produces a much clearer distinction. When CCC holds, Oracle TWFE is nearly centered on zero, with bias $-0.0008$, whereas Harmonized TWFE has bias $-0.0074$. UN-DID is also centered close to zero, with bias $0.0005$, although its sampling distribution is noticeably wider than those of the pooled estimators. The two pooled estimators use the same outcome specification, so the separation between them in Panel C is attributable to the need to infer an individual’s total-schooling category from the education measure actually observed in Ontario. Although the crosswalk preserves the conditional distribution of schooling categories, it cannot recover the unavailable individual-level category perfectly.

Panel D adds a violation of CCC. Oracle TWFE now has bias $-0.0044$ despite observing the correct education category for every individual. This isolates the second problem: perfect harmonization does not make the pooled common-slope specification correct when the untreated covariate-outcome relationship differs across provinces. Harmonized TWFE combines this common-slope misspecification with the loss of information from harmonization and has bias $-0.0108$. UN-DID remains substantially closer to zero, with bias $-0.0015$, because it neither reconstructs a common education measure nor imposes a common provincial covariate relationship. As in the age design, this additional flexibility reduces precision: the Monte Carlo standard deviation of UN-DID is roughly 70 percent larger than that of Oracle TWFE.

These simulations highlight a practical feature of unpoolable data that is separate from the inability to place individual records on a common server. Even when different silos contain substantively comparable covariates, forcing those variables into a common pooled specification may require a harmonization step that discards information. Moreover, perfect harmonization does not resolve misspecification when the relationship between covariates and untreated outcomes differs across jurisdictions. UN-DID permits each silo to adjust for the representation available locally and allows those relationships to differ across silos. The simulations also make clear that this flexibility is not free: when harmonization and common-slope restrictions are innocuous, pooled estimation can be substantially more precise.

\clearpage

\section{Different CCC Violations for Merit Example} \label{sec:merit_ccc}

	\begin{table}[htbp]
		\centering
		\caption{DID-INT ATT and Jackknife SE Estimates}
		\label{tab:didintjl_results_rounded}
		\begin{tabular}{lcc}
			\hline
			\textbf{CCC variation} & \textbf{Aggregate ATT} & \textbf{Jackknife SE} \\
			\hline
			
			Homogeneous                   & 0.0450                 & 0.0092                \\
			Time                  & 0.0412                 & 0.0101                \\
			Region                 & 0.0458                 & 0.0084                \\
			Two One-Way                   & 0.0420                 & 0.0091                \\
			Two Way                   & 0.0511                 & 0.0209                \\
			\hline \hline
		\end{tabular}
	\vskip 6pt
	{\footnotesize Comparison of ATT estimates using cohort aggregation with DID-INT}
	\end{table}

    \clearpage  

\end{document}

%% file: cps_fig.tex
\begin{tikzpicture}
\begin{groupplot}[
    group style={
        group size=3 by 2,
        horizontal sep=1.05cm,
        vertical sep=1.15cm
    },
    width=0.31\textwidth,
    height=0.22\textheight,
    ymin=0.035, ymax=0.185,
    xmin=0.8, xmax=19.2,
    ylabel={Rejection frequency},
    xlabel={$J=L$},
    ytick={0.05,0.10,0.15},
    yticklabels={0.05,0.10,0.15},
    scaled y ticks=false,
    xticklabel style={font=\scriptsize},
    yticklabel style={font=\scriptsize},
    label style={font=\small},
    title style={font=\small},
    tick label style={font=\scriptsize},
    legend style={
        font=\small,
        draw=none,
        fill=white,
        fill opacity=0.9,
        text opacity=1,
        at={(0.04,0.96)},
        anchor=north west
    },
    grid=major,
    grid style={draw=gray!20},
    major grid style={draw=gray!20},
    axis line style={black},
    tick style={black},
    every axis plot/.append style={line width=1.1pt},
]

\nextgroupplot[
    title={(a) $S=8$},
    xtick={1,3},
    xlabel={}
]
\addplot[
    color=blue!70!black,
    mark=*,
    mark options={scale=0.75, fill=blue!70!black}
] coordinates {
    (1,0.11976)
    (2,0.07848)
    (3,0.09042)
};
\addlegendentry{Jackknife}

\addplot[
    color=orange!85!black,
    mark=square*,
    mark options={scale=0.65, fill=orange!85!black}
] coordinates {
    (1,0.04742)
    (2,0.05068)
    (3,0.04808)
};
\addlegendentry{RI}

\addplot[
    color=black!70,
    dashed,
    line width=0.9pt
] coordinates {(0.8,0.05) (19.2,0.05)};

\nextgroupplot[
    title={(b) $S=16$},
    xtick={1,3,5,7},
    ylabel={},
    xlabel={}
]
\addplot[
    color=blue!70!black,
    mark=*,
    mark options={scale=0.75, fill=blue!70!black}
] coordinates {
    (1,0.14194)
    (2,0.08972)
    (3,0.07260)
    (4,0.06450)
    (5,0.06038)
    (6,0.06746)
    (7,0.10690)
};

\addplot[
    color=orange!85!black,
    mark=square*,
    mark options={scale=0.65, fill=orange!85!black}
] coordinates {
    (1,0.04760)
    (2,0.04966)
    (3,0.04970)
    (4,0.05040)
    (5,0.04968)
    (6,0.04984)
    (7,0.05012)
};

\addplot[
    color=black!70,
    dashed,
    line width=0.9pt
] coordinates {(0.8,0.05) (19.2,0.05)};

\nextgroupplot[
    title={(c) $S=24$},
    xtick={1,3,5,7,9,11},
    ylabel={},
    xlabel={}
]
\addplot[
    color=blue!70!black,
    mark=*,
    mark options={scale=0.75, fill=blue!70!black}
] coordinates {
    (1,0.15736)
    (2,0.09682)
    (3,0.07968)
    (4,0.06990)
    (5,0.06318)
    (6,0.05886)
    (7,0.05924)
    (8,0.06006)
    (9,0.06446)
    (10,0.07474)
    (11,0.11776)
};

\addplot[
    color=orange!85!black,
    mark=square*,
    mark options={scale=0.65, fill=orange!85!black}
] coordinates {
    (1,0.04936)
    (2,0.05054)
    (3,0.05034)
    (4,0.04998)
    (5,0.05072)
    (6,0.04932)
    (7,0.05092)
    (8,0.05004)
    (9,0.05182)
    (10,0.04936)
    (11,0.04766)
};

\addplot[
    color=black!70,
    dashed,
    line width=0.9pt
] coordinates {(0.8,0.05) (19.2,0.05)};

\nextgroupplot[
    title={(d) $S=32$},
    xtick={1,3,5,7,9,11,13,15}
]
\addplot[
    color=blue!70!black,
    mark=*,
    mark options={scale=0.75, fill=blue!70!black}
] coordinates {
    (1,0.16396)
    (2,0.10288)
    (3,0.08426)
    (4,0.07392)
    (5,0.06944)
    (6,0.06386)
    (7,0.05978)
    (8,0.05918)
    (9,0.05712)
    (10,0.05722)
    (11,0.05812)
    (12,0.06248)
    (13,0.06868)
    (14,0.07726)
    (15,0.12684)
};

\addplot[
    color=orange!85!black,
    mark=square*,
    mark options={scale=0.65, fill=orange!85!black}
] coordinates {
    (1,0.04978)
    (2,0.04930)
    (3,0.04926)
    (4,0.05104)
    (5,0.05232)
    (6,0.04986)
    (7,0.04852)
    (8,0.05096)
    (9,0.04986)
    (10,0.04986)
    (11,0.04846)
    (12,0.05048)
    (13,0.05134)
    (14,0.05040)
    (15,0.04908)
};

\addplot[
    color=black!70,
    dashed,
    line width=0.9pt
] coordinates {(0.8,0.05) (19.2,0.05)};

\nextgroupplot[
    title={(e) $S=36$},
    xtick={1,3,5,7,9,11,13,15,17},
    ylabel={}
]
\addplot[
    color=blue!70!black,
    mark=*,
    mark options={scale=0.75, fill=blue!70!black}
] coordinates {
    (1,0.16268)
    (2,0.10404)
    (3,0.08522)
    (4,0.07618)
    (5,0.06570)
    (6,0.06234)
    (7,0.06160)
    (8,0.05900)
    (9,0.05852)
    (10,0.05510)
    (11,0.05848)
    (12,0.05610)
    (13,0.05846)
    (14,0.06292)
    (15,0.06838)
    (16,0.07970)
    (17,0.13054)
};

\addplot[
    color=orange!85!black,
    mark=square*,
    mark options={scale=0.65, fill=orange!85!black}
] coordinates {
    (1,0.04956)
    (2,0.05028)
    (3,0.04976)
    (4,0.04940)
    (5,0.04786)
    (6,0.04954)
    (7,0.05114)
    (8,0.05024)
    (9,0.05064)
    (10,0.04852)
    (11,0.05088)
    (12,0.04988)
    (13,0.04980)
    (14,0.04984)
    (15,0.05034)
    (16,0.05034)
    (17,0.05008)
};

\addplot[
    color=black!70,
    dashed,
    line width=0.9pt
] coordinates {(0.8,0.05) (19.2,0.05)};

\nextgroupplot[
    title={(f) $S=40$},
    xtick={1,3,5,7,9,11,13,15,17,19},
    ylabel={}
]
\addplot[
    color=blue!70!black,
    mark=*,
    mark options={scale=0.75, fill=blue!70!black}
] coordinates {
    (1,0.16804)
    (2,0.10696)
    (3,0.08818)
    (4,0.07674)
    (5,0.06992)
    (6,0.06794)
    (7,0.06304)
    (8,0.05820)
    (9,0.05842)
    (10,0.05798)
    (11,0.05770)
    (12,0.05526)
    (13,0.05674)
    (14,0.05776)
    (15,0.05902)
    (16,0.06360)
    (17,0.06890)
    (18,0.08254)
    (19,0.13386)
};

\addplot[
    color=orange!85!black,
    mark=square*,
    mark options={scale=0.65, fill=orange!85!black}
] coordinates {
    (1,0.05028)
    (2,0.04986)
    (3,0.05024)
    (4,0.04840)
    (5,0.04942)
    (6,0.05204)
    (7,0.05162)
    (8,0.05018)
    (9,0.05082)
    (10,0.04960)
    (11,0.05136)
    (12,0.04980)
    (13,0.05006)
    (14,0.04968)
    (15,0.05134)
    (16,0.05032)
    (17,0.04994)
    (18,0.04966)
    (19,0.04874)
};

\addplot[
    color=black!70,
    dashed,
    line width=0.9pt
] coordinates {(0.8,0.05) (19.2,0.05)};

\end{groupplot}
\end{tikzpicture}